\documentclass[11pt]{article}

\usepackage{amsmath} % AMS Math Package
\usepackage{amsthm} % Theorem Formatting
\usepackage{amssymb}	% Math symbols such as \mathbb
\usepackage{graphicx} % Allows for eps images
\usepackage{multicol} % Allows for multiple columns
\usepackage{multirow}
\usepackage{color}
\usepackage[dvips,letterpaper,margin=1in,bottom=1in]{geometry}

\usepackage[utf8]{inputenc}
\usepackage[english]{babel}

\usepackage{diagbox}
\usepackage{mathtools}
\usepackage[flushleft]{threeparttable}
\usepackage{tcolorbox}

\usepackage{setspace}

\usepackage{tabularray}
\UseTblrLibrary{booktabs}

\newtheorem{theorem}{Theorem}[section]

\newtheorem{lemma}[theorem]{Lemma}
\newtheorem{corollary}[theorem]{Corollary}
\newtheorem{proposition}[theorem]{Proposition}
\newtheorem{definition}{Definition}[section]

\newtheorem{remark}{Remark}[section]

\newcommand{\abs}[1]{\left| #1 \right|} % for absolute value
\newcommand{\Abs}[1]{\left\lVert #1 \right\rVert} % for Absolute value
\newcommand{\ket}[1]{\left| #1 \right>} % for Dirac bras
\newcommand{\bra}[1]{\left< #1 \right|} % for Dirac kets
\newcommand{\braket}[2]{\left< #1 \vphantom{#2} \middle| #2 \vphantom{#1} \right>} % for Dirac brackets

\newcommand{\ceil}[1] {\left\lceil #1 \right\rceil}
\newcommand{\floor}[1] {\left\lfloor #1 \right\rfloor}

\newcommand{\tr} {\operatorname{tr}}
\newcommand{\poly} {\operatorname{poly}}
\newcommand{\polylog} {\operatorname{polylog}}
\newcommand{\rank} {\operatorname{rank}}
\newcommand{\supp} {\operatorname{supp}}
\newcommand{\spanspace} {\operatorname{span}}

\newcommand{\BQP} {\mathsf{BQP}}
\newcommand{\QSZK} {\mathsf{QSZK}}

\usepackage{latexsym}
\usepackage{CJK}

\usepackage{enumerate}

\usepackage{algorithm}
\usepackage{algpseudocode}
\usepackage{stmaryrd}
\usepackage{booktabs}

\usepackage{hyperref}
\newcommand{\footremember}[2]{%
    \footnote{#2}
    \newcounter{#1}
    \setcounter{#1}{\value{footnote}}%
}

\usepackage{cite}
\usepackage{bbm}

\usepackage{tablefootnote}
\usepackage{scrextend}

\begin{document}
	%\begin{CJK*}{GBK}{song}

    \title{New Quantum Algorithms for Computing \\ Quantum Entropies and Distances}
        \author{
            Qisheng Wang \footremember{1}{Qisheng Wang is with the Graduate School of Mathematics, Nagoya University, Japan. Part of the work was done when the author was at the Department of Computer Science and Technology, Tsinghua University, Beijing, China. E-mail: \url{QishengWang1994@gmail.com}.}
            \and Ji Guan \footremember{2}{Ji Guan is with the Key Laboratory of System Software (Chinese Academy of Sciences) and State Key Laboratory of Computer Science, Institute of Software, Chinese Academy of Sciences, China. E-mail: \url{GuanJi1992@gmail.com}.}
            \and Junyi Liu \footremember{3}{Junyi Liu is with the Joint Center for Quantum Information and Computer Science, University of Maryland. Part of the work was done when the author was at the State Key Laboratory of Computer Science, Institute of Software, Chinese Academy of Sciences, China and the University of Chinese Academy of Sciences, China. E-mail: \url{JunyiLiu1994@gmail.com}.}
            \and Zhicheng Zhang \footremember{4}{Zhicheng Zhang is with Centre for Quantum Software and Information, University of Technology Sydney, Australia. E-mail: \url{iszczhang@gmail.com}.}
            \and Mingsheng Ying \footremember{5}
            {Mingsheng Ying is with the Centre for Quantum Software and Information, University of Technology Sydney, Australia. E-mail: \url{Mingsheng.Ying@uts.edu.au}.}
        }
        \date{}
        \maketitle

    \begin{abstract}
    We propose a series of quantum algorithms for computing a wide range of quantum entropies and distances, including the von Neumann entropy, quantum R\'{e}nyi entropy, trace distance, and fidelity. The proposed algorithms significantly outperform the prior best (and even quantum) ones in the low-rank case, some of which achieve exponential speedups. 
    In particular, for $N$-dimensional quantum states of rank $r$, our proposed quantum algorithms for computing the von Neumann entropy, trace distance and fidelity within additive error $\varepsilon$ have time complexity of $\tilde O(r/\varepsilon^2)$, $\tilde O(r^5/\varepsilon^6)$ and $\tilde O(r^{6.5}/\varepsilon^{7.5})$, respectively. 
    By contrast, prior quantum algorithms for the von Neumann entropy and trace distance usually have time complexity $\Omega(N)$,
    and the prior best one for fidelity has time complexity $\tilde O(r^{12.5}/\varepsilon^{13.5})$.
    
    The key idea of our quantum algorithms is to extend block-encoding from unitary operators in previous work to  quantum states (i.e., density operators). 
    It is realized by developing several convenient techniques to manipulate quantum states and extract information from them. 
    The advantage of our techniques over the existing methods is that no restrictions on density operators are required; in sharp contrast, the previous methods usually require a lower bound on the minimal non-zero eigenvalue of density operators. 
    \end{abstract}


    \newpage

    \tableofcontents
    \newpage

    \section{Introduction}
    
    Quantum entropies and distances are basic concepts \cite{NC10} in quantum physics and quantum information. Quantum entropies  characterize the randomness of a quantum system, while quantum distances  measure the closeness of quantum systems. 
    It is essential to compute their values in many important applications, from the estimation of the capacity of quantum communication channels and verification of the outcomes of quantum computation to the characterization of quantum physical  systems  (see, e.g., \cite{OP04,ZLW+21,Kuz21}). 
    Several kinds of quantum algorithms for computing quantum entropies and distances have been proposed under different computational resources, e.g., quantum algorithms with access to copies of quantum states \cite{BOW19,AISW19,WZW22}, quantum algorithms with purified quantum query access \cite{GL20,GHS21,SH21,WZC+22}, and variational quantum algorithms \cite{CPCC20, CSZW22, TV21}. 
    
    A main consideration of those quantum algorithms with copy access for computing quantum entropies and distances is the number of copies of quantum states used in the algorithms. 
    This type of input model is known as the ``quantum sample access'' model, where identical copies of quantum states are directly given.
    For example, a method of testing the closeness of $N$-dimensional mixed quantum states was provided in \cite{BOW19} with respect to trace distance and fidelity using $O(N/\varepsilon^2)$ and $O(N/\varepsilon)$ copies, respectively, based on quantum spectrum testing \cite{OW15} and efficient quantum tomography \cite{OW16,OW17,HHJ+17}.
    A method of computing the von Neumann and quantum R\'{e}nyi entropies of an $N$-dimensional quantum state was introduced in \cite{AISW19} using $O(N^2/\varepsilon^2)$ and $O(N^{2/\alpha}/\varepsilon^{2/\alpha})$ copies, respectively. 
    Recently, a new method of computing entropies was proposed in \cite{WZW22}, especially computing the von Neumann entropy uses $\tilde O(\kappa^2/\varepsilon^5)$ copies,\footnote{$\tilde O(\cdot)$ suppresses polylogarithmic factors.} where $\kappa > 0$ is given such that $\Pi / \kappa \leq \rho \leq I$ for some projector $\Pi$. A distributed quantum algorithm for computing $\tr(\rho\sigma)$, i.e., the fidelity of pure quantum states, was proposed in \cite{ALL22} using $O(\max\{\sqrt{N}/\varepsilon,1/\varepsilon^2\})$ copies.

    Another class of quantum algorithms for computing quantum entropies and distances utilizes the conventional ``purified quantum query access'' model, where mixed quantum states are given by quantum oracles that prepare their purifications.
    Quantum algorithms for computing the von Neumann entropy and closeness testing with respect to trace distance were developed in \cite{GL20} with query complexity $\tilde O(N/\varepsilon^{1.5})$ and $\tilde O(N/\varepsilon)$, respectively, both of which have complexity exponential in the number of qubits. Recently, a quantum algorithm for computing the von Neumann entropy within a multiplicative factor was proposed in \cite{GHS21}, which reproduces the result within additive error of \cite{GL20}.
    A method of computing the quantum $\alpha$-R\'{e}nyi entropy was proposed in \cite{SH21} using $O\left(\frac{\kappa}{(x\varepsilon)^2} \log\left(\frac{N}{\varepsilon}\right)\right)$ queries to the oracle, where $\kappa > 0$ is given such that $I/\kappa \leq \rho \leq I$ and $x = \tr(\rho^\alpha)/N$.

    Compared to the ``quantum sample access'' model where only identical copies of quantum states are directly given, the ``purified quantum query access'' model allows more potential operations from which one can learn properties of quantum states.
    This is because any operation allowed in the former model can be trivially simulated in the latter model.
    Consequently, the latter model (in query complexity) usually uses fewer computational resources than the former (in sample complexity). 
    For example, the best known query complexity for computing the von Neumann entropy is $\tilde O(N/\varepsilon^{1.5})$ \cite{GL20}, while the best known sample complexity for the same task is $O(N^2/\varepsilon^2)$ \cite{AISW19}. 
    In addition, the latter model also plays an important role in computational complexity theory when comparing classical and quantum computing. 
    For example, testing the closeness between two quantum states (in trace distance) is known to be $\mathsf{QSZK}$-complete (in certain parameter regime) \cite{Wat02,Wat13}, where the quantum states are given in the ``purified quantum query access'' model. 

    The recent quantum algorithms with purified quantum query access mentioned above are usually developed in the general framework of quantum singular value transformation (QSVT) \cite{GSLW19}. 
    The powerful technique of QSVT on unitary operators developed by \cite{GSLW19} has been successfully applied as a unified framework in a wide range of quantum algorithms, including Grover's search algorithm \cite{Gro96}, the quantum walk algorithms \cite{Amb04, Sze04}, the HHL algorithm for solving systems of linear equations \cite{HHL09}, and Hamiltonian simulation \cite{LC17}.
In the framework of QSVT, a unitary operator $U$ can be regarded as a block-encoding that stores a matrix $A$ in (the upper-left corner of) its matrix representation (see Definition \ref{def:block-encoding}). 
QSVT can be understood as an algorithmic technique that transforms the matrix $A$ to $f(A)$ for some function $f(\cdot)$ of interest, and it also provides a quantum circuit implementation of $\tilde U$ that block-encodes $f(A)$ using queries to $U$. 
The original QSVT deals with unitary operators, while a Hamiltonian variant of the QSVT was proposed in \cite{LKA+21}, which was then used in quantum polar decomposition \cite{LBDP+20, QR22}.
    
    Except for unitary operators and Hamiltonians, density operators (mixed quantum states) are another important class of objects we can manipulate in quantum computation.
    A technique was developed in \cite{LC19} to implement a unitary operator that block-encodes a density operator, using queries to its purified quantum query oracle. 
    Equipped with QSVT, this technique enables us to implement unitary operators that block-encode certain matrix functions of quantum states, and thus strengthens the power of the ``purified quantum query access'' model.
    This technique has been employed in quantum algorithms for semidefinite programming \cite{vAG19} and quantum fidelity estimation \cite{WZC+22}. Conversely, however, it seems difficult to prepare a quantum state from a unitary operator that block-encodes its density operator.
    
    A natural idea of extracting information from a density operator is to directly manipulate the quantum state itself rather than a unitary operator that encodes it. This leads us to extend the definition of block-encoding  proposed originally  for unitary operators \cite{LC19, GSLW19, CGJ19} to that for general operators (see Definition \ref{def:block-encoding}), especially for quantum states (i.e., density operators). Regarding quantum states as block-encodings, we are able to design new quantum algorithms for computing a wide range of quantum entropies and distances, such as the von Neumann entropy, quantum R\'{e}nyi entropy, quantum Tsallis entropy, trace distance, and fidelity.
    These quantum algorithms  significantly outperform the best known ones in the low-rank case, and some of them can even achieve exponential speedups. Here, the low-rank case means that the rank of $N$-dimensional quantum states is much smaller than $N$, e.g., $r = \polylog(N)$, which is of great interest in both theoretical (e.g., \cite{GLF10,EHC22}) and experimental (e.g., \cite{BGK15}) physics. 
    
    In the remainder of this Introduction,  we will first present  our main results in Section \ref{subsec:main-results}. The new techniques that enable us to achieve our results will be outlined in Section \ref{subsec:techniques}. Then related works will be reviewed  in Section \ref{subsec:related-works}, and a discussion will be given in Section \ref{subsec:discussion}.

    \subsection{Main Results} \label{subsec:main-results}
    
    Let us first set the stage for presenting our main results. In order to manipulate quantum states, we extend the definition of block-encoding for unitary operators to that for general operators (see Definition \ref{def:block-encoding}), and use this extended definition of block-encoding to describe our quantum algorithms. In our quantum algorithms, a mixed quantum state is given by a quantum unitary operator (oracle) which prepares a purification of the state (see Definition \ref{def:subnormalized-density-operator}). This conventional model is known as the ``purified quantum query access'' model and has been widely used in developing quantum algorithms \cite{BKL+19,vAG19,GL20,GLM+22,GHS21,SH21}.
    
    Throughout this paper, the quantum query complexity of a quantum query algorithm means the number of queries to the given quantum oracles. The time complexity of a quantum query algorithm is the sum of its quantum query complexity and the number of elementary quantum gates used in it. When quantum algorithms are compared with classical algorithms, quantum oracles are given as classical descriptions of quantum circuits. The actual number of elementary quantum gates performed in the quantum algorithm only has a polynomial overhead compared to its ``time complexity'' defined here. 
    Then our main results can be summarized in the following: 
    
    \begin{theorem} [Informal] \label{thm:main-results}
        In the ``purified quantum query access'' model, given quantum oracles that prepare $N$-dimensional mixed quantum states of rank $r$, there are quantum query algorithms that compute
        \begin{itemize}
          \item von Neumann entropy,
          \item quantum $\alpha$-R\'{e}nyi entropies for $\alpha \in (0, 1) \cup (1, +\infty)$, 
          \item quantum $\alpha$-Tsallis entropies for $\alpha \in (0, 1) \cup (1, +\infty)$, 
          \item $\alpha$-trace distance for $\alpha > 0$ (defined by Eq. (\ref{eq:alpha-trace-distance}), including the trace distance), and
          \item  $\alpha$-fidelity for $0 < \alpha < 1$ (defined by Eq. (\ref{eq:alpha-fidelity}), including the fidelity)
        \end{itemize}
        within additive error $\varepsilon$ with time complexity $\poly(\log(N), r, 1/\varepsilon)$, where the time complexity hides a constant which depends only on $\alpha$.\footnote{A few days after this paper was submitted to arXiv, the concurrent work of Gily\'{e}n and Poremba \cite{GP22} appeared. They proposed a quantum algorithm for fidelity estimation using identical copies of quantum states based on density matrix exponentiation \cite{LMR14,KLL+17}. 
        We note that their techniques of converting identical copies to  unitary block-encodings (Corollary 21 in \cite{GP22}) can be applied to our quantum algorithms in Theorem \ref{thm:main-results}.
        As a result, we can obtain quantum algorithms for computing these quantum entropies and distances using $\poly(r, 1/\varepsilon)$ copies of quantum states, which only has a polynomial overhead compared to the query complexity of our quantum query algorithms.}
    \end{theorem}
    
    Our quantum algorithms are compared with the existing algorithms
    in Table \ref{tab:comparison}. 
    In particular, our algorithms 
    outperform the best known ones in the low-rank case, e.g., $r = \polylog(N)$. To see this more clearly, let us recast the existing results in the low-rank case, and then compare them with ours. 
    Table \ref{tab:comparison} compares the query complexity (in the ``purified quantum query access'' model) and the sample complexity (in the ``quantum sample access'' model). 
    It is important to note the difference between the two models in the comparison. 
    While it is trivial that a sample in the ``quantum sample access'' model can be simulated by a query in the ``purified quantum query access'' model, the vice versa remains unknown. 
    In other words, the ``quantum sample access'' model can be seen as a restricted version of the ``purified quantum query access'' model where the oracle is only used to prepare the mixed quantum states.
    Therefore, any sample complexity in the ``quantum sample access'' model implies the same amount of query complexity in the ``purified quantum query access'' model.
    Regarding these, we still compare the query/sample complexities defined in the two models together.
    For further discussion regarding the two different quantum input models, please refer to \cite{GL20}.

    \begin{table*}[h]
    \begin{threeparttable}
    \centering
    \caption{Our quantum algorithms vs. prior works\tnote{*}}
    \label{tab:comparison}
    \begin{tabular}{@{}p{\textwidth}@{}}
    \centering
    \begin{tabular}{cccc}
    \toprule
    \begin{tabular}[c]{@{}c@{}}Quantum Information\\ Quantity\end{tabular}         & \begin{tabular}[c]{@{}c@{}} Prior Best\\ Sample Complexity\end{tabular} & \begin{tabular}[c]{@{}c@{}} Prior Best\\ Query Complexity\end{tabular} & \begin{tabular}[c]{@{}c@{}}Our \\ Query Complexity\end{tabular} \\
    \midrule
    Von Neumann Entropy      & \begin{tabular}[c]{@{}c@{}}$O(N^2/\varepsilon^2)$, $\tilde O\left(\kappa^2/\varepsilon^5\right)$\\ 
    \cite{AISW19}, \cite{WZW22}\end{tabular}                                                    & \begin{tabular}[c]{@{}c@{}}$\tilde O(\sqrt{Nr}/\varepsilon^{1.5})$, $\tilde O(\kappa^2/\varepsilon)$\\ \cite{GL20}, \cite{CLW20}\end{tabular}                                                    & $\tilde O(r/\varepsilon^2)$                                                  \\ 
    \midrule
     \begin{tabular}[c]{@{}c@{}}Quantum $\alpha$-R\'{e}nyi Entropy\\ (for non-integer $\alpha > 0$)\end{tabular}    & \begin{tabular}[c]{@{}c@{}}$O\left((N/\varepsilon)^{\max\{2/\alpha, 2\}}\right)$\\ \cite{AISW19}\end{tabular}                                                     & \begin{tabular}[c]{@{}c@{}}$\tilde O \left( \kappa N r^{\max\{\alpha - 1, 0\}} / \varepsilon^2 \right)$ \\ \cite{SH21}\end{tabular}                                                & $\tilde O\left(\frac{r^{\alpha-1+\alpha/\{\frac{\alpha-1}{2}\}}}{\varepsilon^{1+1/\{\frac{\alpha-1}{2}\}}}\right)$                                                \\ 
    \midrule
    Trace Distance    & \begin{tabular}[c]{@{}c@{}}$O(r/\varepsilon^2)$\tnote{\textdagger}\\ \cite{BOW19}\end{tabular}                                                 & \begin{tabular}[c]{@{}c@{}}$\tilde O(\min\{ \sqrt{Nr}/\varepsilon, r/\varepsilon^2 \})$\tnote{\textdagger} \\ \cite{GL20}\end{tabular}                                                    & $\tilde O(r^5/\varepsilon^6)$                                                  \\ 
    \midrule
    Fidelity & \begin{tabular}[c]{@{}c@{}}$O(r/\varepsilon)$\tnote{\textdagger} \\ \cite{BOW19}\end{tabular}                                                   & \begin{tabular}[c]{@{}c@{}}$\tilde O\left(r^{12.5}/\varepsilon^{13.5}\right)$\tnote{\textdaggerdbl} \\ \cite{WZC+22}\end{tabular}                                                   & $\tilde O(r^{6.5}/\varepsilon^{7.5})$                                                  \\ \bottomrule
    \end{tabular}
    \end{tabular}
    \begin{tablenotes}
    \footnotesize
    \item[*] $N$ is the dimension of quantum states. $\varepsilon$ is the desired additive error. $r$ is (an upper bound for) the rank of quantum states. $\{x\} = x - \floor{x}$ denotes the decimal part of $x$. $\kappa$ is the parameter associated with mixed quantum state $\rho$ such that $I/\kappa \leq \rho \leq I$, which only appears in the best prior query complexity of quantum $\alpha$-R\'{e}nyi entropy. 
    \item[\textdagger] These are the sample/query complexities for closeness testing with respect to the trace distance (resp. fidelity). It is worth mentioning that closeness testing can be solved by computing the closeness, but the converse seems difficult. 
    \item[\textdaggerdbl] In the concurrent work of Gily\'{e}n and Poremba \cite{GP22}, they presented a different quantum algorithm for fidelity estimation with a better query complexity $\tilde O\left(r^{2.5}/\varepsilon^5\right)$.
    \end{tablenotes}
    \end{threeparttable}
    \end{table*}
    
    \begin{itemize}
    \item For the von Neumann entropy and quantum $\alpha$-R\'{e}nyi entropy, it was shown in \cite{AISW19} that their sample complexities are $O\left(N^2/\varepsilon^2\right)$ and $O\left((N/\varepsilon)^{\max\{2/\alpha, 2\}}\right)$, respectively. Unitarily invariant properties of entropies considered, their method is based on weak Schur sampling (see \cite{MdW16}), and does not imply a straightforward method for low-rank quantum states. The query complexity for the von Neumann entropy was shown in \cite{GL20} to be $\tilde O(N/\varepsilon^{1.5})$, which can be improved to $\tilde O(\sqrt{Nr}/\varepsilon^{1.5})$ for the low-rank case after a careful analysis. 
    It was shown in \cite{CLW20} that $\tilde O(\kappa^2/\varepsilon)$ queries are sufficient to compute the von Neumann entropy of a quantum state $\rho$ if some $\kappa > 0$ is known in advance such that $\rho \geq I/\kappa$.
    Similarly, the method in \cite{SH21} for the quantum $\alpha$-R\'{e}nyi entropy can be improved to $\tilde O\left(\kappa N r^{\max\{\alpha - 1, 0\}} / \varepsilon^2 \right)$ for the low-rank case.
    \item For the trace distance and fidelity, most algorithms are proposed for closeness testing with respect to them. The sample complexities $O\left(N/\varepsilon^2\right)$ and $O\left(N/\varepsilon\right)$ given in \cite{BOW19} can be improved (by their Corollary 1.6) to $O\left(r/\varepsilon^2\right)$ and $O\left(r/\varepsilon\right)$ for the low-rank case, respectively. 
    The query complexity $\tilde O(N/\varepsilon)$ given in \cite{GL20} can be improved to $\tilde O(\min\{ \sqrt{Nr}/\varepsilon, r/\varepsilon^2 \})$ for the low-rank case. The above results do not cover our results, because closeness testing can be solved by computing the closeness but the converse seems difficult. 
    \item For the quantum algorithms in \cite{CLW20}, \cite{SH21}, and \cite{WZW22} that attempt to reduce the dependence on $N$, they introduce an extra dependence on $\kappa$, where $\kappa$ is the reciprocal of the minimal non-zero eigenvalue of quantum states. Our quantum algorithms can be easily adapted to their settings by taking $r = O(\kappa)$, thus with time complexity $\poly(\log(N), \kappa, 1/\varepsilon)$, while the converse seems not applicable.\footnote{This is because $\kappa$ implies an upper bound $r \leq \kappa$ of rank, but $r$ does not imply any upper bound for $\kappa$.} 
    \end{itemize}

    Although our quantum algorithms focus on low-rank quantum states, they are also comparable to those for the general case where one could only assume that the quantum states are full-rank, i.e., $r = N$. 
    \begin{itemize}
        \item For the von Neumann entropy, our quantum algorithm has query complexity $\tilde O(N/\varepsilon^2)$ when the quantum state is full-rank, which is slightly worse than the query complexity $\tilde O(N/\varepsilon^{1.5})$ in \cite{GL20}. 
        \item For the trace distance, our quantum algorithm has query complexity $\tilde O(N^{5}/\varepsilon^6)$ when the quantum state is full-rank. 
        To the best of our knowledge, this is the first quantum algorithm for computing the trace distance between quantum states with time complexity $\poly(N)$ in the general case.
    \end{itemize}
    
    We will further discuss the above results for quantum entropies in Section \ref{subsec:quantum-entropy}  and those for closeness (i.e. trace distance and fidelity) of quantum states in Section \ref{subsec:quantum-distance}.
    
    \subsubsection{Computing quantum entropies} \label{subsec:quantum-entropy}
    
    In quantum information theory, the entropy of a (mixed) quantum state is a measure of its uncertainty, and computing its value is crucial when characterizing and verifying an unknown quantum system.
    After von Neumann \cite{vN32} introduced the famous von Neumann entropy
    \[
        S(\rho) = -\tr\left(\rho \ln\left(\rho\right)\right),
    \]
    which is a natural generalization of the classical Shannon entropy \cite{Sha48}, several other  entropies have been proposed, e.g., R\'{e}nyi entropy \cite{Ren61, LMW13, MDS13, WWY14}, Tsallis entropy \cite{Tsa88, Aud07, PV15}, Min and Max (Hartley) entropies \cite{vDH02, Dat09, KRS09}, and the unified entropy \cite{HY06, Ras11}.
    The quantum $\alpha$-R\'{e}nyi entropy and the quantum $\alpha$-Tsallis entropy are defined by
    \begin{align*}
        S^{R}_{\alpha}(\rho) & = \frac{1}{1-\alpha} \ln \left( \tr \left( \rho^\alpha \right) \right), \\
        S^{T}_{\alpha}(\rho) & = \frac{1}{1-\alpha} \left( \tr\left(\rho^\alpha\right) - 1 \right)
    \end{align*}
    for $\alpha \in (0, 1) \cup (1, +\infty)$, respectively. It is easy to see that the von Neumann entropy is a limiting case of the R\'{e}nyi entropy \cite{MDS13} and the Tsallis entropy \cite{Tsa88}: 
    \[
        S(\rho) = \lim_{\alpha\to 1} S^{R}_{\alpha}(\rho) = \lim_{\alpha\to 1} S^{T}_{\alpha}(\rho).
    \]
    For $\alpha = 0$, the quantum Tsallis entropy degenerates to the rank of quantum states: $$S_0^{T}(\rho) = \rank(\rho) - 1$$ and the quantum R\'{e}nyi entropy becomes the logarithm of the rank, i.e., the quantum Max (Hartley) entropy: $$S^{\max}(\rho) = S_0^{R}(\rho) = \ln(\rank(\rho)).$$
    
    \paragraph{Overview} Given a quantum unitary oracle that prepares a mixed quantum state (see Definition \ref{def:subnormalized-density-operator}), we develop quantum algorithms for computing several quantum entropies. Their quantum query complexities are collected in Table \ref{tab:complexity-entropy}, which are also their quantum time complexities up to polylogarithmic factors. 
    Most of our algorithms do not require any restrictions on the lower bound for the eigenvalues of quantum states except those for computing the quantum Max entropy and the rank of quantum states, where $\Pi/\kappa \leq \rho$ is required for some projector $\Pi$ and $\kappa > 0$.
    
    The prior best quantum algorithms for computing von Neumann entropy \cite{AISW19,GL20,GHS21} and quantum R\'{e}nyi entropy  \cite{AISW19,SH21} have time complexity $\Omega(N)$ even for rank $r = 2$. Compared to them, our quantum algorithms are exponentially faster in the low-rank case. 
    In particular, our quantum algorithm for computing the von Neumann entropy with query complexity $\tilde O(r/\varepsilon^2)$ is comparable to the quantum algorithm given in \cite{CLW20} with query complexity $\tilde O(\kappa^2/\varepsilon)$, where we note that $r \leq \kappa$ always holds.

    It is worth mentioning that for odd integer $\alpha > 1$, the query complexity of computing the quantum Tsallis entropy $S_{\alpha}^T(\rho)$ does not depend on rank $r$. In this case, there is a simple SWAP test-like quantum algorithm that computes $\tr(\rho^\alpha)$ using $O(1/\varepsilon^2)$ copies \cite{EAO+02,KLL+17}. Compared to it, our algorithm 
    (Theorem \ref{thm:trace-positive-powers}) yields a quadratic speedup (see Section \ref{sec:renyi-tsallis} for more discussions). For non-integer $\alpha$, we are not aware of any prior approaches for computing the quantum Tsallis entropy with complexity better than quantum state tomography.

    \begin{table*}[!htp]
    \begin{center}
    \caption{Quantum query complexity for computing quantum entropies\textsuperscript{*}}
    \vspace{4pt}
    \label{tab:complexity-entropy}
    \begin{tblr}{|c|c|c|c|}
    \toprule
    \SetCell[c=2]{}{Parameter $\alpha$} & & Quantum R\'{e}nyi Entropy $S_\alpha^R(\rho)$ & Quantum Tsallis Entropy $S_\alpha^T(\rho)$ \\ 
    \midrule
    \SetCell[c=2]{}{$\alpha = 0$} & & {\begin{tabular}[c]{@{}c@{}}$\tilde O\left(\kappa^2 / \varepsilon\right)$ (Theorem \ref{thm:max-entropy})\\ (Max Entropy)\end{tabular}} & \begin{tabular}[c]{@{}c@{}} $\tilde O\left(\kappa^2\right)$ (Corollary \ref{corollary:exact-rank})\\ (Rank)\end{tabular} \\
    \midrule
    \SetCell[c=2]{}{$0 < \alpha < 1$} & & \SetCell[c=2]{} $\tilde O\left( r^{\frac{3-\alpha^2}{2\alpha}} / \varepsilon^{\frac{3+\alpha}{2\alpha}} \right)$ (Theorem \ref{thm:Renyi-entropy} and \ref{thm:tsallis-entropy}) \\
    \midrule
    \SetCell[c=2]{} \begin{tabular}[c]{@{}c@{}}$\alpha = 1$\\ (Von Neumann Entropy)\end{tabular} & & \SetCell[c=2]{} $\tilde O\left( r / \varepsilon^{2} \right)$ (Theorem \ref{thm: von Neumann entropy})      \\
    \midrule
    \SetCell[r=2]{} $\alpha > 1$ & $\alpha \equiv 1 \pmod 2$ & $O\left(r^{\alpha-1}/\varepsilon\right)$ (Theorem \ref{thm:Renyi-entropy}) & $O\left(1/\varepsilon\right)$ (Theorem \ref{thm:tsallis-entropy}) \\
    \midrule
    & $\alpha \not\equiv 1 \pmod 2$ & \SetCell[c=2]{} $\tilde O\left( r^{\alpha - 1 + \alpha / \{\frac{\alpha - 1}{2}\}} / \varepsilon^{1 + 1 / \{\frac{\alpha - 1}{2}\}} \right)$ (Theorem \ref{thm:Renyi-entropy} and \ref{thm:tsallis-entropy}) \\
    \bottomrule
    \end{tblr}
    \end{center}
    \begin{tablenotes}
    \footnotesize
        \item \textsuperscript{*} $r$ is (an upper bound for) the rank of quantum states. $\varepsilon$ is the desired additive error. $\{x\} = x - \floor{x}$ denotes the decimal part of $x$. $\kappa$ is the parameter associated with mixed quantum state $\rho$ such that $\Pi/\kappa \leq \rho$ for some projector $\Pi$, which is only used in the case $\alpha = 0$.
    \end{tablenotes}
    \end{table*}
    
    \paragraph{Lower bounds} We are able to give a query lower bound $\tilde \Omega(r^{c})$ for computing the quantum R\'{e}nyi entropy $S_\alpha^R(\rho)$ including the von Neumann entropy $S(\rho)$ in terms of rank $r$, where $c \geq 1/3$ is a constant depending only on $\alpha$ (see Theorem \ref{thm:lower-bounds-entropy}). This lower bound is simply derived from the quantum query complexity for computing the R\'{e}nyi (and Shannon) entropy of classical probability distributions \cite{LW19,BKT20}. 

    \subsubsection{Computing quantum distances} \label{subsec:quantum-distance}
    
    Distance measures of quantum states are basic quantities in quantum computation and quantum information. Two of the most important distance measures are the trace distance and fidelity. For each of them, we propose  quantum algorithms  that compute it and its extensions. Here, we assume that there are two quantum oracles $U_\rho$ and $U_\sigma$ that prepare the density operators $\rho$ and $\sigma$, respectively. The query complexity of a quantum algorithm means the total number of queries to both $U_\rho$ and $U_\sigma$. 
    
    \paragraph{Trace distance}
    
    The $\alpha$-trace distance of two quantum states $\rho$ and $\sigma$ is defined by
    \begin{equation} \label{eq:alpha-trace-distance}
        T_\alpha(\rho, \sigma) = \tr\left(\abs{\frac{\rho - \sigma}{2}}^{\alpha}\right) = \Abs{\frac{\rho - \sigma}{2}}_{S, \alpha}^{\alpha},
    \end{equation}
    where $\Abs{A}_{S, \alpha} = \left(\tr\left(\abs{A}^\alpha\right)\right)^{1/\alpha}$ is the Schatten $\alpha$-norm. Here, the $1$-trace distance is the well-known trace distance $T(\rho, \sigma) = T_1(\rho, \sigma)$. 

    We develop quantum algorithms for computing $\alpha$-trace distance for $\alpha > 0$, with their query complexities shown in Table \ref{tab:trace-distance}.
    As a special case, our quantum algorithm (Theorem \ref{thm:trace-distance}) for computing the trace distance (i.e., the $1$-trace distance) has query complexity $\tilde O\left(r^{5}/\varepsilon^{6}\right)$.
    Note that the closeness testing of the $\alpha$-trace distances of quantum states for integer $\alpha$, e.g., the $1$-, $2$- and $3$-trace distances, was studied in \cite{GL20}. 
    For other cases of $\alpha$, we are not aware of any prior approaches to compute the $\alpha$-trace distance with complexity better than quantum state tomography.

    \begin{table*}[!htp]
    \begin{threeparttable}
    \caption{Quantum query complexity for computing $\alpha$-trace distance\tnote{*}}
    \label{tab:trace-distance}
    \begin{tabular}{@{}p{\textwidth}@{}}
    \centering
    \begin{tabular}{cc}
    \toprule
    Parameter $\alpha$                                                       &  $T_\alpha(\rho, \sigma)$ (Theorem \ref{thm:trace-distance})  \\ 
    \midrule
    $0 < \alpha < 1$                                                      & $\tilde O\left( r^{5/\alpha+(1-\alpha)/2}/\varepsilon^{5/\alpha+1} \right)$ \\ 
    \midrule
    $\alpha \equiv 0 \pmod 2$                                                      & $\tilde O\left(r^3/\varepsilon^4\right)$ \\ 
    \midrule
    $\alpha \geq 1$ and $\alpha \not \equiv 0 \pmod 2$                                                     & $\tilde O\left( r^{3+1/\{\alpha/2\}}/\varepsilon^{4+1/\{\alpha/2\}} \right)$ \\ 
    \midrule
    $\alpha = 1$ (Trace Distance) &
    $\tilde O\left(r^5/\varepsilon^6\right)$ \\ 
    \bottomrule
    \end{tabular}
    \end{tabular}
    \begin{tablenotes}
    \footnotesize
    \item[*]
    $r$ is (an upper bound for) the higher rank of the two quantum states. $\varepsilon$ is the desired additive error. $\{x\} = x - \floor{x}$ denotes the decimal part of $x$. 
    \end{tablenotes}
    \end{threeparttable}
    \end{table*}
    
    \paragraph{Fidelity}
    
    The $\alpha$-fidelity of two quantum states $\rho$ and $\sigma$ is defined by
    \begin{equation} \label{eq:alpha-fidelity}
        F_{\alpha}(\rho, \sigma) = \exp\left((\alpha - 1) D_{\alpha}(\rho \| \sigma)\right) = \tr\left( \left( \sigma^{\frac{1-\alpha}{2\alpha}} \rho \sigma^{\frac{1-\alpha}{2\alpha}} \right)^{\alpha} \right),
    \end{equation}
    where $D_{\alpha}(\rho \| \sigma)$ is the sandwiched quantum R\'{e}nyi relative entropy \cite{WWY14, MDS13}. Here, the $1/2$-fidelity is the well-known fidelity $F(\rho, \sigma) = F_{1/2}(\rho, \sigma)$ \cite{Joz94}. 
    
    We develop quantum algorithms for computing the $\alpha$-fidelity for $0 < \alpha < 1$, with their query complexities shown in Table \ref{tab:fidelity}).
    As a special case, our quantum algorithm (Theorem \ref{thm:fidelity}) for computing the fidelity (i.e., the $1/2$-fidelity) has query complexity $\tilde O\left(r^{6.5}/\varepsilon^{7.5}\right)$, which is a polynomial speedup over the best known $\tilde O\left(r^{12.5}/\varepsilon^{13.5}\right)$ in \cite{WZC+22}. For other cases of $\alpha$, we do not know any prior approaches to compute the $\alpha$-fidelity with complexity better than quantum state tomography.

    \begin{table*}[!htp]
    \begin{threeparttable}
    \caption{Quantum query complexity for computing $\alpha$-fidelity\tnote{*}}
    \label{tab:fidelity}
    \begin{tabular}{@{}p{\textwidth}@{}}
    \centering
    \begin{tabular}{cc}
    \toprule
    Parameter $\beta = (1-\alpha)/2\alpha$                                                       &  $F_{\alpha}(\rho, \sigma)$ (Theorem \ref{thm:fidelity})  \\ 
    \midrule
    $\beta \in \mathbb{N}$                                                      & $\tilde O\left(r^{\frac{3-\alpha}{2\alpha}}/\varepsilon^{\frac{3+\alpha}{2\alpha}}\right)$ \\ 
    \midrule
    $\beta \notin \mathbb{N}$                                                      & $\tilde O\left(r^{\frac{3-\alpha}{2\alpha}+\frac{1}{\alpha\{\beta\}}}/\varepsilon^{\frac{3+\alpha}{2\alpha}+\frac{1}{\alpha\{\beta\}}}\right)$ \\ 
    \midrule
    $\alpha = \beta = 1/2$ (Fidelity)
    & $\tilde O\left(r^{6.5}/\varepsilon^{7.5}\right)$ \\ \bottomrule
    \end{tabular}
    \end{tabular}
    \begin{tablenotes}
    \footnotesize
    \item[*]
    $r$ is (an upper bound for) the lower rank of the two quantum states. $\varepsilon$ is the desired additive error. $\{x\} = x - \floor{x}$ denotes the decimal part of $x$. 
    \end{tablenotes}
    \end{threeparttable}
    \end{table*}
    
    \paragraph{Lower bounds and hardness}
    Our quantum algorithms for computing the fidelity and trace distance have a time complexity polynomial in the rank $r$. We show that there is no quantum algorithm that computes the fidelity or trace distance with time complexity $\poly(\log(r), 1/\varepsilon)$ unless $\mathsf{BQP} = \mathsf{QSZK}$ (see Theorem \ref{thm:lower-bounds}), based on the result of \cite{Wat02} that $(\alpha, \beta)$-Quantum State Distinguishability is $\QSZK$-complete for $0 \leq \alpha < \beta^2 \leq 1$.\footnote{The available regime of $\alpha$ and $\beta$ for the $\mathsf{QSZK}$-completeness of $(\alpha, \beta)$-Quantum State Distinguishability was recently improved to $0 \leq \sqrt{2\ln2} \alpha < \beta^2 \leq 1$ in \cite{Liu23}.}
    
    Our quantum algorithms for computing the fidelity and trace distance achieve a significant speedup under the low-rank assumption. 
    We argue that these problems are unlikely to be efficiently solved by classical computers because computing the fidelity and trace distance are $\mathsf{DQC1}$-hard (see Theorem \ref{thm:dqc1-hard}); and it was shown in \cite{FKM18} that $\mathsf{DQC1}$ is not (classically) weakly simulatable unless the polynomial hierarchy collapses to the second level, i.e., $\mathsf{PH} = \mathsf{AM}$.
    
    \subsection{Techniques} \label{subsec:techniques}
    
    In this subsection, we give an overview of the techniques that enable us to achieve the results presented in the above subsection.
    
    \subsubsection{Quantum states as block-encodings}
    
    The key idea of our quantum algorithms is to regard quantum states as block-encodings. 
    To this end, we extend the definition of block-encoding proposed for unitary operators to that for general ones (see Definition \ref{def:block-encoding}). Suppose that a unitary operator $U_A$ prepares a subnormalized density operator $A$ (see Definition \ref{def:subnormalized-density-operator}). In this framework, we provide a convenient way to manipulate the subnormalized density operator $A$ and extract information from it as follows.
    
    \begin{itemize}
      \item \textbf{Evolution}: If $U$ is a unitary operator, which is a block-encoding of an operator $B$, we can prepare a subnormalized density operator $B A B^\dag$ (see Lemma \ref{lemma:density-basic}). This evolution of the subnormalized density operator can be seen as a generalization of quantum unitary operation $\rho \mapsto U \rho U^\dag$ for (normalized) density operator $\rho$.
      
      \item \textbf{Trace Estimation}: We provide an efficient method to estimate the trace of $A$ based on quantum amplitude estimation \cite{BHMT02} (see Lemma \ref{lemma:trace-estimation}). As will be seen, trace estimation is an important subroutine in our quantum algorithms (see Section \ref{subsec:quantum-entropy} and Section \ref{subsec:quantum-distance}).
      
      \item \textbf{Linear Combinations}: As an analog of Linear-Combination-of-Unitaries (LCU) algorithm through a series of work \cite{SOGKL02,CW12,Kothari14,BCCKS15,BCK15,CKS17,GSLW19}, we also provide a technique to prepare a linear (convex) combination of subnormalized density operators (see Lemma \ref{lemma:linear-combination}). This technique will be used in computing the trace distance (see Section \ref{subsec:quantum-distance} and Theorem \ref{thm:trace-distance}).
    \end{itemize}
    
    The technique of ``trace estimation'' is the cornerstone in developing our quantum algorithms. To compute the values of quantum entropies and distances, the key part has the form $\tr(\varrho)$, where $\varrho$ is a (subnormalized) density operator. Our strategy is to prepare a quantum state, which is a block-encoding of $\varrho$, through the technique of ``evolution''. Roughly speaking, we prepare the subnormalized density operator $\varrho$ up to a scaling factor; we will use the phrase ``prepare $\varrho$'' regardless of the scaling factor in the following discussion of this section. For example, we prepare $- \rho \ln(\rho)$ for the von Neumann entropy, and prepare $\rho^{\alpha}$ for the quantum $\alpha$-R\'{e}nyi and Tsallis entropies. To achieve this, we develop techniques for eigenvalue transformation of density operators based on QSVT as follows. 
    
    \begin{itemize}
      \item \textbf{Eigenvalue Transformation}: Based on the evolution, if we can construct a unitary operator $U$, which is a block-encoding of $P(A)$ for some polynomial $P(\cdot)$ as in QSVT \cite{GSLW19}, we can transform $A$ to another subnormalized density operator $A(P(A))^2$ (see Theorem \ref{lemma:technique}).
      
      \item \textbf{Positive Powers}: We develop a technique to prepare the subnormalized density operator $A^c$ for $0 < c < 1$ without any restrictions on $A$ (see Lemma \ref{lemma:positive-power-density}). Inspired by this, we can also obtain a unitary operator, which is a block-encoding of $\abs{A}^c$, using queries to a unitary operator $U$, which is a block-encoding of Hermitian operator $A$ (see Lemma \ref{lemma:positive-power-unitary}). In order to obtain block-encodings of powers of $A$, previously known methods \cite{CGJ19,GSLW19,GLM+22} usually require a lower bound for the minimal non-zero eigenvalues of density operators; for example, $I/\kappa \leq A \leq I$ for some $\kappa > 0$ in \cite{CGJ19}. This technique for positive powers of subnormalized density operators will be frequently used in our quantum algorithms for computing quantum entropies, fidelity and trace distance (see Section \ref{subsec:quantum-entropy} and Section \ref{subsec:quantum-distance}), in order to avoid restrictions on density operators.
    \end{itemize}
    
    We also provide a method to block-encode the eigenvalue threshold projector $\Pi_{\supp(A)}$ of $A$ in a quantum state, where $\supp(A)$ is the support of $A$, and $\Pi_S$ is the projector onto subspace $S$.
    
    \begin{itemize}
      \item \textbf{Eigenvalue threshold projector}: We propose a method to (approximately) block-encode the eigenvalue threshold projector $\Pi_{\supp(A)}$ of $A$ in a subnormalized density operator (see Lemma \ref{lemma:eigenvalue-threshold-projector}). We note that a technique for block-encoding eigenvalue threshold projectors was also provided in \cite{vAG19}, but they required that $A \geq q\Pi$ for some projector $\Pi$ and the value of $q > 0$ is known in advance. In contrast, our method does not impose any restriction on $A$. This method will be used in computing the trace distance (see Section \ref{subsec:quantum-distance} and Theorem \ref{thm:trace-distance}).
    \end{itemize}
    
    A comparison between density operators and unitary operators as block-encodings is given in Table \ref{tab:cmp-qsvt}.

    \begin{table*}[!htp]
    \centering
    \begin{threeparttable}
    \caption{Comparison between density operators and unitary operators as block-encodings\tnote{*}}
    \label{tab:cmp-qsvt}
    \begin{tabular}{ccc}
    \toprule
    Operation Type                      & \begin{tabular}[c]{@{}c@{}}Density Operators\\ $\rho \approx \begin{bmatrix}
        A & \cdot \\
        \cdot & \cdot
    \end{bmatrix}$ \end{tabular} & \begin{tabular}[c]{@{}c@{}}Unitary Operators\\ $U \approx \begin{bmatrix}
        A & \cdot \\
        \cdot & \cdot
    \end{bmatrix}$ \end{tabular} \\
    \midrule
    Evolution                 & $A \to BAB^\dag$       & $A \to AB \text{ or } BA$       \\ 
    Trace Estimation          & $\tr(A)$       & $\tr(A)/2^a$ \tnote{\textdagger}     \\ 
    Linear Combination        & $\alpha_1A_1 + \dots + \alpha_kA_k \  (\alpha_i \in \mathbb{R}^+)$       & $\alpha_1A_1 + \dots + \alpha_kA_k \  (\alpha_i \in \mathbb{C})$       \\ 
    Eigenvalue Transformation & $A \to A(P(A))^2$       & $A \to P(A)$      \\ 
    Positive Powers & $A \to A^c (0 < c < 1)$       & $A \to \abs{A}^c (0 < c < 1)$      \\ 
    Eigenvalue Threshold Projector & $A \to \text{ (scaled) } \Pi_{\supp(A)}$       & $A \to \Pi_{\supp(A)}$      \\ 
    \bottomrule
    \end{tabular}
    \begin{tablenotes}
      \footnotesize
      \item[*] $A$ is an Hermitian operator block-encoded in a density $\rho$ or a unitary operator $U$. $B$ is block-encoded in a unitary operator. $P(\cdot)$ is a polynomial.
      \item[\textdagger] Suppose $A$ is an $a$-qubit Hermitian operator block-encoded in unitary operator $U$. Then $\tr(A)/2^a = \tr\left( (\ket{0}\bra{0} \otimes \frac{I_a}{2^a}) U \right)$ can be computed through the Hadamard test \cite{EAO+02}.
    \end{tablenotes}
    \end{threeparttable}
    \end{table*}
    
    \subsubsection{Example --- computing trace distance}
    
    To give the readers a flavor, we take the quantum algorithm for computing the trace distance (see Theorem \ref{thm:trace-distance} for details) as an illustrative example. The key observation to compute the trace distance is that
    \[
        T(\rho, \sigma) = \tr\left( \abs{\nu}^{1/2} \Pi_{\supp(\mu)} \abs{\nu}^{1/2} \right),
    \]
    where $\nu = (\rho-\sigma)/2$, $\mu = (\rho+\sigma)/2$. The idea is to prepare $\eta = \abs{\nu}^{1/2} \Pi_{\supp(\mu)} \abs{\nu}^{1/2}$ (up to a scaling factor), and then estimate $\tr(\eta)$ through the technique of ``trace estimation'' (Lemma \ref{lemma:trace-estimation}).
    The computation process is shown in Figure \ref{fig:process}.

    \begin{figure*}[!htp]
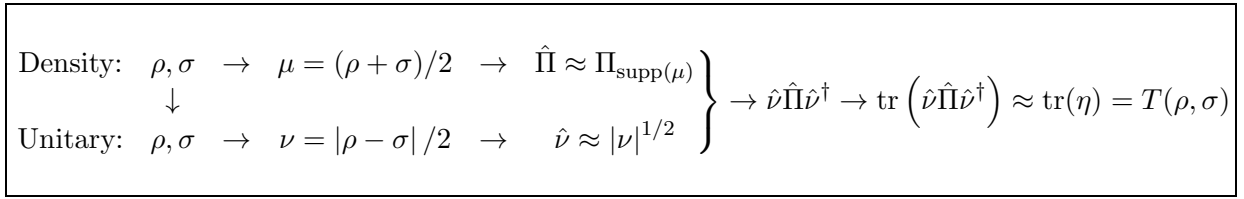

    \centering
    \noindent\fbox{%
    \parbox{16cm}{%
        \[
        \left. \begin{matrix}
            \text{Density:} & \rho, \sigma & \to & \mu = (\rho+\sigma)/2 & \to & \hat\Pi \approx \Pi_{\supp(\mu)} \\
            & \downarrow \\
            \text{Unitary:} & \rho, \sigma & \to & \nu = \abs{\rho-\sigma}/2 & \to & \hat\nu \approx \abs{\nu}^{1/2}
        \end{matrix} \right\} \to \hat\nu \hat\Pi \hat\nu^\dag \to \tr\left(\hat\nu \hat\Pi \hat\nu^\dag\right) \approx \tr(\eta) = T(\rho, \sigma)
        \]
    }
    }
        \caption{The computation process of computing trace distance.}
        \label{fig:process}
    \end{figure*}
    
    To (approximately) prepare $\eta$, we first prepare $\hat \Pi \approx \Pi_{\supp(\mu)}$ through the technique of ``eigenvalue threshold projector'' (Lemma \ref{lemma:eigenvalue-threshold-projector}). Here, $\mu = (\rho + \sigma) / 2$ can be prepared through the technique of ``linear combinations'' (Lemma \ref{lemma:linear-combination}). Then we only need to construct a unitary operator, which is a block-encoding of $\hat \nu \approx \abs{\nu}^{1/2}$. After that, by the technique of ``evolution'' (Lemma \ref{lemma:density-basic}), we can prepare $\hat \nu \hat \Pi \hat \nu^\dag \approx \abs{\nu}^{1/2} \hat \Pi \abs{\nu}^{1/2} \approx \eta$.
    
    In order to construct a unitary operator as a block-encoding of $\abs{\nu}^{1/2}$, we first use the LCU technique (Theorem \ref{thm:lcu}) to block-encode $\nu = (\rho-\sigma)/2$ in a unitary operator $U_{\nu}$. Then applying the technique of ``positive powers'' (Lemma \ref{lemma:positive-power-unitary}) on $U_{\nu}$, we can construct a unitary operator which is a block-encoding of $\hat\nu \approx \abs{\nu}^{1/2}$.
    
    To get an estimation of the trace distance between $\rho$ and $\sigma$, we should just note that $T(\rho, \sigma) = \tr(\eta) \approx \tr\left(\hat \nu \hat \Pi \hat \nu^\dag\right)$, where we have already prepared $\hat \nu \hat \Pi \hat \nu^\dag$ through the above process. Strictly speaking, we have prepared a mixed quantum state, whose density operator is a block-encoding of $\hat \nu \hat \Pi \hat \nu^\dag$ up to a scaling factor. After carefully selecting appropriate parameters that determine the errors in the above process, we obtain a quantum algorithm for computing the trace distance with query complexity $\tilde O(r^5/\varepsilon^6)$, where $r$ is (an upper bound for) the rank of quantum states $\rho$ and $\sigma$, and $\varepsilon$ is the desired additive error. Here, $\tilde O(\cdot)$ suppresses the polylogarithmic factor of $N$, where $N$ is the dimension of the Hilbert space of $\rho$ and $\sigma$.
    
    \subsection{Related Works} \label{subsec:related-works}
    
    \paragraph{Classical property testing}
    
The problems considered in this paper can be thought of as a quantum analog of testing properties of probability distributions. Classical algorithms for testing properties of probability distributions have been widely studied since the beginning of this century.
    The first algorithm was proposed in \cite{BFR+00} for the closeness testing of probability distributions in $\ell^1$ distance using $\tilde O(N^{2/3}/\varepsilon^4)$ samples, which was then improved to use $\tilde O(N^{2/3}/\varepsilon^{8/3})$ samples \cite{BFR+13}. 
    Later, it was shown in \cite{CDVV14} that the optimal sample complexity for this problem is $\Theta\left(\max\{N^{2/3}/\varepsilon^{4/3}, N^{1/2}/\varepsilon^{2}\}\right)$, and they also proved that the optimal sample complexity $\Theta(1/\varepsilon^2)$ for closeness testing in $\ell^2$ distance. 
    The identity testing is a special case of the closeness testing given that one of the distributions is known. 
    It was shown in \cite{BFF+01} that $\tilde O(N^{1/2}/\varepsilon^4)$ samples are sufficient for the identity testing in $\ell^1$ distance, which was improved to optimal $\Theta(N^{1/2}/\varepsilon^2)$ in \cite{Pan08}.
    The independence testing, i.e., whether a distribution on $[N] \times [M]$ $(N \geq M)$ is equal to or $\varepsilon$-far from a product distribution in $\ell^1$ distance, was shown to have sample complexity $\tilde O(N^{2/3}M^{1/3}) \cdot \poly(1/\varepsilon)$ \cite{BFF+01}.
    Recently, a modular reduction-based approach was proposed in \cite{DK16}, which covers the closeness, identity and independence testing. They also gave a tight sample complexity $\Theta\left(\max\{N^{2/3}M^{1/3}/\varepsilon^{4/3}, (NM)^{1/2}/\varepsilon^2\}\right)$ for the independence testing. In addition, the monotonicity testing was also shown to have sample complexity $\tilde O(N^{1/2}/\varepsilon^4)$ \cite{BKR04}.
    
    Apart from property testing between distributions, properties of a single distribution are well studied in the literature, e.g., \cite{BDKR05,Pan03,Pan04}. An algorithm that computes the Shannon entropy using $O\left(\frac{N}{\varepsilon\log(N)}\right)$ samples for $\varepsilon = \Omega(N^{0.03})$ was proposed in \cite{VV11a,VV11b}. After that, the optimal estimator of Shannon entropy using $\Theta\left(\frac{N}{\varepsilon \log(N)}+\frac{(\log(N))^2}{\varepsilon^2}\right)$ samples was given in \cite{JVHW15} and \cite{WY16}. Also, an estimator for $\exp\left( (1-\alpha) S_\alpha^R(p) \right)$ was provided in \cite{JVHW15}, where $S_\alpha^R(p)$ is the $\alpha$-R\'{e}nyi entropy of distribution $p$.
    
    \paragraph{Quantum property testing}
    
    The emerging topic of quantum property testing (see \cite{MdW16}) studies the quantum advantage in testing classical statistical and quantum information properties.
    
    Quantum advantages in testing classical statistical properties have been extensively studied. 
    Quantum algorithms for testing properties of classical distributions was first studied in \cite{BHH11}, which gave quantum query complexity $O(N^{1/2}/\varepsilon^6)$ for the closeness testing, and $O(N^{1/3})$ for identity testing (to the uniform distribution) in $\ell^1$ distance (for constant precision $\varepsilon$). 
    Later, the quantum query complexity of the identity testing (to a known distribution) was improved to $\tilde O(N^{1/3}/\varepsilon^5)$ in \cite{CFMdW10}.
    The quantum query complexity for the closeness testing in $\ell^1$ distance was further improved to $\tilde O(N^{1/2}/\varepsilon^{2.5})$ in \cite{Mon15}, to $\tilde O(N^{1/2}/\varepsilon)$ in \cite{GL20}, and to $O(N^{1/2}/\varepsilon)$ in \cite{LWL24}.
    Recently, the quantum query complexity for computing the Shannon entropy and the R\'{e}nyi entropy was studied in \cite{LW19}; especially, an $\tilde O(N^{1/2}/\varepsilon^2)$ quantum query complexity was shown for the Shannon entropy. 
    
    There are also some quantum algorithms for testing quantum information properties not mentioned above.
    It was shown in \cite{KLL+17} that testing the orthogonality of pure quantum states requires $\Theta(1/\varepsilon)$ copies, promised that either they are orthogonal or have fidelity $\geq \varepsilon$. Recently, it was shown in \cite{Yu21} that quantum identity testing only uses $O(N^{3/2}/\varepsilon^2)$ copies with the help of random choice of independent measurements. 
    
    \subsection{Discussion} \label{subsec:discussion}
    
    In this paper, we suggest a generalized definition of block-encoding, with which we can directly manipulate subnormalized density operators and extract information from them. Based on this, we develop new quantum algorithms that compute a large class of quantum entropies and distances, which achieve a significant speedup over the best known ones in the low-rank case. 
    Several interesting problems remain open: 

    \begin{itemize}
      \item Our upper and lower bounds are far from being tight, a similar issue for computing the von Neumann entropy arose in \cite{GL20}. In the error analysis of our algorithms, the rank $r$ appears in the upper bound for the error as a multiplicative factor. This makes our algorithms unlikely to have complexity sub-polynomial in $r$. Can we find more efficient algorithms (for example, with query complexity sub-polynomial in $r$) or improve the lower bounds (to, for example, $\Omega(r)$)?
      
      \item It would be interesting to study other distance measures of quantum states, e.g., the relative von Neumann entropy \cite{NC10} (the quantum generalization of the Kullback-Leibler divergence \cite{KL51})
          \[
            S(\rho\|\sigma) = \tr\left(\rho \left(\ln(\rho) - \ln(\sigma)\right)\right).
          \]
      
      \item Can we apply the idea of manipulating quantum states to problems other than computing quantum entropies and distances?
    \end{itemize}

    \subsection{Recent Developments}

    After the work described in this paper, a series of quantum algorithms for computing quantum entropies and distances have been developed and applied in practical tasks. 
    \begin{itemize}
        \item \textbf{Von Neumann entropy}. In the ``purified quantum query access'' model, the query complexity for computing the von Neumann entropy was further analyzed in detail and shown to be $O(r \log(r) / \varepsilon^2)$ in \cite{LL23}. 
        Computing the von Neumann entropy in space-bounded quantum computation was investigated in \cite{LGLW23}, and they showed that the space-bounded version of von Neumann entropy difference is $\mathsf{BQL}$-complete.
        In the ``quantum sample access'' model, the time complexity for computing the von Neumann entropy was improved to $\tilde O(N^2)$ in \cite{WZ24}, compared to the $\tilde O(N^6)$ in \cite{AISW19}, while retaining the same (up to polylogarithmic factors) sample complexity $\tilde O(N^2)$. 
        \item \textbf{R\'enyi entropy}. In the ``purified quantum query access'' model, the query complexity for computing the $\alpha$-R\'enyi entropy of a quantum state was improved to $\tilde O(r^{\frac{1}{\alpha}}/\varepsilon^{\frac{1}{\alpha}+1})$ for $0 < \alpha < 1$ and $\tilde O(r/\varepsilon^{1+\frac{1}{\alpha}})$ for $\alpha > 1$ in \cite{LWZ22}.
        In the ``quantum sample access'' model, the time complexity for computing the $\alpha$-R\'enyi entropy was improved to $\tilde O(N^{\frac{4}{\alpha}-2})$ for $0 < \alpha < 1$ and $\tilde O(N^{4-\frac{2}{\alpha}})$ for $\alpha > 1$ in \cite{WZ24}, compared to the $\tilde O(N^{\frac{6}{\alpha}})$ for $0 < \alpha < 1$ and $\tilde O(N^{6})$ for $\alpha > 1$ in \cite{AISW19}, at the cost of larger sample complexity; they also showed sample lower bounds $\Omega(\max\{N, N^{\frac{1}{\alpha}-1}\})$ for computing the $\alpha$-R\'enyi entropy. 
        In addition, variational quantum algorithms for computing the von Neumann and R\'enyi entropies were proposed in \cite{GPSW23}.
        \item \textbf{Trace distance}. In the ``purified quantum query access'' model, the query complexity for computing the trace distance was improved to $\tilde O(r/\varepsilon^{2})$ in \cite{WZ23}, and they showed that low-rank trace distance estimation is $\mathsf{BQP}$-complete based on the result of \cite{RASW23}, improving the $\mathsf{DQC1}$-hardness given in this paper. 
        The space-bounded version of trace distance estimation was shown to be $\mathsf{BQL}$-complete in \cite{LGLW23}, and its certification was shown to be $\mathsf{coRQ_{\text{U}}L}$-complete.
        In the ``quantum sample access'' model, the sample complexity for computing the fidelity was shown to be $\tilde O(r^{2}/\varepsilon^{5})$ in \cite{WZ23}, which was later employed in a hypothesis testing based auditing pipeline for quantum differential privacy with domain knowledge \cite{NGW23}. 
        \item \textbf{Fidelity}. In the ``purified quantum query access'' model, the query complexity for computing the fidelity was improved to $\tilde O(r^{2.5}/\varepsilon^{5})$ in \cite{GP22}. 
        When quantum states are well-conditioned (i.e., $\rho, \sigma \geq I/\kappa$ for some known $\kappa > 0$), the query complexity was shown to be $\tilde O(\kappa^4/\varepsilon)$ in \cite{LWWZ24}, with the dependence on $\varepsilon$ optimal (up to polylogarithmic factors).
        It was shown in \cite{RASW23} that pure-state fidelity estimation is $\mathsf{BQP}$-complete, which, together with the polynomial-time quantum algorithms for low-rank fidelity estimation in \cite{WZC+22,GP22} and this paper, implies that low-rank fidelity estimation is also $\mathsf{BQP}$-complete, improving the $\mathsf{DQC1}$-hardness given in this paper. 
        In the ``quantum sample access'' model, the sample complexity for computing the fidelity was shown to be $\tilde O(r^{5.5}/\varepsilon^{12})$ in \cite{GP22}. 
    \end{itemize}

    \subsection{Organization of This Paper} \label{subsec:organization}
    
    Section \ref{sec:qset} introduces the idea that regards quantum states as block-encodings, and provides a series of basic techniques for manipulating  them. Section \ref{sec:entropy} presents quantum algorithms that compute quantum entropies, including the von Neumann entropy, quantum R\'{e}nyi entropy and quantum Tsallis entropy. Section \ref{sec:distance-states} presents quantum algorithms that compute the trace distance, fidelity and their extensions.

    \section{Quantum States as Block-Encodings} \label{sec:qset}
    
    Since the introduction of qubitization in Hamiltonian simulation \cite{LC19}, block-encodings have been widely used as a basic notion in quantum algorithms, e.g., \cite{CGJ19,GSLW19}. In the existing research, block-encodings are unitary operators that block-encode smaller ones.
    
    Quantum states (i.e., density operators) are often used to contain necessary information in quantum algorithms. For this purpose,
    a technique was provided in \cite{LC19} to implement a unitary operator that block-encodes a mixed quantum state. 
    However, to the best of our knowledge, there is no known method to do the inverse, that is, to prepare a mixed quantum state using queries to the given unitary operator (quantum oracle) that block-encodes its density operator. 
    As a result, it could be difficult to extract information from operators that are block-encoded in unitary operators.
    This motivate us to regard quantum states as block-encodings. As will be seen later in this section, it is convenient to extract information from the operators block-encoded in quantum states as well as to manipulate them. 
    
    In this section, we will extend the definition of block-encoding proposed for unitary operators as in \cite{LC19, GSLW19, CGJ19} to that for general operators, especially for density operators (i.e., quantum states). Then we show the possibility that information can be stored in and extracted from quantum states as block-encodings. Also, we can manipulate the information block-encoded in quantum states. Here, the ``information'' block-encoded in quantum states (i.e., density operators) is essentially subnormalized density operators. 

    \subsection{Subnormalized density operators}

    We will use the language of the conventional block-encoding. Here, we give the definition of block-encoding for ordinary quantum operators as follows.

    \begin{definition} [Block-encoding] \label{def:block-encoding}
    Suppose $A$ is an $n$-qubit operator, $\alpha, \varepsilon \geq 0$ and $a \in \mathbb{N}$. An $(n+a)$-qubit operator $B$ is said to be an $(\alpha, a, \varepsilon)$-block-encoding of $A$, if
    \[
        \Abs{ \alpha \prescript{}{a}{\bra 0} B \ket 0_a - A } \leq \varepsilon.
    \]
    \end{definition}

    Intuitively, $A$ is represented by the matrix in the upper left corner of $B$, i.e.
    \[
        B \approx \begin{bmatrix}
            A/\alpha & * \\
            * & *
        \end{bmatrix}.
    \]
    Here, we write $\ket{0}_a$ to denote $\ket{0}^{\otimes a}$, where the subscript $a$ indicates which (and how many) qubits are involved in the Dirac symbol. For example, if a system consists of two subsystems of $a$ qubits and $b$ qubits and it is in state $\ket{0}^{\otimes(a+b)}$, we can represent it as $\ket{0}_{a+b}$ or $\ket{0}_a \ket{0}_b$.

    We are interested in matrices block-encoded in a mixed quantum state (density operator), which are indeed subnormalized density operators.

    \begin{definition} [Subnormalized density operator] \label{def:subnormalized-density-operator}
        A subnormalized density operator $A$ is a semidefinite operator with $\tr(A) \leq 1$. A (normalized) density operator is a subnormalized density operator with trace $1$.
        An $(n+a+b)$-qubit unitary operator $U$ is said to prepare an $n$-qubit subnormalized density operator $A$, if it prepares the purification $\ket\rho = U \ket{0}_{n+a+b}$ of a density operator $\rho = \tr_b(\ket{\rho}\bra{\rho})$, which is a $(1, a, 0)$-block-encoding of $A$.
    \end{definition}

    Given a subnormalized density operator $A$ prepared by a unitary operator $U$, we usually need to construct another unitary operator $\tilde U$ which is a block-encoding of $A$. This technique was first introduced by \cite{LC19}, then generalized for subnormalized density operators by \cite{vAG19, GSLW19}.

    \begin{lemma} [Block-encoding of subnormalized density operators \cite{LC19, vAG19, GSLW19}] \label{lemma:block-encoding of density operators}
        Suppose $U$ is an $(n+a)$-qubit unitary operator that prepares an $n$-qubit subnormalized density operator $A$. Then there is a $(2n+a)$-qubit unitary operator $\tilde U$ which is a $(1, n+a, 0)$-block-encoding of $A$, using $1$ query to $U$ and $U^\dag$ and $O(a)$ elementary quantum gates.
    \end{lemma}

    \subsection{Generalized evolution}

    It is well known that after applying a unitary operator $U$ on a mixed quantum state $\rho$, it will become another state $U \rho U^\dag$. Here, we extend the basic unitary evolution to the case of subnormalized density operators, which transforms a subnormalized density operator $A$ to $BAB^\dag$, where $B$ is block-encoded in a unitary operator.

    \begin{lemma} [Evolution of subnormalized density operators] \label{lemma:density-basic}
        Suppose that
        \begin{enumerate}
          \item $U$ is an $(n+a)$-qubit unitary operator that prepares an $n$-qubit subnormalized density operator $A$, and
          \item $V$ is an $(n+b)$-qubit unitary operator which is a $(1, b, 0)$-block-encoding of $B$.
        \end{enumerate}
        Then, $\tilde U = (V \otimes I_a) (U \otimes I_b)$ is an $(n+a+b)$-qubit unitary operator that prepares an $n$-qubit subnormalized density operator $BAB^\dag$.
    \end{lemma}
    \begin{proof}
        Let $a = a_1 + a_2$ such that $U$ prepares an $(n+a_1)$-qubit density operator $\rho$, which is a $(1, a_1, 0)$-block-encoding of $A$. Suppose
        \[
            A = \sum_j \lambda_j \ket{u_j} \bra{u_j}.
        \]
        Then we have
        \[
            \ket{\rho}_{n+a_1+a_2} = \sum_j \sqrt{\lambda_j} \ket{u_j}_n \ket{0}_{a_1} \ket{\psi_j}_{a_2} + \ket{\perp_{a_1}}_{n+a_1+a_2},
        \]
        where $\ket{\psi_j}$ is an orthogonal basis, and $$\Abs{\prescript{}{a_1}{\braket{0}{\perp_{a_1}}_{n+a_1+a_2}}} = 0.$$
        Note that
        \begin{align*}
            \ket{\tilde \rho} & \coloneqq \tilde U \ket{0}_{n+a+b}
            = (V \otimes I_a) \ket{\rho}_{n+a_1+a_2} \ket{0}_b \\
            & = \sum_j \sqrt{\lambda_j} \ket{0}_{a_1} \ket{\psi_j}_{a_2} \left( V \ket{u_j}_n \ket{0}_b \right) + V \ket{\perp_{a_1}}_{n+a_1+a_2} \ket{0}_b.
        \end{align*}
        Let $\tilde \rho = \tr_{a_2}\left(\ket{\tilde \rho}\bra{\tilde \rho}\right)$, then $$\prescript{}{a_1+b}{\bra{0}} \tilde \rho \ket{0}_{a_1+b} = \tr_{a_2}\left( \prescript{}{a_1+b}{\braket{0}{\tilde \rho}} \braket{\tilde \rho}{0}_{a_1+b} \right),$$ where
        \begin{align*}
            \prescript{}{a_1+b}{\braket{0}{\tilde \rho}} 
            & = \sum_j \sqrt{\lambda_j} \ket{\psi_j}_{a_2} \otimes \left(\prescript{}{b}{\bra{0}} V \ket{0}_b\right)\ket{u_j}_n \\
            & = \sum_j \sqrt{\lambda_j} \ket{\psi_j}_{a_2} \otimes B \ket{u_j}_n.
        \end{align*}
        We have that
        \[
            \prescript{}{a_1+b}{\bra{0}} \tilde \rho \ket{0}_{a_1+b}
            = \sum_j \lambda_j B \ket{u_j}_n \bra{u_j} B^\dag = B A B^\dag.
        \]
    \end{proof}

    \subsection{Polynomial eigenvalue transformation}

    Now we show how a subnormalized density operator $A$ can be transformed to a new subnormalized density operator $A(P(A))^2$, where $P(x)$ is a polynomial. To this end, we recall the polynomial eigenvalue transformation of unitary operators in \cite{GSLW19}, and extend it to the case of preparing subnormalized density operators.

    \begin{theorem} [Polynomial eigenvalue transformation of unitary operators \cite{GSLW19}] \label{thm:qsvt}
        Suppose that
        \begin{enumerate}
          \item $U$ is an $(n+a)$-qubit unitary operator, which is a $(1, a, 0)$-block-encoding of an Hermitian operator $A$.
          \item $P \in \mathbb{R}[x]$~\footnote{Let $\mathbb{\tilde R} \subseteq \mathbb{R}$ be the set of polynomial-time computable real numbers. That is, for every real number $x \in \mathbb{\tilde R}$, there is a polynomial-time (classical) Turing machine $M$ such that $\abs{M(1^n) - x} < 2^{-n}$, where $M(1^n)$ denotes the output floating point real number of $M$ on input $1^n$. Throughout this paper, we only consider polynomial-time computable real numbers, and for any $S \subseteq \mathbb{R}$, we write $S$ to denote $S \cap \mathbb{\tilde R}$ for convenience. Especially, we just write $\mathbb{R}$ for $\mathbb{\tilde R}$.} is a degree-$d$ polynomial such that $\Abs{P}_{[-1, 1]} \leq \frac 1 2$.\footnote{For a function $f: \mathbb{R} \to \mathbb{C}$ and a set $I \subseteq \mathbb{R}$, we define $\Abs{f}_{I} = \sup\left\{\abs{f(x)}\middle|x \in I\right\}.$} Moreover, if $P$ is even or odd, the condition can be relaxed to $\Abs{P}_{[-1, 1]} \leq 1$.
        \end{enumerate}
        Then for every $\delta > 0$, there is a quantum circuit\footnote{Throughout this paper, without explicit explanation, quantum circuits are uniform. Here, a uniform quantum circuit is a family of quantum circuits whose descriptions can be computed by a polynomial-time (classical) Turing machine.} $\tilde U$ such that
        \begin{enumerate}
          \item $\tilde U$ is a $(1, a+2, \delta)$-block-encoding of $P(A)$.
          \item $\tilde U$ uses $d$ queries to $U$ and $U^\dag$, $1$ query to controlled-$U$ and $O((a+1)d)$ elementary quantum gates.
          \item A description of $\tilde U$ can be computed by a (classical) Turing machine in $O(\poly(d, \log(1/\delta)))$ time.
        \end{enumerate}
    \end{theorem}

    In the following, combining Theorem \ref{thm:qsvt} and Lemma \ref{lemma:density-basic}, we develop a technique of polynomial eigenvalue transformation of subnormalized density operators.

    \begin{theorem} [Polynomial eigenvalue transformation of subnormalized density operators] \label{lemma:technique}
        Suppose that
        \begin{enumerate}
          \item $U$ is an $(n+a)$-qubit unitary operator that prepares an $n$-qubit subnormalized density operator $A$.
          \item $P \in \mathbb{R}[x]$ is a degree-$d$ polynomial such that $\Abs{P}_{[-1, 1]} \leq \frac 1 2$. Moreover, if $P$ is even or odd, the condition can be relaxed to $\Abs{P}_{[-1, 1]} \leq 1$.
        \end{enumerate}
        Then for every $\delta \in (0, 1)$, there is a quantum circuit $\tilde U$ such that
        \begin{enumerate}
          \item $\tilde U$ prepares an $n$-qubit subnormalized density operator $B$, and $B$ is a $(1, 0, \delta)$-block-encoding of $A(P(A))^2$.
          \item $\tilde U$ uses $O(d)$ queries to $U$ and $U^\dag$, $1$ query to controlled-$U$ and controlled-$U^\dag$, and $O((n+a)d)$ elementary quantum gates.
          \item A description of $\tilde U$ can be computed by a (classical) Turing machine in $O(\poly(d, \log(1/\delta)))$ time.
        \end{enumerate}
    \end{theorem}
    \begin{proof}
        By
        Lemma \ref{lemma:block-encoding of density operators}, there is a quantum circuit $V$ that is a $(1, O(n+a), 0)$-block-encoding of $A$, which consists of $1$ query to $U$ and $U^\dag$ and $O(n+a)$ elementary quantum gates.
        Then by
        Theorem \ref{thm:qsvt}, there is a quantum circuit $\tilde V$, which is a $(1, b, \delta)$-block-encoding of $P(A)$, with $d$ queries to $V$ and $V^\dag$, $1$ query to controlled-$V$, and $O((n+a)d)$ elementary quantum gates, where $b = O(n+a)$.

        We claim that $\tilde U = (\tilde V \otimes I_a) (U \otimes I_b)$ is desired. To see this, by Lemma \ref{lemma:density-basic}, $\tilde U$ prepares an $n$-qubit subnormalized density operator $\left(\prescript{}{b}{\bra{0}} \tilde V \ket{0}_b\right) A \left(\prescript{}{b}{\bra{0}} \tilde V \ket{0}_b\right)^\dag$. On the other hand,
        \begin{align*}
            & \Abs{ \left(\prescript{}{b}{\bra{0}} \tilde V \ket{0}_b\right) A \left(\prescript{}{b}{\bra{0}} \tilde V \ket{0}_b\right)^\dag - A(P(A))^2 } \\
            \leq & \Abs{\left(\prescript{}{b}{\bra{0}} \tilde V \ket{0}_b - P(A)\right) A \left(\prescript{}{b}{\bra{0}} \tilde V \ket{0}_b\right)^\dag} + \Abs{P(A) A \left(\left(\prescript{}{b}{\bra{0}} \tilde V \ket{0}_b\right)^\dag - P(A)\right)} \\
            \leq & \delta \Abs{A} \left(\Abs{P(A)} + \delta\right) + \Abs{P(A)} \Abs{A} \delta \\
            \leq & \frac 5 2 \delta = \Theta(\delta).
        \end{align*}
        We conclude that $\tilde U$ prepares the purification $\ket{\tilde \rho}$ of $\tilde \rho$, which is a $(1, O(n+a), \Theta(\delta))$-block-encoding of $A(P(A))^2$, which yields the proof.
    \end{proof}

    By Theorem \ref{lemma:technique}, we are able to transform a subnormalized density operator $A$ to $A(f(A))^2$ for a large range of $f(x)$, provided $f(x)$ can be efficiently approximated by a polynomial.

    \begin{theorem} [Eigenvalue transformation of subnormalized density operators] \label{thm:main}
        Suppose that
        \begin{enumerate}
          \item $U$ is an $(n+a)$-qubit unitary operator that prepares an $n$-qubit subnormalized density operator $A$.
          \item $f: [-1, 1] \to \mathbb{R}$ can be approximated by a degree-$d$ polynomial $P \in \mathbb{R}[x]$ such that there are two parameters $\delta, \varepsilon \in (0, \frac 1 2]$, it holds that $\Abs{P(x) - f(x)}_{[\delta, 1]} \leq \varepsilon$ and $\Abs{P}_{[-1, 1]} \leq \frac 1 2$. Moreover, if $P$ is even or odd, the latter condition can be relaxed to $\Abs{P}_{[-1, 1]} \leq 1$.
        \end{enumerate}
        Then there is a quantum circuit $\tilde U$ such that
        \begin{enumerate}
          \item $\tilde U$ prepares an $n$-qubit subnormalized density operator $B$, and $B$ is a $\left(1, 0, \Theta\left(\varepsilon + \delta + \Abs{x(f(x))^2}_{[0, \delta]}\right)\right)$-block-encoding of $A(f(A))^2$.
          \item $\tilde U$ uses $O(d)$ queries to $U$ and $U^\dag$, $1$ query to controlled-$U$ and controlled-$U^\dag$, and $O((n+a)d)$ elementary quantum gates.
          \item A description of $\tilde U$ can be computed by a (classical) Turing machine in $O(\poly(d, \log(1/\delta)))$ time.
        \end{enumerate}
    \end{theorem}
    \begin{proof}
        Let $\tilde U$ be the quantum circuit obtained by Theorem \ref{lemma:technique}, which prepares an $n$-qubit subnormalized density operator $B$ as a $(1, 0, \delta)$-block-encoding of $A(P(A))^2$. We analyze the error by the fact that
        \[
            \Abs{x(P(x))^2 - x(f(x))^2}_{[0, 1]} \leq \Theta\left(\varepsilon + \delta + \Abs{x(f(x))^2}_{[0, \delta]}\right).
        \]
        We consider two cases.

            \textbf{Case 1}. $x \in [\delta, 1]$.
            \begin{align*}
                \abs{x(P(x))^2 - x(f(x))^2} 
                & \leq \abs{x} \abs{P(x)+f(x)} \abs{P(x)-f(x)} \\
                & \leq (1 + \varepsilon) \varepsilon \leq 2 \varepsilon = \Theta(\varepsilon).
            \end{align*}

            \textbf{Case 2}. $x \in [0, \delta)$.
            \begin{align*}
                \abs{x(P(x))^2 - x(f(x))^2} 
                & \leq \abs{x(P(x))^2} + \abs{x(f(x))^2} \\
                & \leq \delta + \Abs{x(f(x))^2}_{[0, \delta]}.
            \end{align*}
        Then we have
        \begin{align*}
            \Abs{B - A(f(A))^2}
            & \leq \Abs{B - A(P(A))^2} + \Abs{A(P(A))^2 - A(f(A))^2} \\
            & \leq \delta + \Theta\left(\varepsilon + \delta + \Abs{x(f(x))^2}_{[0, \delta]}\right) \\
            & = \Theta\left(\varepsilon + \delta + \Abs{x(f(x))^2}_{[0, \delta]}\right).
        \end{align*}
    \end{proof}

    As will be seen, Theorem \ref{thm:main} can be used to develop a technique of preparing positive powers of Hermitian operators (see Section \ref{sec:positive powers of Hermitian matrices}).

    \subsection{Positive powers} \label{sec:positive powers of Hermitian matrices}

    We will use Theorem \ref{thm:main} to develop an efficient approach for implementing positive powers of Hermitian matrix $A$, which removes the dependence on $\kappa$ that $I/\kappa \leq A \leq I$ as in \cite{CGJ19} (see Lemma 10 of its full version). 
    In this subsection, we will provide the positive power tricks for an Hermitian operator block-encoded in a density operator (Lemma \ref{lemma:positive-power-density}) and in a unitary operator (Lemma \ref{lemma:positive-power-unitary}). To this end, we recall some results of polynomial approximations of power functions.

    \begin{lemma} 
    [Polynomial approximation of positive power functions \cite{Gil19, CGJ19}]
    \label{lemma:positive-power}
        Let $\delta, \varepsilon \in (0, \frac 1 2]$, $c \in (0, 1)$ and $f(x) = \frac 1 2 x^c$. Then there is an even/odd degree-$O\left(\frac{1}{\delta}\log\left(\frac{1}{\varepsilon}\right)\right)$ polynomial\footnote{In this paper, a degree-$d$ polynomial means a uniform family of (polynomial-time computable) polynomials with parameter $d$. That is, suppose $P_d(x) = \sum_{k = 0}^d a_k x^k$, where $a_k$ is a polynomial-time computable real number for $0 \leq k \leq d$, then there is a polynomial-time classical Turing machine that, on input $1^d$, outputs descriptions $M_k$ of $a_k$ for $0 \leq k \leq d$, where $M_k$ is a polynomial-time classical Turing machine that, on input $1^n$, output $a_k$ within additive error $2^{-n}$.} $P \in \mathbb{R}[x]$ such that $\Abs{P(x) - f(x)}_{[\delta, 1]} \leq \varepsilon$ and $\Abs{P}_{[-1, 1]} \leq 1$.
    \end{lemma}

    \begin{lemma} 
    [Polynomial approximation of negative power functions \cite{Gil19, CGJ19, GSLW19}]
    \label{lemma:negative-power}
        Let $\delta, \varepsilon \in (0, \frac 1 2]$, $c > 0$ and $f(x) = \frac {\delta^c} 2 x^{-c}$. Then there is an even/odd degree-$O\left(\frac{c+1}{\delta}\log\left(\frac{1}{\varepsilon}\right)\right)$ polynomial $P \in \mathbb{R}[x]$ such that $\Abs{P(x) - f(x)}_{[\delta, 1]} \leq \varepsilon$ and $\Abs{P}_{[-1, 1]} \leq 1$.
    \end{lemma}

    First, we develop a method of implementing positive powers of Hermitian matrix $A$, which is given as block-encoded in a density operator $\rho$. Here, the purification of $\rho$ can be prepared by a unitary operator $U$.

    \begin{lemma} [Positive powers block-encoded in density operators] \label{lemma:positive-power-density}
        Suppose that
        \begin{enumerate}
          \item $U$ is an $(n+a)$-qubit unitary operator that prepares an $n$-qubit subnormalized density operator $A$.
          \item $\delta, \varepsilon \in \left(0, \frac 1 2\right]$ and $c \in (0, 1)$.
        \end{enumerate}
        Then there is a quantum circuit $\tilde U$ such that
        \begin{enumerate}
          \item $\tilde U$ prepares an $n$-qubit subnormalized density operator $B$, and $B$ is a $(4\delta^{c-1}, 0, \Theta(\delta^{c} + \varepsilon \delta^{c-1}))$-block-encoding of $A^c$.
          \item $\tilde U$ uses $O(d)$ queries to $U$ and $U^\dag$, $1$ query to controlled-$U$ and controlled-$U^\dag$, and $O((n+a)d)$ elementary quantum gates, where $d = O\left(\frac{1}{\delta} \log\left(\frac{1}{\varepsilon}\right)\right)$.
          \item A description of $\tilde U$ can be computed by a (classical) Turing machine in $O(\poly(d))$ time.
        \end{enumerate}
    \end{lemma}
    \begin{proof}
        Let $f(x) = \frac{\delta^c}{2}x^{-c}$, where $c \in \left(0, \frac 1 2\right)$. Then $x(f(x))^2 = \frac{\delta^{2c}}{4} x^{1-2c}$, and $\Abs{x(f(x))^2}_{[0, \delta]} \leq \frac{\delta}{4} = \Theta(\delta)$.
        By
        Lemma \ref{lemma:negative-power}, we can obtain an even/odd polynomial $P(x)$ of degree $O\left(\frac{1}{\delta} \log\left(\frac{1}{\varepsilon}\right)\right)$ such that $\Abs{P-f}_{[\delta, 1]} \leq \varepsilon$ and $\Abs{P}_{[-1, 1]} \leq 1$. By Theorem \ref{thm:main}, there is a quantum circuit $\tilde U$, which prepares an $n$-qubit subnormalized density operator $B$, which is a $\left(1, 0, \Theta(\delta+\varepsilon) \right)$-block-encoding of $A\left(f(A)\right)^2 = \frac{\delta^{2c}}{4} A^{1-2c}$. In other words, $B$ is a $(4\delta^{-2c}, 0, \Theta(\delta^{1-2c} + \varepsilon \delta^{-2c}))$-block-encoding of $A^{1-2c}$. These yield the proof by letting $c' = 1-2c$.
    \end{proof}
    
    \begin{remark}
        In Lemma \ref{lemma:positive-power-density}, we provide a positive power trick for density operators, whose error depends on two parameters $\delta$ and $\varepsilon$. This allows us to prepare positive powers of density operators flexibly, without dealing with their ``condition numbers'' $\kappa$ such that $\Pi/\kappa \leq \rho$ for some projector $\Pi$. The method used in Lemma \ref{lemma:positive-power-density} is a straightforward application of Theorem \ref{thm:main} with the observation that $x(f(x))^2 \to 0$ when $x \to 0$. In fact, any function $f(x)$ that satisfies this condition (and, of course, other conditions required in Theorem \ref{thm:main}) is applicable in this trick, without dealing with  $\kappa$. As will be seen, this technique for positive powers of subnormalized density operators will be frequently used in our quantum algorithms
        (see Section \ref{sec:entropy} and Section \ref{sec:distance-states}), in order to avoid the $\kappa$ restrictions on density operators.
    \end{remark}

    Next, inspired by the above observation, we extend the result to the case that $A$ is given as block-encoded in a unitary operator. Here, we need to implement a threshold projector. The following lemma is Corollary 16 of \cite{GSLW19}.

    \begin{lemma} [Polynomial approximation of threshold projectors \cite{GSLW19}] \label{lemma:threshold-function-main}
        Let $\delta, \varepsilon \in (0, \frac 1 2)$ and $t \in [0, 1]$ such that $0 < t - \delta < t + \delta < 1$. There is an even polynomial $P \in \mathbb{R}[x]$ of degree $O\left(\frac{1}{\delta}\log\left(\frac{1}{\varepsilon}\right)\right)$ such that
        \begin{enumerate}
          \item $\Abs{P}_{[-1, 1]} \leq 1$,
          \item $P(x) \in [1-\varepsilon, 1]$ for $x \in [-t+\delta,t-\delta]$, and
          \item $P(x) \in [0, \varepsilon]$ for $x \in [-1, -t-\delta] \cup [t+\delta, 1]$.
        \end{enumerate}
    \end{lemma}

    We take some special cases of Lemma \ref{lemma:threshold-function-main} as follows, which will be often used to design our quantum algorithms.

    \begin{corollary} \label{lemma:threshold-function}
        Let $\delta, \varepsilon \in (0, \frac 1 4]$. Then there is an even degree-$O\left(\frac{1}{\delta}\log\left(\frac{1}{\varepsilon}\right)\right)$ polynomial $R \in \mathbb{R}[x]$ such that $\Abs{R}_{[-1, 1]} \leq 1$ and
        \[
            R(x) \in \begin{cases}
                [1-\varepsilon, 1] & x \in [-1, -2\delta] \cup [2\delta, 1] \\
                [0, \varepsilon] & x \in [-\delta, \delta]
            \end{cases}.
        \]
    \end{corollary}

    \begin{corollary} \label{lemma:threshold-function-2}
        Let $\delta, \varepsilon \in (0, \frac 1 4]$. Then there is an even degree-$O\left(\frac{1}{\delta}\log\left(\frac{1}{\varepsilon}\right)\right)$ polynomial $R \in \mathbb{R}[x]$ such that $\Abs{R}_{[-1, 1]} \leq 1$ and
        \[
            R(x) \in \begin{cases}
                [1-\varepsilon, 1] & x \in [-1+2\delta, 1-2\delta] \\
                [0, \varepsilon] & x \in [-1, -1+\delta] \cup [1-\delta, 1]
            \end{cases}.
        \]
    \end{corollary}

    We also need the following lemma to multiply block-encoded matrices.
    \begin{lemma}
    [Product of block-encoded matrices \cite{GSLW19}]
    \label{lemma:product-block}
    Suppose that
    \begin{enumerate}
      \item $U$ is an $(n+a)$-qubit unitary operator that is a $(\alpha, a, \delta)$-block-encoding of an $n$-qubit operator $A$.
      \item $V$ is an $(n+b)$-qubit unitary operator that is a $(\beta, b, \varepsilon)$-block-encoding of an $n$-qubit operator $B$.
    \end{enumerate}
    Then there is a quantum circuit $\tilde U$ such that
    \begin{enumerate}
      \item $\tilde U$ is an $(\alpha\beta, a+b, \alpha\varepsilon + \beta\delta)$-block-encoding of $AB$.
      \item $\tilde U$ uses $1$ query to each of $U$ and $V$.
    \end{enumerate}
    \end{lemma}

    With the approximation of threshold functions, we are able to implement positive powers of Hermitian matrix $A$, which is given as block-encoded in a unitary operator.

    \begin{lemma} [Positive powers block-encoded in unitary operators] \label{lemma:positive-power-unitary}
        Suppose that
        \begin{enumerate}
          \item $U$ is an $(n+a)$-qubit unitary operator which is a $(1, a, 0)$-block-encoding of an $n$-qubit Hermitian operator $A$.
          \item $\delta, \varepsilon \in \left(0, \frac 1 4\right]$ and $c \in (0, 1)$.
        \end{enumerate}
        Then there is a quantum circuit $W$ such that
        \begin{enumerate}
          \item $W$ is a $(2, O(a+1), \Theta(\varepsilon + \delta^c))$-block-encoding of $\abs{A}^c$.
          \item $W$ uses $O(Q)$ queries to $U$ and $U^\dag$, $O(1)$ query to controlled-$U$ and $O((a+1)Q)$ elementary quantum gates, where $Q = O\left(\frac{1}{\delta}\log\left(\frac{1}{\varepsilon}\right)\right)$.
          \item A description of $W$ can be computed by a (classical) Turing machine in $\poly(Q)$ time.
        \end{enumerate}
    \end{lemma}
    \begin{proof}
        Let $f(x) = \frac 1 2 \abs{x}^c$ be an even function and $\delta, \varepsilon \in (0, \frac 1 4]$. By
        Lemma \ref{lemma:positive-power}, there is an even degree-$d_P$ polynomial $P \in \mathbb{R}[x]$, where $d_P = O\left(\frac{1}{\delta}\log\left(\frac{1}{\varepsilon}\right)\right)$, such that $\Abs{P(x) - f(x)}_{[\delta, 1]} \leq \varepsilon$, $\Abs{P(x) - f(x)}_{[-1, -\delta]} \leq \varepsilon$ and $\Abs{P}_{[-1, 1]} \leq 1$. By Theorem \ref{thm:qsvt}, for $\delta_U > 0$, there is a quantum circuit $\tilde U$ such that
        \begin{enumerate}
          \item $\tilde U$ is a $(1, a+2, 0)$-block-encoding of $B$, and $B$ is a $(1, 0, \delta_U)$-block-encoding of $P(A)$.
          \item $\tilde U$ uses $d_P$ queries to $U$ and $U^\dag$, $1$ query to controlled-$U$ and $O((a+1)d_P)$ elementary quantum gates.
          \item A description of $\tilde U$ can be computed by a (classical) Turing machine in $\poly(d_P, \log(1/\delta_U))$ time.
        \end{enumerate}

        Let $R \in \mathbb{R}[x]$ be the even degree-$d_R$ polynomial in Corollary \ref{lemma:threshold-function}, where $d_R = O\left(\frac{1}{\delta}\log\left(\frac{1}{\varepsilon}\right)\right)$, such that $\Abs{R}_{[-1, 1]} \leq 1$ and
        \[
            R(x) \in \begin{cases}
                [1-\varepsilon, 1] & x \in [-1, -2\delta] \cup [2\delta, 1] \\
                [0, \varepsilon] & x \in [-\delta, \delta]
            \end{cases}.
        \]
        By Theorem \ref{thm:qsvt}, for $\delta_V > 0$, there is a quantum circuit $\tilde V$ such that
        \begin{enumerate}
          \item $\tilde V$ is a $(1, a+2, 0)$-block-encoding of $C$, and $C$ is a $(1, 0, \delta_V)$-block-encoding of $R(A)$.
          \item $\tilde V$ uses $d_R$ queries to $U$ and $U^\dag$, $1$ query to controlled-$U$ and $O((a+1)d_R)$ elementary quantum gates.
          \item A description of $\tilde V$ can be computed by a (classical) Turing machine in $\poly(d_R, \log(1/\delta_V))$ time.
        \end{enumerate}

        By Lemma \ref{lemma:product-block}, using one query to each of $\tilde U$ and $\tilde V$, we can obtain a quantum circuit $W$, which is a $(1, 2a+4, 0)$-block-encoding of $B C$. In the following, we will show that $\Abs{BC - f(A)} \leq \Theta(\varepsilon + \delta^c)$.
        We note that
        \[
            \Abs{P(x)R(x) - f(x)}_{[0, 1]} \leq \Theta(\varepsilon + \delta^c).
        \]
        This is seen by the following three cases:
        \begin{enumerate}
          \item $\abs{x} > 2\delta$. We have $\abs{P(x)R(x) - f(x)} \leq \abs{(P(x) - f(x))R(x)} + \abs{f(x)(R(x) - 1)} \leq \Theta(\varepsilon)$.
          \item $\abs{x} < \delta$. We have $\abs{P(x)R(x) - f(x)} \leq \abs{P(x)}\abs{R(x)} + \abs{f(x)} \leq \Theta(\varepsilon + \delta^c)$.
          \item $\delta \leq \abs{x} \leq 2\delta$. We have $\abs{P(x)R(x) - f(x)} \leq \abs{(P(x)-f(x))R(x)} + \abs{f(x)} \abs{R(x)-1} \leq \Theta(\varepsilon + \delta^c)$.
        \end{enumerate}
        Note that $A$ is Hermitian, we have $\Abs{P(A)R(A) - f(A)} \leq \Theta(\varepsilon + \delta^c)$. Finally,
        \begin{align*}
            \Abs{BC - f(A)} 
            & \leq \Abs{BC - P(A)R(A)} + \Abs{P(A)R(A) - f(A)} \\
            & \leq \Theta(\delta_U + \delta_V + \varepsilon + \delta^c).
        \end{align*}
        These conclude that $W$ is a $(1, 2a+4, \Theta(\varepsilon + \delta^c))$-block-encoding of $f(A)$ by setting $\delta_U = \delta_V = \varepsilon$, and therefore a $(2, O(n+a), \Theta(\varepsilon + \delta^c))$-block-encoding of $\abs{A}^c$.
    \end{proof}

    By setting $\kappa = 1/ \delta$, Lemma \ref{lemma:positive-power-unitary} reproduces the result of \cite{CGJ19} (see Lemma 10 of its full version). The strength of Lemma \ref{lemma:positive-power-unitary} is to allow implement positive powers of $A$, regardless of whether the condition $I/\kappa \leq A \leq I$ holds. 
    The technique used in Lemma \ref{lemma:positive-power-unitary} by multiplying the polynomial approximation with a threshold function is similar to Corollary 42 of the full version of \cite{GSLW19} for implementing the threshold pseudoinverse.

    \subsection{Trace estimation}

    In this subsection, we will provide a method of estimating the trace of an Hermitian matrix which is block-encoded in a density operator. Before that, we recall the quantum amplitude estimation \cite{BHMT02}.

    \begin{theorem} [Quantum amplitude estimation, \cite{BHMT02}] \label{thm:amplitude-estimation}
        Suppose $U$ is an $(a+b)$-qubit unitary operator such that
        \[
            U \ket{0}_{a+b} = \sqrt{p} \ket{0}_a \ket{\phi_0}_b + \sqrt{1 - p} \ket{1}_a \ket{\phi_1}_b,
        \]
        where $\ket{\phi_0}$ and $\ket{\phi_1}$ are normalized (pure) quantum states and $p \in [0, 1]$. There is a quantum algorithm that outputs $\tilde p \in [0, 1]$ such that
        \[
            \abs{\tilde p - p} \leq \frac{2\pi\sqrt{p(1-p)}}{M} + \frac{\pi^2}{M^2}
        \]
        with probability $\geq \frac{8}{\pi^2}$, using $O(M)$ queries to $U$ and $U^\dag$.
    \end{theorem}
    If we know an upper bound $B$ of $p$, then we can take $M = \ceil{2\pi \left( \frac{2\sqrt{B}}{\varepsilon} + \frac{1}{\sqrt{\varepsilon}} \right)} = \Theta\left( \frac{\sqrt{B}}{\varepsilon} + \frac{1}{\sqrt{\varepsilon}} \right)$ to guarantee that $\abs{\tilde p - p} \leq \varepsilon$.

    Based on quantum amplitude estimation, we develop the trace estimation of subnormalized density operators as shown below, which will be useful to design quantum algorithms, see Section \ref{sec:entropy} and Section \ref{sec:distance-states}.

    \begin{lemma} [Trace estimation of subnormalized density operators] \label{lemma:trace-estimation}
        Suppose $U$ is an $(n+a)$-qubit unitary operator that prepares an $n$-qubit subnormalized density operator $A$, and $B > 0$ is a known constant that $\tr(A) \leq B$. For every $\varepsilon > 0$, there is a quantum algorithm that estimates $\tr(A)$ within additive error $\varepsilon$ with $O\left( \frac{\sqrt{B}}{\varepsilon} + \frac{1}{\sqrt{\varepsilon}} \right)$ queries to $U$ and $U^\dag$.
    \end{lemma}
    \begin{proof}
        Let $a = a_1 + a_2$ such that $U$ prepares an $(n+a_1)$-qubit density operator $\rho$, which is a $(1, a_1, 0)$-block-encoding of $A$. Suppose
        \[
            A = \sum_j \lambda_j \ket{u_j} \bra{u_j}.
        \]
        Then we have
        \begin{align*}
            U \ket{0}_{n+a_1+a_2} & = \ket{\rho}_{n+a_1+a_2} \\
            &
            = \sum_j \sqrt{\lambda_j} \ket{u_j}_n \ket{0}_{a_1} \ket{\psi_j}_{a_2} + \ket{\perp_{a_1}}_{n+a_1+a_2},
        \end{align*}
        where $\ket{\psi_j}$ is an orthogonal basis, and $$\Abs{\prescript{}{a_1}{\braket{0}{\perp_{a_1}}_{n+a_1+a_2}}} = 0.$$ Moreover, we have
        \[
            U \ket{0}_{n+a_1+a_2} = \sqrt{p} \ket{0}_{a_1} \ket{\phi_0}_{n+a_2} + \sqrt{1-p} \ket{1}_{a_1} \ket{\phi_1}_{n+a_2}
        \]
        for some (pure) quantum states $\ket{\phi_0}$ and $\ket{\phi_1}$, where $p = \tr(A)$.

        If we know an upper bound $B$ of $\tr(A)$, then let $M = \ceil{2\pi \left( \frac{2\sqrt{B}}{\varepsilon} + \frac{1}{\sqrt{\varepsilon}} \right)} = \Theta\left( \frac{\sqrt{B}}{\varepsilon} + \frac{1}{\sqrt{\varepsilon}} \right)$, and by Theorem \ref{thm:amplitude-estimation}, we can computes $\tilde p$ such that
        \[
            \abs{\tilde p - p} \leq \frac{2\pi\sqrt{p(1-p)}}{M} + \frac{\pi^2}{M^2} \leq \varepsilon
        \]
        with success probability $\geq \frac{8}{\pi^2}$, using $M$ queries to $U$ and $U^\dag$.
    \end{proof}

    \subsection{Linear combinations}

    We will provide a technique (Lemma \ref{lemma:linear-combination}) that prepares a linear combination of subnormalized density operators, which is a natural analog of Linear-Combination-of-Unitaries (LCU) algorithm through a series of work \cite{SOGKL02,CW12,Kothari14,BCCKS15,BCK15,CKS17,GSLW19}.

    Before stating our linear combination result of subnormalized density operators, we introduce a technique that embeds a density operator in a larger space.

    \begin{lemma} \label{lemma:prepare-density-operator-with-extra-qubits}
        Suppose $U$ is an $(n+a)$-qubit unitary operator that prepares an $n$-qubit density operator $\rho$. For every $b \geq 0$, there is an $(n+a+b)$-qubit unitary operator $U^{(b)} = U \otimes I_b$ such that
        \begin{enumerate}
          \item $U^{(b)}$ prepares an $(n+b)$-qubit density operator $\rho^{(b)}$, and $\rho^{(b)}$ is a $(1, b, 0)$-block-encoding of $\rho$.
          \item $U^{(b)}$ uses $1$ query to $U$.
        \end{enumerate}
    \end{lemma}
    \begin{proof}
        Let $\ket{\psi}_{n+a} = U\ket{0}_{n+a}$ and then $\rho_{n} = \tr_a(\ket\psi_{n+a}\bra\psi)$. Let $U^{(b)} = U \otimes I_b$ and $\ket{\psi^{(b)}}_{n+a+b} = U^{(b)} \ket{0}_{n+a+b} = \ket{\psi}_{n+a} \ket{0}_b$. We have
        \begin{align*}
            \rho^{(b)} & = \tr_{a}\left(\ket{\psi^{(b)}}_{n+a+b}\bra{\psi^{(b)}}\right) \\
            & = \tr_{a}\left(\ket{\psi}_{n+a}\bra{\psi} \otimes \ket0_b\bra0\right) \\
            & = \rho_n \otimes \ket{0}_b\bra{0}.
        \end{align*}
        The proof is completed by noting that $\prescript{}{b}{\bra 0} \rho^{(b)} \ket 0_b = \rho_n$.
    \end{proof}

    Now we are ready to show the technique to prepare a linear combination of subnormalized density operators. The basic idea is to prepare a linear combination of (normalized) density operators, but a careful qubit alignment is needed with the help of Lemma \ref{lemma:prepare-density-operator-with-extra-qubits}.

    \begin{lemma} [Linear combination of subnormalized density operators] \label{lemma:linear-combination}
        Suppose
        \begin{enumerate}
          \item $V$ is an $m$-qubit unitary operator such that $V \ket{0} = \sum_{k \in [2^m]} \sqrt{\alpha_k} \ket{k}$.\footnote{Here, we use the notation $[n] = \{0, 1, 2, \dots, n - 1\}$.}
          \item For every $k \in [2^m]$, $U_k$ is an $(n+a_k+b_k)$-qubit unitary operator that prepares an $(n+a_k)$-qubit density operator $\rho_k$, and $\rho_k$ is a $(1, a_k, 0)$-block-encoding of an $n$-qubit subnormalized density operator $A_k$.
        \end{enumerate}
        Let $a = \max_{k \in [2^m]}\{a_k\}$ and $b = \max_{k \in [2^m]}\{b_k\}$.
        Then there is an $(m+n+a+b)$-qubit unitary operator $U$ such that
        \begin{enumerate}
          \item $U\left(V \otimes I_{n+a+b}\right)$ prepares an $n$-qubit subnormalized density operator
            \[
                A = \sum_{k\in[2^m]} \alpha_k A_k.
            \]
          \item $U$ uses $1$ query to each $U_k$ for $k \in [2^m]$.
        \end{enumerate}
    \end{lemma}
    \begin{proof}
        Let $a_k' = a-a_k$, $b_k' = b - b_k$, and
        \[
            U = \sum_{k \in [2^m]} \ket{k}\bra{k} \otimes U_k^{\left(a_k'+b_k'\right)}
        \]
        be an $(m+n+a+b)$-qubit unitary operator, where $U^{(b)} = U \otimes I_b$ is defined as in Lemma \ref{lemma:prepare-density-operator-with-extra-qubits}. Here, $U_k^{\left(a_k'+b_k'\right)}$ acts on $n+a+b$ qubits. To be precise, if we split the $n+a+b$ qubits into three parts: (i) $n$ qubits, (ii) $a$ qubits and (iii) $b$ qubits, then $U_k^{\left(a_k'+b_k'\right)}$ uses $1$ query to $U_k$ which acts on: (i) the whole $n$ qubits, (ii) the first $a_k$ qubits and (iii) the first $b_k$ qubits.

        Let $\ket{\psi_k}_{n+a_k+b_k} = U_k \ket{0}_{n+a_k+b_k}$. Then we note that
        \begin{align*}
            \ket{\psi}
            & = U\left(V \otimes I_{n+a+b}\right) \ket{0}_m \ket{0}_n \ket{0}_a \ket{0}_b \\
            & = U \sum_{k \in [2^m]} \sqrt{\alpha_k} \ket{k}_m \ket{0}_n \ket{0}_a \ket{0}_b \\
            & = \sum_{k \in [2^m]} \sqrt{\alpha_k} \ket{k}_m \left(U_k^{\left(a_k'+b_k'\right)} \ket{0}_n \ket{0}_a \ket{0}_b\right) \\
            & = \sum_{k \in [2^m]} \sqrt{\alpha_k} \ket{k}_m \ket{\psi_k}_{n+a_k+b_k} \ket{0}_{a_k'} \ket{0}_{b_k'},
        \end{align*}
        and
        \begin{align*}
            \rho 
            & = \tr_{m+b}\left(\ket\psi\bra\psi\right) \\
            & = \sum_{k \in [2^m]} \alpha_k \tr_{b}\left( \ket{\psi_k}_{n+a_k+b_k}\bra{\psi_k} \otimes \ket{0}_{a_k'+b_k'}\bra{0} \right) \\
            & = \sum_{k \in [2^m]} \alpha_k \tr_{b_k}\left( \ket{\psi_k}_{n+a_k+b_k}\bra{\psi_k} \right) \otimes \ket{0}_{a_k'}\bra{0} \\
            & = \sum_{k \in [2^m]} \alpha_k (\rho_k)_{n+a_k} \otimes \ket{0}_{a_k'}\bra{0} \\
            & = \sum_{k \in [2^m]} \alpha_k \left(\rho_k^{(a_k')}\right)_{n+a},
        \end{align*}
        where $\rho^{(b)} = \rho \otimes \ket{0}_b\bra{0}$ is defined as in Lemma \ref{lemma:prepare-density-operator-with-extra-qubits}. To see that $\rho$ is a $(1, a, 0)$-block-encoding of $A = \sum_{k \in [2^m]} A_k$, we note that
        \begin{align*}
            \prescript{}{a}{\bra 0} \rho \ket 0_a & = \sum_{k \in [2^m]} \alpha_k \prescript{}{a_k}{\bra 0} (\rho_k)_{n+a_k} \ket 0_{a_k} \\
            & = \sum_{k \in [2^m]} \alpha_k A_k = A.
        \end{align*}
        Therefore, $U\left(V \otimes I_{a} \otimes I_b\right)$ prepares a subnormalized density operator $A$.
    \end{proof}

    As it will be also used to design quantum algorithms in this paper, we provide the LCU algorithm for comparison. Here, we use the notion as in the full version of \cite{GSLW19}.

    \begin{definition} [State preparation pair] \label{def:state-preparation-pair}
        Let $y \in \mathbb{C}^m$ with $\Abs{y}_1 \leq \beta$, and $\varepsilon \geq 0$. A pair of unitary operator $(P_L, P_R)$ is called a $(\beta, b, \varepsilon)$-state-preparation-pair if $P_L \ket{0}_{b} = \sum_{j \in [2^b]} c_j \ket{j}$ and $P_R \ket{0}_b = \sum_{j \in [2^b]} d_j \ket{j}$ such that $\sum_{j \in [m]} \abs{ \beta c_j^* d_j - y_j } \leq \varepsilon$ and $c_j^*d_j = 0$ for all $m \leq j < 2^b$.
    \end{definition}

    \begin{theorem} [Linear combination of unitary operators \cite{GSLW19}] \label{thm:lcu}
        Suppose
        \begin{enumerate}
          \item $y \in \mathbb{C}^m$ with $\Abs{y}_1 \leq \beta$, and $(P_L, P_R)$ is a $(\beta, b, \varepsilon_1)$-state-preparation-pair for $y$.
          \item For every $k \in [m]$, $U_k$ is an $(n+a)$-qubit unitary operator that is an $(\alpha, a, \varepsilon_2)$-block-encoding of an $n$-qubit operator $A_k$.
        \end{enumerate}
        Then there is an $(n+a+b)$-qubit quantum circuit $W$ such that
        \begin{enumerate}
          \item $W$ is a $(\alpha\beta, a+b, \alpha\varepsilon_1+\alpha\beta\varepsilon_2)$-block-encoding of $A = \sum_{k \in [m]} y_k A_k$.
          \item $W$ uses $1$ query to each of $P_L^\dag$, $P_R$ and (controlled-)$U_k$ for $k \in [m]$, and $O(b^2)$ elementary quantum gates.
        \end{enumerate}
    \end{theorem}

    \subsection{Eigenvalue threshold projector}

    In this subsection, we show how to block-encode an eigenvalue threshold projector $\Pi_{\supp(A)}$ of a subnormalized density operator $A$ in another.
    We note that a technique for block-encoding eigenvalue threshold projectors in subnormalized density operators was also provided in \cite{vAG19}.
    The major difference is that our approach does not have any further requirements on the subnormalized density operator $A$, while the method of \cite{vAG19} requires that $A \geq q\Pi$ for some projector $\Pi$ and the value of $q > 0$ is known in advance.

    First, we introduce the notion of truncated support. Let $\delta > 0$ and $A$ be an Hermitian operator with spectral decomposition $A = \sum_{j} \lambda_j \ket{\psi_j} \bra{\psi_j}$.
    The $\delta$-support of $A$ is
        \[
            \supp_\delta(A) = \spanspace \{ \ket{\psi_j}: \abs{\lambda_j} > \delta \}.
        \]
    Note that $\supp_\delta(A) \subseteq \supp_0(A) = \supp(A)$. Here, we write $\Pi_{S}$ to denote the projector onto a subspace $S$.

    \begin{lemma} [Eigenvalue threshold projector] \label{lemma:eigenvalue-threshold-projector}
        Suppose
        \begin{enumerate}
          \item $U$ is an $(n+a)$-qubit unitary operator that prepares an $n$-qubit subnormalized density operator $A$.
          \item $\delta, \varepsilon \in (0, \frac 1 {10}]$ and $32\varepsilon^2 \leq \delta$.
        \end{enumerate}
        For every $\delta' > 0$, there is a quantum circuit $\tilde U$ such that
        \begin{enumerate}
          \item $\tilde U$ prepares an $n$-qubit subnormalized density operator, which is a $(1, 0, \delta')$-block-encoding of $B$ such that
          \begin{align*}
              \left( \frac {\delta} 4 (1 - 2\varepsilon) - \delta^{1/2}\varepsilon \right) \Pi_{\supp_{2\delta}(A)} \leq B   \leq \left( \frac {\delta} 4 + \varepsilon^2 + \delta^{1/2}\varepsilon \right) \Pi_{\supp(A)}.
          \end{align*}
          \item $\tilde U$ uses $O(d)$ queries to $U$ and $U^\dag$, $1$ query to controlled-$U$ and $O((n+a)d)$ elementary quantum gates, where $d = O\left(\frac{1}{\delta}\log\left(\frac{1}{\varepsilon}\right)\right)$.
          \item A description of $\tilde U$ can be computed by a (classical) Turing machine in $\poly(d, \log(1/\delta'))$ time.
        \end{enumerate}
    \end{lemma}

    \begin{proof}
        Let $f(x) = \frac{\delta^{1/2}}{2} x^{-1/2}$, and by Lemma \ref{lemma:negative-power}, there is a degree-$O\left(\frac 1 {\delta} \log\left(\frac 1 {\varepsilon}\right)\right)$ even polynomial $P \in \mathbb{R}[x]$ such that $\Abs{P-f}_{[\delta, 1]} \leq \varepsilon$ and $\Abs{P}_{[-1, 1]} \leq 1$. By Corollary \ref{lemma:threshold-function}, there is a degree-$O\left(\frac 1 {\delta} \log\left(\frac 1 {\varepsilon}\right)\right)$ even polynomial $R \in \mathbb{R}[x]$ such that $\Abs{R}_{[-1, 1]} \leq 1$ and
        \[
            R(x) \in \begin{cases}
                [1-\varepsilon, 1] & x \in [-1, -2\delta] \cup [2\delta, 1] \\
                [0, \varepsilon] & x \in [-\delta, \delta]
            \end{cases}.
        \]
        Note that $Q = PR \in \mathbb{R}[x]$ is a degree-$d$ even polynomial, where $d =O\left(\frac 1 {\delta} \log\left(\frac 1 {\varepsilon}\right)\right)$. By Theorem \ref{thm:qsvt}, for $\delta_Q > 0$, there is a quantum circuit $U_Q$ such that
        \begin{enumerate}
          \item $U_Q$ is a $(1, O(n+a), 0)$-block-encoding of $A_1$, and $A_1$ is a $(1, 0, \delta_Q)$-block-encoding of $Q(A)$.
          \item $U_Q$ uses $d$ queries to $U$ and $U^\dag$, $1$ query to controlled-$U$ and $O((n+a)d)$ elementary quantum gates.
          \item A description of $U_Q$ can be computed by a (classical) Turing machine in $\poly(d, \log(1/\delta_Q))$ time.
        \end{enumerate}
        By Lemma \ref{lemma:density-basic}, using $1$ query to each of $U_Q$ and $U$, we obtain a unitary operator $\tilde U$ that prepares $A_1 A A_1^\dag$. We note that
        \begin{align*}
            \Abs{A_1A A_1^\dag - A(Q(A))^2} & \leq \Abs{A_1A A_1^\dag - Q(A)A A_1^\dag} + \Abs{Q(A)A A_1^\dag - Q(A)A Q(A)} \\
            & \leq 2\delta_Q.
        \end{align*}
        Moreover, $A_1 A A_1^\dag$ can be regarded as a (scaled) projector.
        To see this, let $\mathbbm{1}_{S}$ be the indicator function that
        \[
            \mathbbm{1}_{S}(x) = \begin{cases}
                1, & x \in S, \\
                0, & \text{otherwise}.
            \end{cases}
        \]
        Then $\Pi_{\supp_{\delta}(A)} = \mathbbm{1}_{[\delta, 1]}(A)$. We note that if $32\varepsilon^2 \leq \delta$, then
        \begin{equation} \label{eq:ineq-supp-proj}
        \begin{aligned}
            & \left( \frac {\delta} 4 (1 - 2\varepsilon) - \delta^{1/2}\varepsilon \right) \Pi_{\supp_{2\delta}(A)} \leq A(Q(A))^2 \leq \left( \frac {\delta} 4 + \varepsilon^2 + \delta^{1/2}\varepsilon \right) \Pi_{\supp(A)}.
        \end{aligned}
        \end{equation}
        We need to show that
        \begin{align*}
            & \left( \frac {\delta} 4 (1 - 2\varepsilon) - \delta^{1/2}\varepsilon \right) \mathbbm{1}_{[2\delta, 1]}(x) \leq x(Q(x))^2 \leq \left( \frac {\delta} 4 + \varepsilon^2 + \delta^{1/2}\varepsilon \right) \mathbbm{1}_{(0, 1]}(x)
        \end{align*}
        for every $x \in [0, 1]$. This is seen by the following four cases.
        \begin{enumerate}
          \item $x = 0$. This case is trivial as each hand side is equal to $0$.
          \item $x \in (0, \delta]$. We have $0 \leq x(Q(x))^2 \leq x(R(x))^2 \leq x \varepsilon^2 \leq \delta \varepsilon^2$.
          \item $x \in (\delta, 2\delta]$. We have
          \begin{align*}
              0 & \leq x(Q(x))^2 \leq x(P(x))^2 \leq x(f(x)+\varepsilon)^2 \\
              & = x(f(x))^2 + x \varepsilon^2 + 2xf(x)\varepsilon \leq \frac {\delta} 4 + \varepsilon^2 + \delta^{1/2} \varepsilon.
          \end{align*}
          \item $x \in (2\delta, 1]$. The right hand side is $x(Q(x))^2 \leq \frac {\delta} 4 + \varepsilon^2 + \delta^{1/2} \varepsilon$ is similar to Case 3. The left hand side is as follows.
          \begin{align*}
            x(Q(x))^2
            & \geq x(P(x))^2(1-\varepsilon)^2 \\
            & \geq x(f(x)-\varepsilon)^2(1-\varepsilon)^2 \\
            & \geq x((f(x))^2-2f(x)\varepsilon)(1-2\varepsilon) \\
            & = \left( \frac {\delta} 4 - \delta^{1/2} x^{1/2} \varepsilon \right) (1 - 2\varepsilon) \\
            & \geq \frac {\delta} 4 (1 - 2\varepsilon) - \delta^{1/2} \varepsilon.
          \end{align*}
        \end{enumerate}
        Therefore, we conclude that Eq. (\ref{eq:ineq-supp-proj}) holds. We claim the lemma by setting $B = A (Q(A))^2$ and $\delta' = 2 \delta_Q$.
    \end{proof}

    \section{Quantum Entropies} \label{sec:entropy}
    
    In the Introduction, we have already introduced several quantum entropies such as the von Neumann entropy, quantum R\'{e}nyi entropy and quantum Tsallis entropy. In addition to this, the quantum Min entropy $S^{\min}(\rho)$ and the quantum Max (Hartley) entropy $S^{\max}(\rho)$ are defined as limits of R\'{e}nyi entropies by
    \begin{align*}
        S^{\min}(\rho) & = S^{R}_\infty(\rho) = \lim_{\alpha \to \infty} S^{R}_\alpha(\rho) = - \ln \left( \Abs{\rho} \right), \\
        S^{\max}(\rho) & = S^{R}_0(\rho) = \lim_{\alpha \to 0} S^{R}_\alpha(\rho) = \ln\left(\rank(\rho)\right).
    \end{align*}
    The unified entropy \cite{HY06} is defined by
    \[
        S_\alpha^s(\rho) = \frac{1}{(1-\alpha)s} \left( \left(\tr\left(\rho^\alpha\right)\right)^{s} - 1 \right)
    \]
    for $\alpha \in (0, 1) \cup (1, +\infty)$ and $s \neq 0$, which includes the von Neumann entropy $S(\rho) = \lim_{\alpha \to 1} S_\alpha^s(\rho)$, the R\'{e}nyi entropy $S^{R}_{\alpha}(\rho) = \lim_{s \to 0} S_\alpha^s(\rho)$ and the Tsallis entropy $S^{T}_{\alpha}(\rho) = S_{\alpha}^{1}(\rho)$.
    
    In this section, we will propose a series of quantum algorithms for computing several quantum entropies. Section \ref{subsec:von-Neumann} provides a quantum algorithm for computing the von Neumann entropy. Section \ref{subsec:rank} is for the Max entropy. Section \ref{sec:renyi-tsallis} is for the quantum R\'{e}nyi and Tsallis entropies.

    \subsection{Von Neumann Entropy} \label{subsec:von-Neumann}

    The von Neumann entropy is one of the most important quantum information quantities. As mentioned above, both quantum algorithms for estimating von Neumann entropy provided by \cite{AISW19} and \cite{GL20} have time complexity exponential in the number $n = \log_2 (N)$ of qubits of the quantum state. Here, we provide a different approach that exponentially improves the dependence on $n$ given that the density operator of the mixed quantum state is low-rank. Our key technique used here is different from that of \cite{GL20}, where they approximated a function $\propto -\ln(x)$ and constructed a unitary operator that is a block-encoding of $S(\rho)$, while we approximate a function $\propto \sqrt{-\ln(x)}$ and prepare a density operator that is a block-encoding of $S(\rho)$.

    \begin{theorem} \label{thm: von Neumann entropy}
        Suppose that
        \begin{enumerate}
          \item $U_\rho$ is an $(n+n_\rho)$-qubit unitary operator that prepares an $n$-qubit density operator $\rho$ with $\rank(\rho) \leq r$.
          \item $n_\rho$ is a polynomial in $n$.\footnote{Theoretically, any $n$-qubit mixed quantum state has a purification with at most $n$ ancilla qubits, so it is sufficient to assume that $n_\rho \leq n$. Here, we just assume that $n_\rho = \poly(n)$ for convenience.}
        \end{enumerate}
        There is a quantum algorithm that computes the von Neumann entropy $S(\rho)$ within additive error $\varepsilon$ using $\tilde O(\frac{r}{\varepsilon^2})$ queries to $U_\rho$ and $\tilde O(\frac{r}{\varepsilon^2} \poly(n))$ elementary quantum gates.\footnote{Since the quantum algorithm is complicated, we do not distinguish between the queries to a unitary operator and those to its controlled versions, and we ignore poly-logarithmic factors.}
    \end{theorem}

    Before the proof of Theorem \ref{thm: von Neumann entropy}, we need the following method of approximating functions by polynomials based on Taylor series.

    \begin{lemma} [Corollary 66 of the full version of \cite{GSLW19}] \label{lemma:approximate based on Taylor}
        Let $x_0 \in [-1, 1]$, $r \in (0, 2]$, $\delta \in (0, r]$ and $f: [x_0-r-\delta, x_0+r+\delta] \to \mathbb{C}$ such that $$f(x_0+x) = \sum_{k=0}^{\infty} a_k x^k$$ for all $x \in [-r-\delta, r+\delta]$. Suppose $B > 0$ and $$\sum_{k=0}^{\infty} \abs{a_k} (r+\delta)^k \leq B.$$ Let $\varepsilon \in (0, \frac {1} {2B}]$, then there is an efficient computable polynomial $P \in \mathbb{C}[x]$ of degree $O\left(\frac{1}{\delta} \log\left(\frac{B}{\varepsilon}\right)\right)$ such that
        \begin{align*}
            \Abs{f(x)-P(x)}_{[x_0-r,x_0+r]} \leq \varepsilon, \\
            \Abs{P}_{[-1,1]} \leq \varepsilon + \Abs{f}_{[x_0-r-\delta/2, x_0+r+\delta/2]} \leq \varepsilon + B, \\
            \Abs{P}_{[-1,1]\setminus[x_0-r-\delta/2, x_0+r+\delta/2]} \leq \varepsilon.
        \end{align*}
    \end{lemma}

    By Lemma \ref{lemma:approximate based on Taylor}, we are able to give an approximation of scaled $\sqrt{-\ln(x)}$ as follows.

    \begin{lemma} \label{lemma:approx-sqrt-ln}
        For every $\delta', \varepsilon \in (0, \frac 1 4]$,
        there is an efficient computable polynomial $P \in \mathbb{R}[x]$ of degree $O\left(\frac{1}{\delta'}\log\left(\frac{1}{\delta'\varepsilon}\right)\right)$ such that
        \begin{align*}
            \Abs{P(x) - \frac{\sqrt{-\ln(x)}}{2\sqrt{-\ln(\delta')}}}_{[\delta', 1-\delta']} \leq \varepsilon, \\
            \Abs{P}_{[-1, 1]} \leq 1.
        \end{align*}
    \end{lemma}

    \begin{proof}
        Let $f(x) = \frac{\sqrt{-\ln(x)}}{2\sqrt{-\ln(\delta')}}$ whose Taylor series expanded around $x_0 = \frac 1 2$ is $f(x_0+x) = \sum_{k=0}^\infty a_k x^k$. We note that $f(x)$ is holomorphic in $\mathbb{C} \setminus (-\infty, 0] \setminus [1, \infty)$ if we choose the definitions of $\ln(\cdot)$ and $\sqrt{(\cdot)}$ to be their principle branches in complex analysis. Thus the radius of convergence of the Taylor series of $f(x)$ expanded around $x_0 = \frac 1 2$ is $R = \frac 1 2$. We have
        \[
            \limsup_{k\to\infty} \sqrt[k]{\abs{a_k}} = \frac 1 R = 2.
        \]
        Therefore, there exists $k_0 \in \mathbb{N}$ such that for every $k > k_0$, it holds that $\sqrt[k]{\abs{a_k}} < 2 + \delta'$. Now we set $r = \frac 1 2 - \delta'$ and $\delta = \delta'/2$, and we have
        \begin{align*}
            \sum_{k=0}^{\infty} \abs{a_k} (r+\delta)^k 
            & \leq \sum_{k=0}^{\infty} (2+\delta')^k \left(\frac 1 2 - \frac {\delta'} 2\right)^k + O(1) \\
            & \leq \sum_{k=0}^{\infty} \left( 1 - \frac {\delta'} {2} \right)^k + O(1) \\
            & = \frac 2 {\delta'} + O(1) \eqqcolon B.
        \end{align*}
        Now applying Lemma \ref{lemma:approximate based on Taylor}, there is an efficiently computable polynomial $P \in \mathbb{C}[x]$ of degree $O\left(\frac{1}{\delta}\log\left(\frac{B}{\varepsilon}\right)\right) = O\left(\frac{1}{\delta'}\log\left(\frac{1}{\delta'\varepsilon}\right)\right)$ such that
        \begin{align*}
            \Abs{f(x) - P(x)}_{[\delta', 1-\delta']} \leq \varepsilon, \\
            \Abs{P(x)}_{[-1,1]} \leq \varepsilon + \Abs{f}_{[3\delta'/4, 1-3\delta'/4]}, \\
            \Abs{P}_{[-1,1]\setminus[3\delta'/4, 1-3\delta'/4]} \leq \varepsilon.
        \end{align*}
        Here, we note that
        \[
            \Abs{f}_{[3\delta'/4, 1-3\delta'/4]} = f\left(\frac{3\delta'}{4}\right) \leq \frac 3 4
        \]
        and we obtain $\Abs{P}_{[-1, 1]} \leq 1$ combining the three cases above. Finally, since we are only interested in the real part of $P(x)$ and $x$ is always a real number, just selecting the real part of the coefficients of $P(x)$ will obtain a desired polynomial with real coefficients.
    \end{proof}

    Now we are ready to give the proof of Theorem \ref{thm: von Neumann entropy}.

    \begin{proof} [Proof of Theorem \ref{thm: von Neumann entropy}]
        Let $\delta, \varepsilon_1 \in (0, \frac 1 4)$ and
        \[
            f(x) = \frac{\sqrt{-\ln(x)}}{2\sqrt{-\ln(\delta)}}.
        \]
        By Lemma \ref{lemma:approx-sqrt-ln}, there is a polynomial $P(x)$ of degree $O\left(\frac{1}{\delta}\log\left(\frac{1}{\delta\varepsilon_1}\right)\right)$ such that $\Abs{P}_{[-1,1]}\leq 1$ and $\Abs{P(x)-f(x)}_{[\delta, 1-\delta]} \leq \varepsilon_1$. By Corollary \ref{lemma:threshold-function-2}, there is an even polynomial $R(x)$ of degree $O\left(\frac{1}{\delta}\log\left(\frac{1}{\varepsilon_1}\right)\right)$ such that $\Abs{R}_{[-1, 1]} \leq 1$ and
        \[
            R(x) \in \begin{cases}
                [1-\varepsilon_1, 1] & x \in [-1+2\delta, 1-2\delta] \\
                [0, \varepsilon_1] & x \in [-1, -1+\delta] \cup [1-\delta, 1]
            \end{cases}.
        \]
        We note that $\frac 1 2 P(x)R(x)$ is a polynomial of degree $d = O\left(\frac{1}{\delta}\log\left(\frac{1}{\delta\varepsilon_1}\right)\right)$ satisfying $\left\| \frac 1 2 P(x)R(x) \right\|_{[-1,1]} \leq \frac 1 2$. Let $\delta_1 \in (0, 1)$ and by Lemma \ref{lemma:technique}, there is a quantum circuit $\tilde U$ such that
        \begin{enumerate}
          \item $\tilde U$ prepares an $n$-qubit subnormalized density operator $A_1$, and $A_1$ is a $(1, 0, \delta_1)$-block-encoding of $\rho\left(\frac 1 2 P(\rho)R(\rho)\right)^2$.
          \item $\tilde U$ uses $O(d)$ queries to $U_\rho$ and $U_\rho^\dag$, $1$ query to controlled-$U$ and controlled-$U_\rho^\dag$, and $O((n+a)d)$ elementary quantum gates.
          \item A description of $\tilde U$ can be computed by a (classical) Turing machine in $O(\poly(d, \log(1/\delta_1)))$ time.
        \end{enumerate}

        Now we are going to show that
        \[
            \Abs{x(P(x)R(x))^2 - x(f(x))^2}_{[0,1]} \leq \Theta\left(\varepsilon_1 + \frac{\delta}{\ln\left(\frac{1}{\delta}\right)}\right),
        \]
        which immediately yields
        \[
            \Abs{\rho(P(\rho)R(\rho))^2 - \rho(f(\rho))^2} \leq \Theta\left(\varepsilon_1 + \frac{\delta}{\ln\left(\frac{1}{\delta}\right)}\right).
        \]
        We consider four cases.
        \begin{enumerate}
          \item $x \in [0, \delta]$. 
          In this case, 
          \begin{align*}
              \abs{x(P(x)R(x))^2 - x(f(x))^2} 
              & \leq \abs{x(P(x)R(x))^2} + \abs{x(f(x))^2} \\
              & \leq \delta + \delta (f(\delta))^2 \\
              & = \delta + \delta / 4 \\
              & = \Theta(\delta).
          \end{align*}
          \item $x \in [\delta, 1-2\delta]$. In this case,
            \begin{align*}
              \abs{x(P(x)R(x))^2 - x(f(x))^2}
              & \leq \abs{x} \big( \abs{(P(x))^2 - (f(x))^2} \abs{(R(x))^2} + \abs{(R(x))^2 - 1} \abs{(f(x))^2} \big) \\
              & \leq 2 \abs{P(x) - f(x)} + 2 \abs{R(x) - 1} \\
              & \leq 2\varepsilon_1 + 2\varepsilon_1 = \Theta(\varepsilon_1).
            \end{align*}
          \item $x \in [1-2\delta, 1-\delta]$. We note that $-x\ln(x) < 1 - x$ for every $x \in (0, 1)$. In this case,
            \begin{align*}
              \abs{x(P(x)R(x))^2 - x(f(x))^2} 
              & \leq \abs{x(P(x)R(x))^2} + \abs{x(f(x))^2} \\
              & \leq \abs{x(P(x))^2} + \abs{x(f(x))^2} \\
              & \leq 2 \abs{x(f(x))^2} + \abs{x(P(x))^2 - x(f(x))^2} \\
              & \leq 2 \frac{-x\ln(x)}{-4\ln(\delta)} + 2 \abs{P(x) - f(x)} \\
              & \leq \frac{1-x}{-2\ln(\delta)} + 2 \varepsilon_1 \\
              & \leq \frac{2\delta}{-2\ln(\delta)} + 2 \varepsilon_1 \\
              & = \Theta\left(\frac{\delta}{\ln\left(\frac{1}{\delta}\right)} + \varepsilon_1\right).
            \end{align*}
          \item $x \in [1-\delta, 1]$. In this case,
            \begin{align*}
              \abs{x(P(x)R(x))^2 - x(f(x))^2} 
              & \leq \abs{x(P(x)R(x))^2} + \abs{x(f(x))^2} \\
              & \leq \abs{(R(x))^2} + \abs{x(f(x))^2} \\
              & \leq \Theta\left(\varepsilon_1 + \frac{\delta}{\ln\left(\frac{1}{\delta}\right)}\right).
            \end{align*}
        \end{enumerate}
        We note that $\tr\left(\rho(f(\rho))^2\right) = \frac{S(\rho)}{4\ln\left(\frac{1}{\delta}\right)}$. Based on this, we have
        \begin{align*}
            \abs{16 \ln\left(\frac{1}{\delta}\right) \tr(A_1) - S(\rho)} \leq \Theta\left(r\left((\varepsilon_1+\delta_1)\ln\left(\frac{1}{\delta}\right) + \delta\right)\right).
        \end{align*}
        On the other hand, $\tr(A_1)$ has an upper bound that
        \[
            \tr(A_1) \leq \frac{S(\rho)}{16 \ln \left(\frac{1}{\delta}\right)} + O(1) \leq \frac{\ln(r)}{16 \ln \left(\frac{1}{\delta}\right)} + O(1) \eqqcolon B.
        \]

        Let $\varepsilon_2 \in (0, 1)$. By Lemma \ref{lemma:trace-estimation}, we can compute $\tilde p$ such that $\abs{\tr(A_1) - \tilde p} \leq \varepsilon_2$ with $O\left(\frac{\sqrt{B}}{\varepsilon_2} + \frac{1}{\sqrt{\varepsilon_2}}\right)$ queries to $\tilde U$ and $\tilde U^\dag$.

        Finally, we choose $16 \ln \left(\frac{1}{\delta}\right) \tilde p$ to be the approximation of $S(\rho)$. The error is bounded by
        \begin{equation} \label{eq:von-Neumann-entropy-error}
        \begin{aligned}
            \abs{ 16 \ln \left(\frac{1}{\delta}\right) \tilde p - S(\rho) } \leq \Theta\left(r\left((\varepsilon_1+\delta_1)\ln\left(\frac{1}{\delta}\right) + \delta\right) + \varepsilon_2 \ln\left(\frac{1}{\delta}\right) \right).
        \end{aligned}
        \end{equation}
        Let the right hand side of Eq. (\ref{eq:von-Neumann-entropy-error}) become $\leq \varepsilon$, then the number of queries to $U_\rho$ is 
        \[
        O\left( d \right) \cdot O\left( \frac{\sqrt{B}}{\varepsilon_2} + \frac{1}{\sqrt{\varepsilon_2}} \right) = \tilde O\left(\frac{d}{\varepsilon_2}\right) = \tilde O\left(\frac{r}{\varepsilon^2}\right).
        \]
    \end{proof}

    \subsection{Max Entropy and Rank of Quantum States} \label{subsec:rank}

    The Max (Hartley) entropy $S^{\max}(\rho) = \ln\left(\rank(\rho)\right)$ of a quantum state $\rho$ is the logarithm of its rank.
    Low-rank quantum states turn out to be useful in quantum algorithms, e.g. \cite{GLF10, LMR14, RSML18, BKL+19, WZC+22}.
    Estimating the rank of quantum states is important in checking whether a quantum state fits in a certain low-rank condition of a quantum algorithm. Recently, a variational quantum algorithm was proposed in \cite{TV21} to estimate the rank of quantum states. Here, we provide a quantum algorithm for rank estimation. For $\delta > 0$, the $\delta$-rank of a matrix $A$ is defined by
    \[
        \rank_{\delta}(A) = \min \{ \rank(B): \Abs{A-B} \leq \delta \}.
    \]
    In particular, $\rank_0(A) = \rank(A)$. We note that if $A$ is Hermitian and has the spectrum decomposition $A = \sum_j \lambda_j \ket{\psi_j} \bra{\psi_j}$, then $\rank_\delta(A) = \tr\left(\Pi_{\supp_\delta(A)}\right)$ is the number of eigenvalues $\lambda_j$ such that $\abs{\lambda_j} > \delta$.

    \begin{theorem} [Rank estimation] \label{thm:rank}
        Suppose that
        \begin{enumerate}
          \item $U$ is an $(n+a)$-qubit unitary operator that prepares an $n$-qubit density operator $\rho$.
          \item $\delta \in (0, \frac 1 {10}]$.
        \end{enumerate}
        For $\varepsilon, \varepsilon' \in (0, 1)$, there is a quantum algorithm that outputs $\tilde r$ such that
        \[
            (1-\varepsilon) \rank_{\delta}(\rho) - \varepsilon' \leq \tilde r \leq (1+\varepsilon) \rank(\rho) + \varepsilon',
        \]
        using $O\left( \frac{1}{\delta^2 \varepsilon'} \log\left(\frac{1}{\delta\varepsilon}\right) \right)$ queries to $U$ and $O\left( \frac{1}{\delta^2 \varepsilon'} \log\left(\frac{1}{\delta\varepsilon}\right) \poly(n) \right)$ elementary quantum gates.
    \end{theorem}

    \begin{proof}

        \textbf{Step 1}. By Lemma \ref{lemma:eigenvalue-threshold-projector}, introducing parameters $\delta_1, \varepsilon_1 \in (0, \frac 1 {10}]$ with $32 \varepsilon_1^2 \leq \delta$, there is a quantum circuit $\tilde U$ such that
        \begin{enumerate}
          \item $\tilde U$ prepares an $n$-qubit subnormalized density operator $A$, which is a $(1, 0, \delta_1)$-block-encoding of $\tilde \Pi$ such that
            \begin{equation} \label{eq:rank-supp-proj}
            \begin{aligned}
                \left( \frac {\delta} 8 (1 - 2\varepsilon_1) - \left(\frac \delta 2\right)^{1/2}\varepsilon_1 \right) \Pi_{\supp_{\delta}(\rho)} \leq \tilde \Pi \leq \left( \frac {\delta} 8 + \varepsilon_1^2 + \left(\frac\delta 2\right)^{1/2}\varepsilon_1 \right) \Pi_{\supp(\rho)}.
            \end{aligned}
            \end{equation}
          \item $\tilde U$ uses $O(d)$ queries to $U$ and $U^\dag$, $1$ query to controlled-$U$ and $O((n+a)d)$ elementary quantum gates, where $d = O\left(\frac{1}{\delta}\log\left(\frac{1}{\varepsilon_1}\right)\right)$.
          \item A description of $\tilde U$ can be computed by a (classical) Turing machine in $\poly(d, \log(1/\delta_1))$ time.
        \end{enumerate}
        We note that $\abs{\tr(A) - \tr(\tilde \Pi)} \leq 2^n \delta_1$ and by Eq. (\ref{eq:rank-supp-proj}) we have
        \begin{align*}
            \left( 1 - 2\varepsilon_1 - 4\sqrt{2} \delta^{-1/2} \varepsilon_1 \right) \rank_\delta(\rho) \leq 8\delta^{-1} \tr(\tilde \Pi) \leq \left( 1 + 8 \delta^{-1} \varepsilon_1^2 + 4 \sqrt{2} \delta^{-1/2} \varepsilon_1 \right) \rank(\rho).
        \end{align*}

        \textbf{Step 2}. Introducing a parameter $\varepsilon_2$, by Lemma \ref{lemma:trace-estimation}, we can compute $\tilde p$ that estimates $\tr(A)$ such that $\abs{\tilde p - \tr(A)} \leq \varepsilon_2$ with $O\left( \frac{\sqrt{B}}{\varepsilon_2} + \frac{1}{\sqrt{\varepsilon_2}} \right)$ queries to $\tilde U$ and $\tilde U^\dag$, where $B$ is an upper bound for $\tr(A)$. Here, we just choose $B = 1$.

        \textbf{Step 3}. Output $\tilde r = 8\delta^{-1}\tilde p$ as the estimation of the rank of $\rho$. Here, we see that
        \begin{align*}
            & \left( 1 - 2\varepsilon_1 - 4\sqrt{2} \delta^{-1/2} \varepsilon_1 \right) \rank_\delta(\rho) - 8\delta^{-1}\left(2^n\delta_1 + \varepsilon_2\right) \leq \tilde r \\
            & \leq \left( 1 + 8 \delta^{-1} \varepsilon_1^2 + 4 \sqrt{2} \delta^{-1/2} \varepsilon_1 \right) \rank(\rho) + 8\delta^{-1}\left(2^n\delta_1 + \varepsilon_2\right).
        \end{align*}
        By letting $\varepsilon_1 = \frac{1}{32} \delta \varepsilon < 1$ (note that it also holds that $32\varepsilon_1^2 < 32 \varepsilon_1 = \delta \varepsilon < \delta$), $\varepsilon_2 = \frac{1}{16}\delta\varepsilon'$, and $\delta_1 = \frac{1}{2^{n+4}}\delta\varepsilon'$, the above inequality becomes $(1-\varepsilon) \rank_{\delta}(\rho) - \varepsilon' \leq \tilde r \leq (1+\varepsilon) \rank(\rho) + \varepsilon'$. 
        To see this, we have to show the following three inequalities. 
        \begin{enumerate}
            \item $2\varepsilon_1 + 4\sqrt{2} \delta^{-1/2} \varepsilon_1 \leq \varepsilon$. Note that $\delta^{-1} \geq \delta^{-1/2} \geq 3$. We have $2\varepsilon_1 + 4\sqrt{2} \delta^{-1/2} \varepsilon_1 \leq \delta^{-1} \varepsilon_1 + 4\sqrt{2} \delta^{-1} \varepsilon_1 = (1+4\sqrt{2}) \delta^{-1} \varepsilon_1 \leq 32 \delta^{-1} \varepsilon_1 = \varepsilon$.
            \item $8 \delta^{-1} \varepsilon_1^2 + 4 \sqrt{2} \delta^{-1/2} \varepsilon_1 \leq \varepsilon$. 
            Note that $\varepsilon_1^2 < \varepsilon_1 < 1$ and $\delta^{-1} \geq \delta^{-1/2} \geq 3$.
            We have
            $8 \delta^{-1} \varepsilon_1^2 + 4 \sqrt{2} \delta^{-1/2} \varepsilon_1 \leq 8 \delta^{-1} \varepsilon_1 + 4 \sqrt{2} \delta^{-1} \varepsilon_1 = (8 + 4\sqrt{2}) \delta^{-1} \varepsilon_1 \leq 32 \delta^{-1} \varepsilon_1 = \varepsilon$.
            \item $8\delta^{-1}\left(2^n\delta_1 + \varepsilon_2\right) \leq \varepsilon'$. This is simple by directly taking the values of $\varepsilon_2$ and $\delta_1$. 
        \end{enumerate}
        Finally, the number of queries to $U$ is
        \[
            O\left( \frac{1}{\delta} \log \left(\frac{1}{\varepsilon_1}\right) \cdot \left( \frac{\sqrt{B}}{\varepsilon_2} + \frac{1}{\sqrt{\varepsilon_2}} \right) \right) = O\left( \frac{1}{\delta^2 \varepsilon'} \log\left(\frac{1}{\delta\varepsilon}\right) \right).
        \]
        And similarly, the number of elementary quantum gates is $O\left( \frac{1}{\delta^2 \varepsilon'} \log\left(\frac{1}{\delta\varepsilon}\right) \poly(n) \right)$.
    \end{proof}
    
    Based on Theorem \ref{thm:rank}, we can obtain the exact rank of $\rho$ if $\Pi/\kappa \leq \rho$ for some projector $\Pi$ and $\kappa > 0$.
    
    \begin{corollary} [Exact rank] \label{corollary:exact-rank}
        Suppose $U$ is an $(n+a)$-qubit unitary operator that prepares an $n$-qubit density operator $\rho$.
        If $\Pi/\kappa \leq \rho$ for some projector $\Pi$ and $\kappa > 0$, then
        there is a quantum algorithm that outputs $r = \rank(\rho)$
        using $\tilde O(\kappa^2)$ queries to $U$ and $\tilde O(\kappa^2 \poly(n))$ elementary quantum gates.
    \end{corollary}
    
    \begin{proof}
        The claim holds immediately by letting $\delta = \varepsilon = \frac{1}{10\kappa}$ and $\varepsilon' = 1/10$ in Theorem \ref{thm:rank}. 
    \end{proof}

    Now we furthermore define the $\delta$-Max entropy by $S^{\max}_\delta(\rho) = \ln(\rank_{\delta}(\rho))$. Based on Theorem \ref{thm:rank}, we are able to give an estimation of the Max entropy.

    \begin{theorem} [Max entropy estimation] \label{thm:max-entropy}
        Suppose that
        \begin{enumerate}
          \item $U$ is an $(n+a)$-qubit unitary operator that prepares an $n$-qubit density operator $\rho$.
          \item $\delta \in (0, \frac 1 {10}]$.
        \end{enumerate}
        For $\varepsilon > 0$, there is a quantum algorithm that outputs $\tilde s$ such that
        \[
            S^{\max}_{\delta}(\rho) - \varepsilon \leq \tilde s \leq S^{\max}(\rho) + \varepsilon,
        \]
        using $\tilde O(\frac{1}{\delta^2 \varepsilon})$ queries to $U$ and $\tilde O(\frac{1}{\delta^2 \varepsilon} \poly(n))$ elementary quantum gates.

        Moreover, if $\Pi/\kappa \leq \rho$ for some projector $\Pi$ and $\kappa > 0$, there is a quantum algorithm that computes the Max entropy $S^{\max}(\rho)$ within additive error $\varepsilon$ using $\tilde O(\frac{\kappa^2}{\varepsilon})$ queries to $U$ and $\tilde O(\frac{\kappa^2}{\varepsilon} \poly(n))$ elementary quantum gates.
    \end{theorem}
    \begin{proof}
        By Theorem \ref{thm:rank}, there is a quantum algorithm that outputs $\tilde r$ such that
        \begin{align*}
            \left( 1-\frac \varepsilon 2 \right) \rank_{\delta}(\rho) & \leq \left( 1-\frac \varepsilon 4 \right) \rank_{\delta}(\rho) - \frac \varepsilon 4 \leq \tilde r \\
            & \leq \left( 1+\frac \varepsilon 4 \right) \rank(\rho) + \frac \varepsilon 4 \\
            & \leq \left( 1+\frac \varepsilon 2 \right) \rank(\rho),
        \end{align*}
        using $\tilde O(\frac{1}{\delta^2 \varepsilon})$ queries to $U$ and $\tilde O(\frac{1}{\delta^2 \varepsilon} \poly(n))$ elementary quantum gates. After taking logarithm of both sides, we have
        \begin{align*}
            \ln \left( \rank_{\delta}(\rho) \right) + \ln \left( 1-\frac \varepsilon 2 \right) 
            \leq \ln(\tilde r) \leq \ln\left(\rank(\rho)\right) + \ln \left( 1+\frac \varepsilon 2 \right),
        \end{align*}
        which is
        \[
            S^{\max}_\delta(\rho) - \varepsilon \leq \ln(\tilde r) \leq S^{\max}(\rho) + \varepsilon.
        \]
        We only need to output $\tilde s = \ln(\tilde r)$ as an estimation of the Max entropy $S^{\max}(\rho)$.

        Moreover, if $\Pi/\kappa \leq \rho$ for some projector $\Pi$ and $\kappa > 0$, the claim is obtained by letting $\delta = \frac{1}{\kappa}$ (and then $S^{\max}_\delta(\rho) = S^{\max}(\rho)$).
    \end{proof}

    \subsection{Quantum R\'{e}nyi Entropy and Quantum Tsallis Entropy} \label{sec:renyi-tsallis}

    Now we are going to discuss how quantum algorithms can compute the quantum R\'{e}nyi entropy $S_\alpha^R(\rho)$ and the quantum Tsallis entropy $S_\alpha^T(\rho)$. The key is to compute $\tr(\rho^\alpha)$, which is given as follows.

    \begin{theorem} [Trace of positive powers] \label{thm:trace-positive-powers}
        Suppose that
        \begin{enumerate}
          \item $U_\rho$ is an $(n+n_\rho)$-qubit unitary operator that prepares an $n$-qubit density operator $\rho$ with $\rank(\rho) \leq r$.
          \item $n_\rho$ is a polynomial in $n$.
        \end{enumerate}
        For $a \in (0, 1) \cup (1, +\infty)$, there is a quantum algorithm that computes $\tr(\rho^\alpha)$ within additive error $\varepsilon$ using $Q$ queries to $U_\rho$ and $Q \cdot \poly(n)$ elementary quantum gates, where
        \[
            Q = \begin{cases}
                \tilde O\left( r^{\frac{3-\alpha^2}{2\alpha}}/\varepsilon^{\frac{3+\alpha}{2\alpha}} \right), & 0 < \alpha < 1, \\
                O\left( 1/\varepsilon \right), & \alpha > 1 \land \alpha \text{ is odd}, \\
                \tilde O \left( r^{1/\left\{\frac{\alpha - 1}{2}\right\}}/\varepsilon^{1+1/\left\{\frac{\alpha - 1}{2}\right\}} \right), & \text{otherwise}.
            \end{cases}
        \]
        and $\{ x \} = x - \floor{x}$ is the decimal part of real number $x$.
    \end{theorem}

    \begin{proof}
        We will discuss in three cases as stated in the statement of this theorem.

        \textbf{Case 1}. $0 < \alpha < 1$. In this case, we have $\tr(\rho^\alpha) \leq r^{1-\alpha}$. By Lemma \ref{lemma:positive-power-density} and introducing $\delta_1$ and $\varepsilon_1$, there is a quantum circuit $\tilde U$ such that
        \begin{enumerate}
          \item $\tilde U$ prepares an $n$-qubit subnormalized density operator $A$, and $A$ is a $\left( 4\delta_1^{\alpha-1}, 0, \Theta\left(\delta_1^{\alpha} + \varepsilon_1 \delta_1^{\alpha - 1} \right) \right)$-block-encoding of $\rho^{\alpha}$.
          \item $\tilde U$ uses $O(d)$ queries to $U_\rho$ and $O((n+n_{\rho})d)$ elementary quantum gates, where $d = O\left( \frac{1}{\delta_1} \log \left( \frac {1}{\varepsilon_1} \right) \right)$.
          \item The description of $\tilde U$ can be computed by a classical Turing machine in $O(\poly(d))$ time.
        \end{enumerate}
        Note that $A$ is a $\left( 4\delta_1^{\alpha-1}, 0, \Theta\left(\delta_1^{\alpha} + \varepsilon_1 \delta_1^{\alpha - 1} \right) \right)$-block-encoding of $\rho^{\alpha}$, i.e., $\Abs{4\delta_1^{\alpha - 1} A - \rho^{\alpha}} \leq \Theta (\delta_1^{\alpha} + \varepsilon_1 \delta_1^{\alpha - 1})$, we have $\abs{4 \delta_1^{\alpha - 1} \tr(A) - \tr(\rho^\alpha) } \leq \Theta (r(\delta_1^{\alpha} + \varepsilon_1 \delta_1^{\alpha-1}))$. Therefore,
        \begin{align*}
            \tr(A)
            & \leq \frac{1}{4\delta_1^{\alpha - 1}} \left( \tr(\rho^\alpha) + \Theta(r(\delta_1^{\alpha} + \varepsilon_1 \delta_1^{\alpha-1})) \right) \\
            & \leq \Theta(r^{1-\alpha} \delta_1^{1-\alpha} + r\delta_1 + r\varepsilon_1) \eqqcolon B.
        \end{align*}
        By Lemma \ref{lemma:trace-estimation} and introducing $\varepsilon_2$, we can obtain $\tilde x$ as an approximation of $\tr(A)$ with $Q = O\left( \frac{\sqrt{B}}{\varepsilon_2} + \frac{1}{\sqrt{\varepsilon_2}} \right)$ queries to $\tilde U$ such that $\abs{\tr(A) - \tilde x} \leq \varepsilon_2$. Then we output $4\delta_1^{\alpha-1}\tilde x$ as an approximation of $\tr(\rho^\alpha)$ by noting that
        \[
            \abs{4\delta_1^{\alpha-1}\tilde x - \tr(\rho^{\alpha})} \leq \Theta\left( \delta_1^{\alpha-1}\varepsilon_2 + r(\delta_1^{\alpha} + \varepsilon_1 \delta_1^{\alpha-1}) \right).
        \]
        The number of queries to $U_\rho$ is
        \[
            Q d = \tilde O \left( \frac{1}{\delta_1} \left(\frac{\sqrt{B}}{\varepsilon_2} + \frac{1}{\sqrt{\varepsilon_2}}\right) \right) = \tilde O\left( r^{\frac{3-\alpha^2}{2\alpha}} / \varepsilon^{\frac{3+\alpha}{2\alpha}} \right)
        \]
        by taking $\delta_1 = \tilde \Theta\left( \left(\varepsilon/r\right)^{1/\alpha} \right)$, $\varepsilon_1 = \tilde \Theta\left( \left(\varepsilon/r\right)^{1/\alpha} \right)$ and $\varepsilon_2 = \tilde \Theta\left( \varepsilon^{1/\alpha}/r^{1/\alpha - 1} \right)$.

        \textbf{Case 2}. $\alpha > 1$ and $\alpha$ is an odd number. Let $\alpha = 2\beta + 1$. By Lemma \ref{lemma:block-encoding of density operators}, there is a unitary operator $U_1$, which is a $(1, n+n_\rho, 0)$-block-encoding of $\rho$ with $1$ query to $U_\rho$. By Lemma \ref{lemma:product-block}, there is a unitary operator $U_\beta$, which is a $(1, \Theta(\beta(n+n_\rho)), 0)$-block-encoding of $\rho^\beta$ with $\beta$ queries to $U_1$. By Lemma \ref{lemma:density-basic}, there is a quantum circuit $U$ that prepares $\rho^{2\beta + 1} = \rho^{\alpha}$ with $1$ query to $U_\rho$ and $1$ query to $U_\beta$. Note that $\tr(\rho^{\alpha}) \leq 1$ and by Lemma \ref{lemma:trace-estimation}, we can obtain $\tilde x$ as an approximation of $\tr(\rho^{\alpha})$ with $O(1/\varepsilon)$ queries to $U$ such that $\abs{\tilde x - \tr(\rho^{\alpha})} \leq \varepsilon$. In total, the number of queries to $U_\rho$ is $O(1/\varepsilon) \cdot \beta = O(1 / \varepsilon)$.

        \textbf{Case 3}. $\alpha > 1$ and $\alpha$ is not an odd number. Let $\beta = \floor{\frac{\alpha - 1}{2}}$ and $c = \left\{ \frac{\alpha - 1}{2} \right\}$. By Lemma \ref{lemma:positive-power-unitary} and introducing $\delta_1$ and $\varepsilon_1$, there is a unitary operator $U_c$ that is a $(2, O(n+n_\rho), \Theta(\varepsilon_1 + \delta_1^c))$-block-encoding of $\rho^c$ with $Q_1 = O\left(\frac{1}{\delta_1}\log\left(\frac{1}{\varepsilon_1}\right)\right)$ queries to $U_1$. By Lemma \ref{lemma:product-block}, there is a unitary operator $U_{\beta + c}$ that is a $(2, O(\beta(n+n_\rho)), \Theta(\varepsilon_1 + \delta_1^c))$-block-encoding of $\rho^{\beta + c}$ with $1$ query to $U_\beta$ and $1$ query to $U_c$. By Lemma \ref{lemma:density-basic}, there is a unitary operator $U$ that prepares a subnormalized density operator $A$ with $1$ query to $U_\rho$ and $1$ query to $U_{\beta+c}$, and $A$ is a $(4, 0, \Theta(\varepsilon_1 + \delta_1^c))$-block-encoding of $\rho^{2(\beta + c) + 1} = \rho^{\alpha}$, i.e., $\Abs{4A - \rho^{\alpha}} \leq \Theta(\varepsilon_1 + \delta_1^c)$. Note that \[
        \tr(A) \leq \frac{1}{4} \left( \tr(\rho^{\alpha}) + \Theta(r(\varepsilon_1 + \delta_1^c)) \right) \eqqcolon B.
        \]
        By Lemma \ref{lemma:trace-estimation} and introducing $\varepsilon_2$, we can obtain $\tilde x$ as an approximation of $\tr(A)$ with $Q_2 = O\left(\frac{\sqrt{B}}{\varepsilon_2}+\frac{1}{\sqrt{\varepsilon_2}}\right)$ queries to $U$ such that $\abs{\tilde x - \tr(A)} \leq \varepsilon_2$. Then we output $4\tilde x$ as an approximation of $\tr(\rho^{\alpha})$ by noting that
        \[
            \abs{4\tilde x - \tr(\rho^{\alpha})} \leq \Theta(\varepsilon_2 + r(\varepsilon_1 + \delta_1^c)).
        \]
        The number of queries to $U_\rho$ is
        \begin{align*}
            (\beta + Q_1) Q_2 
            & = \tilde O \left( \left(\beta + \frac{1}{\delta_1} \right) \left(\frac{\sqrt{B}}{\varepsilon_2}+\frac{1}{\sqrt{\varepsilon_2}}\right) \right) \\
            & = \tilde O \left( \frac{r^{1/c}}{\varepsilon^{1+1/c}} \right)
        \end{align*}
        by taking $\delta_1 = \tilde \Theta\left((\varepsilon/r)^{1/c}\right)$, $\varepsilon_1 = \tilde \Theta\left(\varepsilon/r\right)$ and $\varepsilon_2 = \tilde \Theta\left(\varepsilon\right)$.
    \end{proof}

    In the following, we are going to show how to estimate the quantum R\'{e}nyi and Tsallis entropies based on Theorem \ref{thm:trace-positive-powers}.

    \begin{theorem} [Quantum R\'{e}nyi entropy] \label{thm:Renyi-entropy}
        Suppose that
        \begin{enumerate}
          \item $U_\rho$ is an $(n+n_\rho)$-qubit unitary operator that prepares an $n$-qubit density operator $\rho$ with $\rank(\rho) \leq r$.
          \item $n_\rho$ is a polynomial in $n$.
        \end{enumerate}
        For $a \in (0, 1) \cup (1, +\infty)$, there is a quantum algorithm that computes $S^{R}_\alpha(\rho)$ within additive error $\varepsilon$ using $Q$ queries to $U_\rho$ and $Q \cdot \poly(n)$ elementary quantum gates, where
        \[
            Q = \begin{cases}
                \tilde O\left( r^{\frac{3-\alpha^2}{2\alpha}}/\varepsilon^{\frac{3+\alpha}{2\alpha}} \right), & 0 < \alpha < 1, \\
                O\left( r^{\alpha - 1} / \varepsilon \right), & \alpha > 1 \land \alpha \text{ is odd}, \\
                \tilde O \left( r^{\alpha - 1 + \alpha / \{ \frac{\alpha-1}{2} \}} / \varepsilon^{1+1/\{ \frac{\alpha-1}{2} \}} \right), &  \text{otherwise}.
            \end{cases}
        \]
        and $\{ x \} = x - \floor{x}$ is the decimal part of real number $x$.
    \end{theorem}
    \begin{proof}
        We will discuss in cases as stated in the statement of this theorem.

        \textbf{Case 1}. $0 < \alpha < 1$. In this case, $1 \leq \tr(\rho^{\alpha}) \leq r^{1-\alpha}$. By Theorem \ref{thm:trace-positive-powers}, there is a quantum algorithm that outputs $1 \leq x \leq r^{1-\alpha}$ such that $\abs{x - \tr(\rho^{\alpha})} \leq \varepsilon'$. Then
        \[
            \abs{\frac{1}{1-\alpha} \ln(x) - S^{R}_{\alpha}(\rho)} \leq \Theta\left( \frac{\varepsilon'}{1 - \alpha} \right).
        \]
        Let $\varepsilon' = \Theta((1-\alpha) \varepsilon)$. The number of queries to $U_\rho$ is
        \[
            \tilde O\left( r^{\frac{3-\alpha^2}{2\alpha}}/\varepsilon'^{\frac{3+\alpha}{2\alpha}} \right) = \tilde O\left( r^{\frac{3-\alpha^2}{2\alpha}}/\varepsilon^{\frac{3+\alpha}{2\alpha}} \right).
        \]

        \textbf{Case 2}. $\alpha > 1$. In this case, $r^{1-\alpha} \leq \tr(\rho^{\alpha}) \leq 1$. By Theorem \ref{thm:trace-positive-powers}, there is a quantum algorithm that outputs $r^{1-\alpha} \leq x \leq 1$ such that $\abs{x - \tr(\rho^{\alpha})} \leq \varepsilon'$. Then
        \[
            \abs{\frac{1}{1-\alpha} \ln(x) - S^{R}_{\alpha}(\rho)} \leq \Theta\left( \frac{r^{\alpha - 1}\varepsilon'}{\alpha-1} \right).
        \]
        Let $\varepsilon' = \Theta((\alpha-1)r^{1-\alpha} \varepsilon)$.

        \textbf{Subcase 2.1}. $\alpha > 1$ and $\alpha$ is an odd number. The number of queries to $U_\rho$ is
        \[
        O\left(\frac{\alpha}{\varepsilon'}\right) = O\left(\frac{r^{\alpha - 1}}{\varepsilon}\right).
        \]

        \textbf{Subcase 2.2}. $\alpha > 1$ and $\alpha$ is not an odd number. The number of queries to $U_\rho$ is
        \begin{align*}
            \tilde O \left( \left( \alpha + (r/\varepsilon')^{1/\left\{\frac{\alpha - 1}{2}\right\}} \right) / \varepsilon' \right) = \tilde O \left( r^{\alpha - 1 + \alpha / \{ \frac{\alpha-1}{2} \}} / \varepsilon^{1+1/\{ \frac{\alpha-1}{2} \}} \right).
        \end{align*}
    \end{proof}

    \begin{theorem} [Quantum Tsallis entropy] \label{thm:tsallis-entropy}
        Suppose that
        \begin{enumerate}
          \item $U_\rho$ is an $(n+n_\rho)$-qubit unitary operator that prepares an $n$-qubit density operator $\rho$ with $\rank(\rho) \leq r$.
          \item $n_\rho$ is a polynomial in $n$.
        \end{enumerate}
        For $a \in (0, 1) \cup (1, +\infty)$, there is a quantum algorithm that computes $S^{T}_\alpha(\rho)$ within additive error $\varepsilon$ using $Q$ queries to $U_\rho$ and $Q \cdot \poly(n)$ elementary quantum gates, where
        \[
            Q = \begin{cases}
                \tilde O\left( r^{\frac{3-\alpha^2}{2\alpha}}/\varepsilon^{\frac{3+\alpha}{2\alpha}} \right), & 0 < \alpha < 1, \\
                O\left( 1 / \varepsilon \right), & \alpha > 1 \land \alpha \text{ is odd}, \\
                \tilde O \left( r^{1/\{\frac{\alpha-1}{2}\}} / \varepsilon^{1+1/\{\frac{\alpha-1}{2}\}} \right), & \text{otherwise}.
            \end{cases}
        \]
        and $\{ x \} = x - \floor{x}$ is the decimal part of real number $x$.
    \end{theorem}
    \begin{proof}
        By Theorem \ref{thm:trace-positive-powers}, there is a quantum algorithm that outputs $1 \leq x \leq r^{1-\alpha}$ such that $\abs{x - \tr(\rho^{\alpha})} \leq \varepsilon'$. Then
        \[
            \abs{ \frac{x-1}{1-\alpha}} - S^{T}_{\alpha}(\rho) \leq \Theta\left( \frac{\varepsilon'}{\abs{1 - \alpha}} \right).
        \]
        Let $\varepsilon' = \Theta(\abs{1-\alpha} \varepsilon)$. We will discuss in three cases as stated in the statement of this theorem.

        \textbf{Case 1}. $0 < \alpha < 1$. The number of queries to $U_\rho$ is
        \[
            \tilde O\left( r^{\frac{3-\alpha^2}{2\alpha}}/\varepsilon'^{\frac{3+\alpha}{2\alpha}} \right) = \tilde O\left( r^{\frac{3-\alpha^2}{2\alpha}}/\varepsilon^{\frac{3+\alpha}{2\alpha}} \right).
        \]

        \textbf{Case 2}. $\alpha > 1$ and $\alpha$ is an odd number. The number of queries to $U_\rho$ is
        \[
        O\left(\frac{\alpha}{\varepsilon'}\right) = O\left(\frac{1}{\varepsilon}\right).
        \]

        \textbf{Case 3}. $\alpha > 1$ and $\alpha$ is not an odd number. The number of queries to $U_\rho$ is
        \begin{align*}
            \tilde O \left( \left( \alpha + (r/\varepsilon')^{1/\left\{\frac{\alpha - 1}{2}\right\}} \right) / \varepsilon' \right) = \tilde O \left( r^{1/\{\frac{\alpha-1}{2}\}} / \varepsilon^{1+1/\{\frac{\alpha-1}{2}\}} \right).
        \end{align*}
    \end{proof}

    We note that when $\alpha > 1$ is an odd number, the query complexities of our quantum algorithms for $\tr(\rho^\alpha)$ and $S_\alpha^T(\rho)$ do not depend on $r$. This result can be applied to computing the classical Tsallis entropy of a distribution $p: [N] \to [0, 1]$. Suppose a quantum oracle $U$ is given as
    \begin{equation} \label{eq:prob-oracle}
        U \ket{0} = \sum_{i \in [N]} \sqrt{p(i)} \ket{i}.
    \end{equation}
    For convenience, we assume that $N = 2^n$. This kind of quantum oracle is called ``classical distribution with pure state preparation access''. We start from the quantum state $\ket{0}_n \ket{0}_n$, and the algorithm for estimating the classical Tsallis entropy of $p$ is as follows.
    \begin{enumerate}
      \item Apply $U$ on the first part of the quantum state, then the state becomes
      \[
        \sum_{i\in[N]} \sqrt{p(i)} \ket{i}_n \ket{0}_n.
      \]
      \item Apply a CNOT gate with control qubit the $j$-th qubit of the first part and target qubit the $j$-th qubit of the second part for each $j \in [n]$, then the state becomes
      \[
        \sum_{i\in[N]} \sqrt{p(i)} \ket{i}_n \ket{i}_n.
      \]
      \item Note that if the second part of the quantum state after step 2 is traced out, it becomes a mixed quantum state (density operator)
      \[
        \rho = \sum_{i \in [N]} p(i) \ket{i} \bra{i}.
      \]
      Apply the algorithm of Theorem \ref{thm:tsallis-entropy}, we are able to obtain an estimation of the quantum Tsallis entropy $S^{T}_{\alpha}(\rho)$ (i.e., the classical $\alpha$-Tsallis entropy of $p$) within additive error $\varepsilon$ with $O(1/\varepsilon)$ queries to $U$ and $O(1/\varepsilon \cdot \poly(n))$ elementary quantum gates.
    \end{enumerate}
    
    \begin{corollary} \label{corollary:tsallis}
        For odd integer $\alpha > 1$, given the quantum oracle $U$ to a probability distribution $p:[N] \to [0, 1]$ as in Eq. (\ref{eq:prob-oracle}), there is a quantum algorithm that computes the Tsallis entropy $S_{\alpha}^{T}(p)$ of $p$ within additive error $\varepsilon$, using $O(1/\varepsilon)$ queries to $U$ and $O(1/\varepsilon \cdot \poly(n))$ elementary quantum gates.
    \end{corollary}
    
    For integer $\alpha \geq 2$, according to \cite{EAO+02} (see also \cite{KLL+17}), there is a simple SWAP test like quantum circuit that outputs $0$ and $1$ with probability $(1 \pm \tr(\rho^{\alpha}))/2$, respectively, using $\alpha$ copies of $\rho$. With a straightforward statistical method, we can estimate $\tr(\rho^\alpha)$ within additive error $\varepsilon$ by $O(1/\varepsilon^2)$ repetitions of that quantum circuit. Directly applying this method also implies a quantum algorithm, which computes the Tsallis entropy (both quantum and classical) within additive error $\varepsilon$, using $O(1/\varepsilon^2)$ queries to quantum oracles. Compare to this simple algorithm, our quantum algorithm (of Theorem \ref{thm:tsallis-entropy} and Corollary \ref{corollary:tsallis}) yields a quadratic speedup for odd $\alpha > 1$. 
    
    \subsection{Lower Bounds}
    
    We establish lower bounds for estimating quantum entropies in the low-rank case, which are derived from known lower bounds for estimating entropies of classical distributions \cite{LW19,BKT20}.
    
    \begin{theorem} \label{thm:lower-bounds-entropy}
        Suppose $U_\rho$ is a quantum oracle that prepares the density operator $\rho$ (see Definition \ref{def:subnormalized-density-operator}), and $r = \rank(\rho)$. 
        For $\alpha \geq 0$, any quantum algorithm that computes $S_\alpha^R(\rho)$ within additive error $\varepsilon = \Theta(1)$ requires $\Omega(Q(r))$ queries to $U_\rho$, where
        \[
            Q(r) = \begin{cases}
                r^{1/3}, & \alpha = 0, \\
                \max\{r^{1/7\alpha-o(1)}, r^{1/3}\}, & 0 < \alpha < 1, \\
                \tilde \Omega(r^{1/2}),  & \alpha = 1, \\
                \max\{r^{\frac 1 2 - \frac 1 {2\alpha}}, r^{1/3}\}, & \alpha > 1.
            \end{cases}
        \]
    \end{theorem}
    \begin{proof}
        We obtain lower bounds from those for classical distributions. The quantum query model for classical distributions we adopt is called the ``classical distribution with quantum query access'', which is defined as follows. Suppose $p\colon [N] \to [0, 1]$ is a probability distribution on $[N]$. The quantum oracle $U$ is defined by a function $f \colon [S] \to [N]$ such that
        \[
            U\ket{s, 0} = \ket{s, f(s)}
        \]
        for $s \in [S]$ and $S \in \mathbb{N}$, satisfying
        \[
            p(i) = \frac{1}{S} \abs{\{ s \in [S]: f(s) = i \}}
        \]
        for $i \in [N]$. It is pointed out in \cite{GL20} that we can easily construct a ``purified quantum query access'' oracle by preparing a uniform superposition through the Hadamard gates, and then makes a query to $U$. Therefore, all lower bounds in the ``classical distribution with quantum query access'' model also hold in the ``purified quantum query access'' model.
        
        Let $B(N)$ be the lower bound for estimating the R\'{e}nyi entropy. It is straightforward to see that if there is a quantum algorithm that computes the quantum R\'{e}nyi entropy using $Q(r)$ queries to $U_\rho$, then with the same algorithm, we can compute the R\'{e}nyi entropy of a probability distribution $p:[N] \to [0, 1]$ using $Q(N)$ queries to the ``classical distribution with quantum query access'' oracle. Therefore, we have $Q(r) \geq B(r)$. Our claim immediately holds by taking the lower bounds of $B(N)$ given in \cite{LW19,BKT20}.
    \end{proof}

    \section{Quantum Distances} \label{sec:distance-states}
    
    Distance measures of quantum states are quantum information quantities that indicate their closeness. 
    Testing the closeness of quantum states is a basic problem in quantum property testing. Two of the most important distance measures of quantum states are the trace distance and fidelity. 
    
    In this section, we will provide quantum algorithms for computing them as well as their extensions.
    Section \ref{sec:trace-distance} presents quantum algorithms for computing the trace distance and its extension. 
    Section \ref{sec:fidelity} presents quantum algorithms for computing the fidelity and its extensions.

    \subsection{Trace Distance} \label{sec:trace-distance}
    
    The $\alpha$-trace distance of two quantum states $\rho$ and $\sigma$ is defined by
    \[
    T_{\alpha}(\rho, \sigma) = \tr\left(\abs{\frac{\rho-\sigma}{2}}^\alpha\right).
    \]
    Here, $T_1(\rho, \sigma) = T(\rho, \sigma)$ is the trace distance.
    The trace distance of two pure quantum states (i.e., quantum states of rank $1$) can be computed directly from their fidelity, which can be solved by the SWAP test \cite{BCW01}. The closeness testing of the $1$-, $2$-, and $3$-trace distances were studied in \cite{GL20}.
    
    Our quantum algorithms for computing the $\alpha$-trace distance are given as follows. 
    
    \begin{theorem} \label{thm:trace-distance}
        Suppose that
        \begin{enumerate}
          \item $U_\rho$ is an $(n+n_\rho)$-qubit unitary operator that prepares an $n$-qubit density operator $\rho$.
          \item $U_\sigma$ is an $(n+n_\sigma)$-qubit unitary operator that prepares an $n$-qubit density operator $\sigma$.
          \item $n_\rho$ and $n_\sigma$ are polynomials in $n$.
          \item $r = \max\{\rank(\rho), \rank(\sigma)\}$.
        \end{enumerate}
        For $\alpha > 0$, there is a quantum algorithm that computes $T_{\alpha}(\rho, \sigma)$ within additive error $\varepsilon$ using $Q$ queries to both $U_\rho$ and $U_\sigma$, and $Q \cdot \poly(n)$ elementary quantum gates, where
        \[
        Q = \begin{cases}
            \tilde O(r^3/\varepsilon^4), & \alpha \equiv 0 \pmod 2 \\
            \tilde O(r^{5/\alpha}/\varepsilon^{5/\alpha+1}), & 0 < \alpha < 1 \\
            \tilde O\left( {r^{3+1/\{\alpha/2\}}}/{\varepsilon^{4+1/\{\alpha/2\}}} \right), & \text{otherwise}
        \end{cases}
        \]
        and $\{\beta\} = \beta - \floor{\beta}$.
        
        Especially, taking $\alpha = 1$, we obtain a quantum algorithm for trace distance estimation using $\tilde O(r^5/\varepsilon^6)$ queries to $U_\rho$ and $U_\sigma$. 
    \end{theorem}
    
    The key observation of our quantum algorithm is that
    \begin{align*}
        T_\alpha(\rho, \sigma) 
        & = \tr\left(\abs{\frac{\rho - \sigma}{2}}^\alpha\right) \\
        & = \tr\left(\abs{\frac{\rho - \sigma}{2}}^{\alpha/2} \Pi_{\supp(\mu)} \abs{\frac{\rho - \sigma}{2}}^{\alpha/2} \right),
    \end{align*}
    where $\mu = \left( \rho + \sigma \right) / 2$, $\supp(\mu)$ is the support of $\mu$ and $\Pi_{S}$ is the projector onto a subspace $S$. A straightforward idea is to first prepare a subnormalized density operator that is a block-encoding of $\Pi_{\supp(\mu)}$, and then prepare another subnormalized density operator that is a block-encoding of $\abs{\nu}^{\alpha/ 2} \Pi_{\supp(\mu)} \abs{\nu}^{\alpha/ 2}$ by the evolution of subnormalized density operators (Lemma \ref{lemma:density-basic}), where $\nu = (\rho-\sigma)/2$. However, we are only able to prepare subnormalized density operator that is a block-encoding of a projector onto a truncated support by Lemma \ref{lemma:eigenvalue-threshold-projector}.
    In this way, we can prepare a subnormalized density operator $\abs{\nu}^{\alpha / 2} \Pi_{\supp_\delta(\mu)} \abs{\nu}^{\alpha / 2}$ for some $\delta > 0$, instead of $\abs{\nu}^{\alpha / 2} \Pi_{\supp(\mu)} \abs{\nu}^{\alpha / 2}$.
    
    When $\delta$ is fixed, the inherent error of this approach is
    \begin{proposition} \label{prop:inherent-error-trace-distance}
    \begin{align*}
        \abs{\tr\left(\abs{\nu}^{\alpha/ 2} \Pi_{\supp(\mu)} \abs{\nu}^{\alpha/ 2}\right) - \tr\left(\abs{\nu}^{\alpha/ 2} \Pi_{\supp_\delta(\mu)} \abs{\nu}^{\alpha/ 2}\right)} \leq 2 r \delta^{\min\{\alpha, 1\}/2}.
    \end{align*}
    \end{proposition}
    \begin{proof}
    Let $\mu = \sum_j \lambda_j \ket{\psi_j}\bra{\psi_j}$, where $\lambda_j \in [0, 1]$ and $\sum_j \lambda_j = 1$. Then $\Abs{\abs{\nu} \ket{\psi_j}} \leq \sqrt{\lambda_j}$. This is seen by the following. 
    \begin{align*}
        \Abs{\abs{\rho - \sigma} \ket{\psi_j}}^2
        & = \bra{\psi_j} (\rho - \sigma)^2 \ket{\psi_j} \\
        & = \bra{\psi_j} \rho^2 \ket{\psi_j} + \bra{\psi_j} \sigma^2 \ket{\psi_j} - \bra{\psi_j} (\rho\sigma + \sigma\rho) \ket{\psi_j} \\
        & \leq \bra{\psi_j} \rho^2 \ket{\psi_j} + \bra{\psi_j} \sigma^2 \ket{\psi_j} + \abs{\bra{\psi_j} \rho\sigma \ket{\psi_j}} + \abs{\bra{\psi_j} \sigma\rho \ket{\psi_j}} \\
        & \leq \bra{\psi_j} \rho^2 \ket{\psi_j} + \bra{\psi_j} \sigma^2 \ket{\psi_j} + 2 \Abs{\rho \ket{\psi_j}} \Abs{\sigma \ket{\psi_j}} \\
        & = \bra{\psi_j} \rho^2 \ket{\psi_j} + \bra{\psi_j} \sigma^2 \ket{\psi_j} + 2 \sqrt{\bra{\psi_j} \rho^2 \ket{\psi_j} \bra{\psi_j} \sigma^2 \ket{\psi_j}} \\
        & \leq 2 \left( \bra{\psi_j} \rho^2 \ket{\psi_j} + \bra{\psi_j} \sigma^2 \ket{\psi_j} \right) \\
        & \leq 2 \left( \bra{\psi_j} \rho \ket{\psi_j} + \bra{\psi_j} \sigma \ket{\psi_j} \right) \\
        & = 4 \bra{\psi_j} \mu \ket{\psi_j} \\
        & \leq 4 \Abs{\mu \ket{\psi_j}} \\
        & = 4 \lambda_j.
    \end{align*}
    Note that 
    \begin{align*}
        \abs{\tr\left(\abs{\nu}^{\alpha/ 2} \Pi_{\supp(\mu)} \abs{\nu}^{\alpha/ 2}\right) - \tr\left(\abs{\nu}^{\alpha/ 2} \Pi_{\supp_\delta(\mu)} \abs{\nu}^{\alpha/ 2}\right)}
        = \tr\left(\left( \Pi_{\supp(\mu)} - \Pi_{\supp_\delta(\mu)} \right) \abs{\nu}^{\alpha}\right).
    \end{align*}
    Then for $\alpha \geq 1$, we have
    \begin{align*}
        \tr\left(\left( \Pi_{\supp(\mu)} - \Pi_{\supp_\delta(\mu)} \right) \abs{\nu}^\alpha\right)
        & = \sum_{0 < \lambda_j < \delta} \bra{\psi_j} \abs{\nu}^\alpha \ket{\psi_j} \\
        & \leq \sum_{0 < \lambda_j < \delta} \bra{\psi_j} \abs{\nu} \ket{\psi_j} \\
        & \leq \sum_{0 < \lambda_j < \delta} \Abs{\abs{\nu} \ket{\psi_j}} \\
        & \leq \sum_{0 < \lambda_j < \delta} \sqrt{\lambda_j} \\
        & \leq \sqrt{\delta} \rank(\mu) \\
        & \leq 2 r \sqrt{\delta}.
    \end{align*}
    Before deriving bounds for $0 < \alpha < 1$, we need the following lemma.
    \begin{lemma} \label{lemma:norm-operator-power}
        Suppose that $A$ is an $n$-dimensional positive semidefinite operator, $\ket{\psi}$ is an $n$-dimensional vector, and $0 < \alpha < 1$. Then
        \[
            \Abs{A^\alpha\ket{\psi}} \leq \Abs{A\ket{\psi}}^{\alpha} \Abs{\ket{\psi}}^{1-\alpha}.
        \]
    \end{lemma}
    \begin{proof}
        Let
        \[
            A = \sum_{j = 1}^n \lambda_j \ket{\psi_j} \bra{\psi_j},
        \]
        where $\{\ket{\psi_j}\}$ is an orthonormal basis, and $\lambda_j \geq 0$ for all $1 \leq j \leq n$. Let
        \[
            \ket{\psi} = \sum_{j=1}^n \beta_j \ket{\psi_j}. 
        \]
        By H\"{o}lder's inequality, we have
        \begin{align*}
            \Abs{A^\alpha \ket{\psi}}^2
            & = \sum_{j=1}^n \abs{\lambda_j^\alpha \beta_j}^2 \\
            & = \sum_{j=1}^n \abs{\lambda_j^{2\alpha} \beta_j^{2\alpha}} \cdot \abs{\beta_j^{2(1-\alpha)}} \\
            & \leq \left( \sum_{j=1}^n \abs{\lambda_j^{2\alpha} \beta_j^{2\alpha}}^{1/\alpha} \right)^{\alpha} \left( \sum_{j=1}^n \abs{\beta_j^{2(1-\alpha)}}^{1/(1-\alpha)} \right)^{1-\alpha} \\ 
            & = \left( \sum_{j=1}^n \abs{\lambda_j \beta_j}^2 \right)^{\alpha} \left( \sum_{j=1}^n \abs{\beta_j}^{2} \right)^{1-\alpha} \\
            & = \Abs{A\ket{\psi}}^{2\alpha} \Abs{\ket{\psi}}^{2(1-\alpha)}.
        \end{align*}
    \end{proof}
    For $0 < \alpha < 1$, by Lemma \ref{lemma:norm-operator-power}, we have
    \begin{align*}
        \tr & \left(\left( \Pi_{\supp(\mu)} - \Pi_{\supp_\delta(\mu)} \right) \abs{\nu}^\alpha\right) \\
        & = \sum_{0 < \lambda_j < \delta} \bra{\psi_j} \abs{\nu}^\alpha \ket{\psi_j} \\
        & \leq \sum_{0 < \lambda_j < \delta} \Abs{\abs{\nu}^{\alpha} \ket{\psi_j}} \\
        & \leq \sum_{0 < \lambda_j < \delta} \Abs{\abs{\nu} \ket{\psi_j}}^{\alpha} \Abs{\ket{\psi_j}}^{1-\alpha} \\
        & \leq \sum_{0 < \lambda_j < \delta} \lambda_j^{\alpha/2} \\
        & \leq 2 r \delta^{\alpha/2}.
    \end{align*}
    
    Therefore,
    \begin{align*}
        & \abs{\tr\left(\abs{\nu}^{\alpha/ 2} \Pi_{\supp(\mu)} \abs{\nu}^{\alpha/ 2}\right) - \tr\left(\abs{\nu}^{\alpha/ 2} \Pi_{\supp_\delta(\mu)} \abs{\nu}^{\alpha/ 2}\right)}
        \leq 2 r \delta^{\min\{\alpha, 1\}/2}.
    \end{align*}
    \end{proof}
    Hence, we could obtain a reasonable error by setting $\delta$ small enough. 
    
    \textbf{Step 1}. By Lemma \ref{lemma:linear-combination} with the Hadamard gate
    \[
    H \ket{0} = \sqrt{\frac{1}{2}} \ket{0}+ \sqrt{\frac{1}{2}}\ket{1},
    \]
    we obtain an $O(n+n_\rho+n_\sigma)$-qubit unitary operator $U_\mu$ that prepares density operator $\mu = (\rho + \sigma) / 2$, using $1$ query to $U_\rho$ and $U_\sigma$, and $O(1)$ elementary quantum gates.

    Introducing three parameters $\delta_1, \varepsilon_1, \delta_Q \in (0, \frac 1 {10}]$, by Lemma \ref{lemma:eigenvalue-threshold-projector}, there is a quantum circuit $U_1$, which prepares a subnormalized density operator $A_1$ and $A_1$ is a $(1, 0, \delta_Q)$-block-encoding of $\Pi_1$, where
    \begin{align*}
        & \left( \frac {\delta_1} 4 (1 - 2\varepsilon_1) - \delta_1^{1/2}\varepsilon_1 \right) \Pi_{\supp_{2\delta_1}(\mu)} \leq \Pi_1 \leq \left( \frac {\delta_1} 4 + \varepsilon_1^2 + \delta_1^{1/2}\varepsilon_1 \right) \Pi_{\supp(\mu)}.
    \end{align*}
    Here, $U_1$ uses $Q_1$ queries to $U_\mu$ and $O(Q_1 (n+n_\rho+n_\sigma))$ elementary quantum gates, where $Q_1 = O\left(\frac{1}{\delta_1}\log\left(\frac{1}{\varepsilon_1}\right)\right)$. Moreover, $U_1$ can be computed by a classical Turing machine in $\poly \left(Q_1, \log\left( \frac 1 {\delta_Q}\right) \right)$ time.

    Next, we are going to construct a unitary operator that is a block-encoding of $\abs{\nu}^{\alpha / 2}$. By Lemma \ref{lemma:block-encoding of density operators}, there are two unitary operators $V_\sigma$ and $V_\rho$ which are a $(1, n+n_\sigma, 0)$-block-encoding of $\sigma$ and a $(1, n+n_\rho, 0)$-block-encoding of $\rho$, respectively. Here, $V_\sigma$ uses $1$ query to each of $U_\sigma$ and $U_\sigma^\dag$ and $n_\sigma$ elementary quantum gates, and $V_\rho$ uses $1$ query to each of $U_\rho$ and $U_\rho^\dag$ and $n_\rho$ elementary quantum gates. Let $n' = \max(n_\rho, n_\sigma)$, then $V_\rho \otimes I_{n' - n_\rho}$ and $V_\sigma \otimes I_{n' - n_\sigma}$ are $(1, n+n', 0)$-block-encodings of $\rho$ and $\sigma$, respectively.
    According to Definition \ref{def:state-preparation-pair}, we note that $(HX, H)$ is a $(2, 1, 0)$-state-preparation-pair for $y = (1, -1)$, where $H$ is the Hadamard gate and $X$ is the Pauli matrix.
    By Theorem \ref{thm:lcu}, there is a $(2n+n'+1)$-qubit quantum circuit $W$ which is a $(1, n+n'+1, 0)$-block-encoding of $\nu = (\rho - \sigma)/2$, using $1$ query to each of $V_\rho$ and $V_\sigma$ and $O(1)$ elementary quantum gates.

    Now we analyze the error if we replace $\Pi_{\supp(\mu)}$ by $4\delta_1^{-1}\Pi_1$ in the following.
    \begin{proposition} 
    \begin{align*}
        \abs{ \tr\left(\abs{\nu}^{\alpha/2} \left(4 \delta_1^{-1} \Pi_1\right) \abs{\nu}^{\alpha/2}\right) - T_{\alpha}(\rho, \sigma) } \leq \begin{cases}
        \Theta\left( \varepsilon_1 \delta_1^{-1/2} + r \delta_1^{1/2} \right), & \alpha \geq 1, \\
        \Theta\left( r^{1-\alpha} \varepsilon_1 \delta_1^{-1/2} + r \delta_1^{\alpha/2} \right), & 0 < \alpha < 1.
    \end{cases}
    \end{align*}
    \end{proposition}
    \begin{proof}
        Note that $\Pi_L \leq 4 \delta_1^{-1} \Pi_1 \leq \Pi_U$, where
        \[
            \Pi_L = \left(1 - 2\varepsilon_1 - 4\varepsilon_1\delta_1^{-1/2}\right) \Pi_{\supp_{2\delta_1}(\mu)},
        \]
        \[
            \Pi_U = \left(1 + 4\varepsilon_1^2\delta_1^{-1} + 4\varepsilon_1\delta_1^{-1/2}\right) \Pi_{\supp(\mu)}.
        \]
        This leads to 
        \[
            f(\Pi_L) \leq f \left(4 \delta_1^{-1} \Pi_1\right) \leq f(\Pi_U),
        \]
        where $f(\Pi) = \tr\left(\abs{\nu}^{\alpha/2} \Pi \abs{\nu}^{\alpha/2}\right)$ for convenience.
        Therefore, 
        \[
        \abs{ f\left(4 \delta_1^{-1} \Pi_1\right) - f\left(\Pi_{\supp(\mu)}\right) } \leq
        \max\{ T_L, T_U \},
        \]
        where
        \[
            T_L = \abs{ f\left(\Pi_L\right) - f\left(\Pi_{\supp(\mu)}\right) },
        \]
        \[
            T_U = \abs{ f\left(\Pi_U\right) - f\left(\Pi_{\supp(\mu)}\right) }.
        \]
        By Proposition \ref{prop:inherent-error-trace-distance}, we have
        \begin{align*}
            T_L 
            & \leq \abs{ f\left(\Pi_L\right) - f\left(\Pi_{\supp_{2\delta_1}(\mu)}\right) } + \abs{ f\left(\Pi_{\supp_{2\delta_1}(\mu)}\right) - f\left(\Pi_{\supp(\mu)}\right) } \\
            & \leq f\left( \left( 2\varepsilon_1 + 4\varepsilon_1\delta_1^{-1/2} \right) \Pi_{\supp_{2\delta_1}(\mu)}\right) + 2 r (2\delta_1)^{\min\{\alpha, 1\}/2} \\
            & \leq \Theta\left( \varepsilon_1 \delta_1^{-1/2} T_{\alpha}(\rho, \sigma) + r \delta_1^{\min\{\alpha, 1\}/2} \right).
        \end{align*}
        Also we have
        \begin{align*}
            T_U
            & \leq f\left( \left( 4 \varepsilon_1^2 \delta_1^{-1} + 4 \varepsilon_1 \delta_1^{-1/2} \right) \Pi_{\supp(\mu)} \right) \\
            & \leq \Theta\left( \varepsilon_1 \delta_1^{-1/2} T_{\alpha}(\rho, \sigma) \right).
        \end{align*}
        Combining the both, we have
        \begin{align*}
            \abs{f\left(4\delta_1^{-1}\Pi_1\right) - f\left(\Pi_{\supp(\mu)}\right) } \leq \Theta\left( \varepsilon_1 \delta_1^{-1/2} T_{\alpha}(\rho, \sigma) + r \delta_1^{\min\{\alpha, 1\}/2} \right).
        \end{align*}
        We note that $T_{\alpha}(\rho, \sigma) \leq 1$ if $\alpha \geq 1$, and $T_{\alpha}(\rho, \sigma) \leq (2r)^{1-\alpha}$ if $0 < \alpha < 1$. 
        These yield the claim. 
    \end{proof}
    
    In the following, we separately consider different cases of $\alpha$.
    
    \subsubsection*{Case 1: $\boldsymbol{\alpha}$ is an even integer}
    \qquad
    
    \textbf{Step 2}. By Lemma \ref{lemma:product-block}, there is a unitary operator $W_{\alpha/2}$, which is a $(1, O(\alpha(n+n')), 0)$-block-encoding of $\abs{\nu}^{\alpha/2}$, using $\alpha/2$ queries to $W$. By Lemma \ref{lemma:density-basic}, using $1$ query to each of $W_{\alpha/2}$ and $U_1$, we obtain a quantum circuit $\tilde U$, which prepares a subnormalized density operator $\abs{\nu}^{\alpha/2} A_1 \abs{\nu}^{\alpha/2}$. We note that
    \[
        \Abs{\abs{\nu}^{\alpha/2} A_1 \abs{\nu}^{\alpha/2} - \abs{\nu}^{\alpha / 2} \Pi_1 \abs{\nu}^{\alpha / 2}}
        \leq \Abs{A_1 - \Pi_1} \leq \delta_Q.
    \]
    
    \textbf{Step 3}. Introducing a parameter $\varepsilon_3$, by Lemma \ref{lemma:trace-estimation}, we can compute $\tilde p$ that estimates $\tr(\abs{\nu}^{\alpha/2} A_1 \abs{\nu}^{\alpha/2})$ such that $\abs{\tilde p - \tr(\abs{\nu}^{\alpha/2} A_1 \abs{\nu}^{\alpha/2})} \leq \varepsilon_3$ with $O\left( \frac{\sqrt{B}}{\varepsilon_3} + \frac{1}{\sqrt{\varepsilon_3}} \right)$ queries to $\tilde U$, where $B = \Theta\left( \delta_1+\varepsilon_1^2+\delta_1^{1/2}\varepsilon_1 + r\delta_Q \right)$ is an upper bound for $\tr(\abs{\nu}^{\alpha/2} A_1 \abs{\nu}^{\alpha/2})$. Note that
    \begin{align*}
        \tr(\abs{\nu}^{\alpha/2} A_1 \abs{\nu}^{\alpha/2})
        & \leq \tr\left(\abs{\nu}^{\alpha / 2} \Pi_1 \abs{\nu}^{\alpha / 2}\right) + \Theta\left( r \delta_Q \right) \\
        & \leq \Theta\left( \delta_1+\varepsilon_1^2+\delta_1^{1/2}\varepsilon_1 + r\delta_Q \right).
    \end{align*}

    \textbf{Step 4}. Output $4 \delta_1^{-1} \tilde p \approx T_\alpha(\rho, \sigma)$ as the estimation. The additive error is
    \begin{align*}
        \abs{ 4 \delta_1^{-1} \tilde p - T_\alpha(\rho, \sigma)} 
        & \leq 4 \delta_1^{-1} \abs{\tilde p - \tr(\abs{\nu}^{\alpha/2} A_1 \abs{\nu}^{\alpha/2})} + 4 \delta_1^{-1} \abs{ \tr(\abs{\nu}^{\alpha/2} A_1 \abs{\nu}^{\alpha/2}) - \tr(\abs{\nu}^{\alpha/2} \Pi_1 \abs{\nu}^{\alpha/2}) } \\
        & \qquad + \Big| \tr\left(\abs{\nu}^{\alpha/2} \left(4 \delta_1^{-1} \Pi_1\right) \abs{\nu}^{\alpha/2}\right) - \tr\left(\abs{\nu}^{\alpha/2} \Pi_{\supp(\mu)} \abs{\nu}^{\alpha/2}\right) \Big| \\
        & \leq \Theta\left( r \delta_1^{1/2} + \varepsilon_1 \delta_1^{-1/2} + \delta_1^{-1}(\varepsilon_3 + r\delta_Q) \right).
    \end{align*}
    In order to make $\abs{ 4 \delta_1^{-1} \tilde p - T_\alpha(\rho, \sigma)} \leq \varepsilon$, it is sufficient to let $\delta_1 = \Theta(\varepsilon^2/r^2)$, $\varepsilon_3 = \Theta(\varepsilon^3 / r^2)$,
    and other parameters become small enough polynomials in $r$ and $1/\varepsilon$. Under these conditions, the number of queries to $U_\rho$ and $U_\sigma$ is
        \[
        Q = O\left( Q_1 \alpha \left(\frac{\sqrt{B}}{\varepsilon_3} + \frac{1}{\sqrt{\varepsilon_3}}\right) \right) = \tilde O\left( \frac{r^3}{\varepsilon^4} \right),
        \]
        and the number of elementary quantum gates is $O(Q \cdot \poly(n))$. 
    
    \subsubsection*{Case 2: $\boldsymbol{\alpha \geq 1}$ and $\boldsymbol{\alpha}$ is not an even integer}
    \qquad
    
    \textbf{Step 2}. Now introducing two parameters $\delta_2, \varepsilon_2 \in (0, \frac 1 4]$, by Lemma \ref{lemma:positive-power-unitary}, there is a quantum circuit $W_{\{\alpha/2\}}$ that is a $(1, O(n+n'), 0)$-block-encoding of $A_2$, and $A_2$ is a $(2, 0, \Theta(\varepsilon_2 + \delta_2^{\{\alpha / 2\}}))$-block-encoding of $\abs{\nu}^{\{\alpha/2\}}$, using $Q_2$ queries to $W$ and $O\left(Q_2 (n+n')\right)$ elementary quantum gates, where $Q_2 = O\left(\frac{1}{\delta_2} \log\left(\frac 1 {\varepsilon_2}\right)\right)$.
    By Lemma \ref{lemma:product-block}, there is a unitary operator $W_{\floor{\alpha/2}}$, which is a $(1, O(\alpha(n+n')), 0)$-block-encoding of $\abs{\nu}^{\floor{\alpha/2}}$, using $\floor{\alpha/2}$ queries to $W$. Again by Lemma \ref{lemma:product-block}, there is a unitary operator $W_{\alpha/2}$, which is a $(2, O(\alpha(n+n')), \Theta(\varepsilon_2 + \delta_2^{\{\alpha/2\}}))$-block-encoding of $\abs{\nu}^{\alpha/2}$, using $1$ query to $W_{\{\alpha/2\}}$ and $1$ query to $W_{\floor{\alpha/2}}$.
    By Lemma \ref{lemma:density-basic}, using $1$ query to each of $W_{\alpha/2}$ and $U_1$, we obtain a quantum circuit $\tilde U$, which prepares a subnormalized density operator $A_2A_1 A_2^\dag$. We note that
    \begin{align*}
        \Abs{A_2A_1 A_2^\dag - \frac 1 4 \abs{\nu}^{\alpha / 2} \Pi_1 \abs{\nu}^{\alpha / 2}} 
        & \leq \Abs{A_2 A_1 A_2^\dag - A_2 \Pi_1 A_2^\dag} + \Abs{A_2 \Pi_1 A_2^\dag - \frac 1 4 \abs{\nu}^{\alpha / 2} \Pi_1 \abs{\nu}^{\alpha / 2}} \\
        & \leq \Abs{A_1 - \Pi_1} + \Abs{A_2 - \frac 1 2 \abs{\nu}^{\alpha/2}} \Abs{\Pi_1} \Abs{A_2^\dag} \\
        & \qquad + \Abs{A_2^\dag - \frac 1 2 \abs{\nu}^{\alpha/2}} \Abs{\Pi_1} \Abs{\frac 1 2 \abs{\nu}^{\alpha / 2}} \\
        & \leq \Theta\left( \delta_Q + (\delta_1+\varepsilon_1^2+\delta_1^{1/2}\varepsilon_1) (\varepsilon_2 + \delta_2^{\{\alpha/2\}}) \right).
    \end{align*}

    \textbf{Step 3}. Introducing a parameter $\varepsilon_3$, by Lemma \ref{lemma:trace-estimation}, we can compute $\tilde p$ that estimates $\tr(A_2A_1 A_2^\dag)$ such that $$\abs{\tilde p - \tr(A_2A_1 A_2^\dag)} \leq \varepsilon_3$$ with $O\left( \frac{\sqrt{B}}{\varepsilon_3} + \frac{1}{\sqrt{\varepsilon_3}} \right)$ queries to $\tilde U$, where $$B = \Theta\left( (\delta_1+\varepsilon_1^2+\delta_1^{1/2}\varepsilon_1)(1+r(\varepsilon_2+\delta_2^{\{\alpha/2\}})) + r\delta_Q \right)$$ is an upper bound for $\tr(A_2A_1 A_2^\dag)$. Note that
    \begin{align*}
        \tr(A_2A_1 A_2^\dag) 
        & \leq \tr\left(\frac 1 4 \abs{\nu}^{\alpha / 2} \Pi_1 \abs{\nu}^{\alpha / 2}\right) + \Theta\left( r \left( \delta_Q + (\delta_1+\varepsilon_1^2+\delta_1^{1/2}\varepsilon_1) (\varepsilon_2 + \delta_2^{\{\alpha/2\}}) \right) \right) \\
        & \leq \Theta\left( (\delta_1+\varepsilon_1^2+\delta_1^{1/2}\varepsilon_1)(1+r(\varepsilon_2+\delta_2^{\{\alpha/2\}})) + r\delta_Q \right).
    \end{align*}

    \textbf{Step 4}. Output $4 \delta_1^{-1} \tilde p \approx T_\alpha(\rho, \sigma)$ as the estimation. The additive error is
    \begin{align*}
        \abs{ 4 \delta_1^{-1} \tilde p - T_\alpha(\rho, \sigma)} 
        & \leq 4 \delta_1^{-1} \abs{\tilde p - \tr(A_2A_1 A_2^\dag)} + 4 \delta_1^{-1} \abs{ \tr(A_2A_1 A_2^\dag) - \tr(\abs{\nu}^{\alpha/2} \Pi_1 \abs{\nu}^{\alpha/2}) } \\
        & \qquad + \Big| \tr\left(\abs{\nu}^{\alpha/2} \left(4 \delta_1^{-1} \Pi_1\right) \abs{\nu}^{\alpha/2}\right) - \tr\left(\abs{\nu}^{\alpha/2} \Pi_{\supp(\mu)} \abs{\nu}^{\alpha/2}\right) \Big| \\
        & \leq \Theta\Big( r \delta_1^{1/2} + \varepsilon_1 \delta_1^{-1/2} + r(\varepsilon_2+\delta_2^{\{\alpha/2\}})(1+\varepsilon_1\delta_1^{-1/2}) + \delta_1^{-1}(\varepsilon_3 + r\delta_Q) \Big).
    \end{align*}
    In order to make $\abs{ 4 \delta_1^{-1} \tilde p - T_\alpha(\rho, \sigma)} \leq \varepsilon$, it is sufficient to let $\delta_1 = \Theta(\varepsilon^2/r^2)$, $\delta_2 = \Theta( (\varepsilon / r)^{1/\{\alpha/2\}} )$, $\varepsilon_3 = \Theta(\varepsilon^3 / r^2)$, and other parameters become small enough polynomials in $r$ and $1/\varepsilon$. Under these conditions, the number of queries to $U_\rho$ and $U_\sigma$ is
    \begin{align*}
        Q & = O\left( Q_1 (Q_2 + \alpha) \left(\frac{\sqrt{B}}{\varepsilon_3} + \frac{1}{\sqrt{\varepsilon_3}}\right) \right)\\
        & = \tilde O\left( \frac{r^{3+1/\{\alpha/2\}}}{\varepsilon^{4+1/\{\alpha/2\}}} \right),
    \end{align*}
    and the number of elementary quantum gates is $O(Q \cdot \poly(n))$.
        
    \subsubsection*{Case 3: $\boldsymbol{0 < \alpha < 1}$}
    \qquad
    
    \textbf{Step 2}. Introducing two parameters $\delta_2, \varepsilon_2 \in (0, \frac 1 4]$, by Lemma \ref{lemma:positive-power-unitary}, there is a quantum circuit $W_{\alpha/2}$ that is a $(1, O(n+n'), 0)$-block-encoding of $A_2$, and $A_2$ is a $(2, 0, \Theta(\varepsilon_2 + \delta_2^{\alpha / 2}))$-block-encoding of $\abs{\nu}^{\alpha/2}$, using $Q_2$ queries to $W$ and $O\left(Q_2 (n+n')\right)$ elementary quantum gates, where $Q_2 = O\left(\frac{1}{\delta_2} \log\left(\frac 1 {\varepsilon_2}\right)\right)$.
    By Lemma \ref{lemma:density-basic}, using $1$ query to each of $W_{\alpha/2}$ and $U_1$, we obtain a quantum circuit $\tilde U$, which prepares a subnormalized density operator $A_2A_1 A_2^\dag$. We note that
    \begin{align*}
        \Abs{A_2A_1 A_2^\dag - \frac 1 4 \abs{\nu}^{\alpha / 2} \Pi_1 \abs{\nu}^{\alpha / 2}}
        & \leq \Abs{A_2 A_1 A_2^\dag - A_2 \Pi_1 A_2^\dag} + \Abs{A_2 \Pi_1 A_2^\dag - \frac 1 4 \abs{\nu}^{\alpha / 2} \Pi_1 \abs{\nu}^{\alpha / 2}} \\
        & \leq \Abs{A_1 - \Pi_1} + \Abs{A_2 - \frac 1 2 \abs{\nu}^{\alpha/2}} \Abs{\Pi_1} \Abs{A_2^\dag} \\
        & \qquad + \Abs{A_2^\dag - \frac 1 2 \abs{\nu}^{\alpha/2}} \Abs{\Pi_1} \Abs{\frac 1 2 \abs{\nu}^{\alpha / 2}} \\
        & \leq \Theta\left( \delta_Q + (\delta_1+\varepsilon_1^2+\delta_1^{1/2}\varepsilon_1) (\varepsilon_2 + \delta_2^{\alpha/2}) \right).
    \end{align*}

    \textbf{Step 3}. Introducing a parameter $\varepsilon_3$, by Lemma \ref{lemma:trace-estimation}, we can compute $\tilde p$ that estimates $\tr(A_2A_1 A_2^\dag)$ such that $$\abs{\tilde p - \tr(A_2A_1 A_2^\dag)} \leq \varepsilon_3$$ with $O\left( \frac{\sqrt{B}}{\varepsilon_3} + \frac{1}{\sqrt{\varepsilon_3}} \right)$ queries to $\tilde U$, where $$B = \Theta\left( r^{1-\alpha} \delta_1 + r (\delta_1+\varepsilon_1^2+\delta_1^{1/2}\varepsilon_1)(\varepsilon_2+\delta_2^{\alpha/2}) + r\delta_Q \right)$$ is an upper bound for $\tr(A_2A_1 A_2^\dag)$. Note that
    \begin{align*}
        \tr(A_2A_1 A_2^\dag) 
        & \leq \tr\left(\frac 1 4 \abs{\nu}^{\alpha / 2} \Pi_1 \abs{\nu}^{\alpha / 2}\right) + \Theta\left( r \left( \delta_Q + (\delta_1+\varepsilon_1^2+\delta_1^{1/2}\varepsilon_1) (\varepsilon_2 + \delta_2^{\alpha/2}) \right) \right) \\
        & \leq \Theta\left( r^{1-\alpha} \delta_1 + r (\delta_1+\varepsilon_1^2+\delta_1^{1/2}\varepsilon_1)(\varepsilon_2+\delta_2^{\alpha/2}) + r\delta_Q \right).
    \end{align*}

    \textbf{Step 4}. Output $4 \delta_1^{-1} \tilde p \approx T_\alpha(\rho, \sigma)$ as the estimation. The additive error is
    \begin{align*}
        \abs{ 4 \delta_1^{-1} \tilde p - T_\alpha(\rho, \sigma)} 
        & \leq 4 \delta_1^{-1} \abs{\tilde p - \tr(A_2A_1 A_2^\dag)} + 4 \delta_1^{-1} \abs{ \tr(A_2A_1 A_2^\dag) - \tr(\abs{\nu}^{\alpha/2} \Pi_1 \abs{\nu}^{\alpha/2}) } \\
        & \qquad + \Big| \tr\left(\abs{\nu}^{\alpha/2} \left(4 \delta_1^{-1} \Pi_1\right) \abs{\nu}^{\alpha/2}\right) - \tr\left(\abs{\nu}^{\alpha/2} \Pi_{\supp(\mu)} \abs{\nu}^{\alpha/2}\right) \Big| \\
        & \leq \Theta\Big( r \delta_1^{\alpha/2} + r^{1-\alpha}\varepsilon_1 \delta_1^{-1/2} + r(\varepsilon_2+\delta_2^{\alpha/2})(1+\varepsilon_1\delta_1^{-1/2}) + \delta_1^{-1}(\varepsilon_3 + r\delta_Q) \Big).
    \end{align*}
    In order to make $\abs{ 4 \delta_1^{-1} \tilde p - T_\alpha(\rho, \sigma)} \leq \varepsilon$, it is sufficient to let $\delta_1 = \Theta((\varepsilon / r)^{2/\alpha})$, $\delta_2 = \Theta( (\varepsilon / r)^{2/\alpha} )$, $\varepsilon_3 = \Theta(\varepsilon^{2/\alpha+1} / r^{2/\alpha})$, and other parameters become small enough polynomials in $r$ and $1/\varepsilon$. Under these conditions, the number of queries to $U_\rho$ and $U_\sigma$ is
    \[
    Q = O\left( Q_1 Q_2 \left(\frac{\sqrt{B}}{\varepsilon_3} + \frac{1}{\sqrt{\varepsilon_3}}\right) \right) = \tilde O\left( \frac{r^{5/\alpha + (1-\alpha)/2}}{\varepsilon^{5/\alpha+1}} \right),
    \]
    and the number of elementary quantum gates is $O(Q \cdot \poly(n))$.
    
    \subsection{Fidelity}
    \label{sec:fidelity}

    Quantum fidelity estimation is a problem to compute the fidelity of two mixed quantum states given their purifications. The well-known SWAP test \cite{BCW01} can solve this problem when one of the quantum states is pure. Recently, a polynomial-time quantum algorithm was proposed in \cite{WZC+22} for the case that one of the quantum states is low-rank. However, their algorithm has very large exponents of $r$ (rank) and $\varepsilon$ (additive error) in the complexity. Here, we are able to improve the complexity with much smaller exponents, compared to the $\tilde O(r^{12.5}/\varepsilon^{13.5})$ quantum algorithm for fidelity estimation proposed by \cite{WZC+22}.

    In addition, the sandwiched quantum R\'{e}nyi relative entropy $D_{\alpha}(\rho \| \sigma)$ \cite{WWY14, MDS13} is a generalization of quantum state measures, defined by
    \[
        F_{\alpha}(\rho, \sigma) = \exp\left((\alpha - 1) D_{\alpha}(\rho \| \sigma)\right) = \tr\left( \left( \sigma^{\frac{1-\alpha}{2\alpha}} \rho \sigma^{\frac{1-\alpha}{2\alpha}} \right)^{\alpha} \right).
    \]
    Here, $F_{1/2}(\rho, \sigma) = F(\rho, \sigma)$ is the quantum state fidelity. It is clear that $0 \leq F_\alpha(\rho, \sigma) \leq 1$ for $\alpha \in (0, 1)$ (see \cite{MDS13}). Recently, the sandwiched quantum R\'{e}nyi relative entropy is used in quantum machine learning \cite{KMW21}.
    
    Our quantum algorithms for computing the $\alpha$-fidelity are given as follows.

    \begin{theorem} \label{thm:fidelity}
        Suppose that
        \begin{enumerate}
          \item $U_\rho$ is an $(n+n_\rho)$-qubit unitary operator that prepares an $n$-qubit density operator $\rho$ with $\rank(\rho) = r$.
          \item $U_\sigma$ is an $(n+n_\sigma)$-qubit unitary operator that prepares an $n$-qubit density operator $\sigma$.
          \item $n_\rho$ and $n_\sigma$ are polynomials in $n$.
        \end{enumerate}
        For $\alpha \in (0, 1)$, there is a quantum algorithm that computes $F_\alpha(\rho, \sigma)$ within additive error $\varepsilon$, using $\tilde O\left({r^{\frac{3-\alpha}{2\alpha}}}/{\varepsilon^{\frac{3+\alpha}{2\alpha}}}\right)$ queries to $U_\rho$, $Q$ queries to $U_\sigma$, and $Q \cdot \poly(n)$ elementary quantum gates, where
        \[
            Q = \begin{cases}
                \tilde O\left( r^{\frac{3-\alpha}{2\alpha}}/\varepsilon^{\frac{3+\alpha}{2\alpha}} \right), & \beta \in \mathbb{N}, \\
                \tilde O\left( r^{\frac{3-\alpha}{2\alpha} + \frac{1}{\alpha\{\beta\}}}/\varepsilon^{\frac{3+\alpha}{2\alpha} + \frac{1}{\alpha\{\beta\}}}  \right), & \beta \notin \mathbb{N},
            \end{cases}
        \]
        and $\beta = (1-\alpha)/2\alpha$, $\{\beta\} = \beta - \floor{\beta}$.

        Especially, taking $\alpha = \frac 1 2$, we obtain a quantum algorithm for fidelity estimation using $\tilde O\left({r^{2.5}}/{\varepsilon^{3.5}}\right)$ queries to $U_\rho$ and $\tilde O\left({r^{6.5}}/{\varepsilon^{7.5}}\right)$ queries to $U_\sigma$.
    \end{theorem}

    We put the detailed proofs into the following subsubsections. 
    In fact, such techniques used in estimating the relative sandwiched R\'{e}nyi entropy can also be used to compute the relative R\'{e}nyi entropy.

    \subsubsection*{Case 1: $\boldsymbol{\beta}$ is an integer}
    \qquad

    \textbf{Step 1}. By Lemma \ref{lemma:block-encoding of density operators}, there is a unitary operator $U_1$, which is a $(1, n + n_\sigma, 0)$-block-encoding of $\sigma$, using $O(1)$ queries to $U_\sigma$ and $O(n_\sigma)$ elementary quantum gates. By Lemma \ref{lemma:product-block}, there is a unitary operator $U_\beta$, which is a $(1, O(\beta(n+n_\sigma)), 0)$-block-encoding of $\sigma^\beta$, using $\beta$ queries to $U_1$. By Lemma \ref{lemma:density-basic}, there is a unitary operator $U_\eta$, which prepares a subnormalized density operator $\eta = \sigma^\beta \rho \sigma^\beta$, using $1$ query to $U_\beta$ and $1$ query to $U_\rho$.

    Now introducing two parameters $\delta_1$ and $\varepsilon_1$, by Lemma \ref{lemma:positive-power-density}, there is a unitary operator $\tilde U$, which prepares a subnormalized density operator $\eta'$ and $\eta'$ is a $(4\delta_1^{\alpha-1}, 0, \Theta(\delta_1^\alpha + \varepsilon_1 \delta_1^{\alpha-1}))$-block-encoding of $\eta^\alpha$, using $O(d_1)$ queries to $U_\eta$, where $d_1 = O(\frac{1}{\delta_1}\log\frac{1}{\varepsilon_1})$.

    \textbf{Step 2}. Introducing a parameter $\varepsilon_2$, by Lemma \ref{lemma:trace-estimation}, we can compute $\tilde p$ such that $\abs{\tilde p - \tr(\eta')} \leq \varepsilon_2$, using $O\left(\frac{\sqrt{B}}{\varepsilon_2} + \frac{1}{\sqrt{\varepsilon_2}}\right)$ queries to $\tilde U$, where $B = \Theta\left(\delta_1^{1-\alpha} + r(\delta_1 + \varepsilon_1)\right)$ is an upper bound for $\tr(\eta')$. Note that
    \begin{align*}
        \tr(\eta')
    & \leq \frac 1 4 \delta_1^{1-\alpha} \tr\left(\eta^{\alpha}\right) + \Theta\left(r(\delta_1+\varepsilon_1)\right) \\
    & \leq \Theta\left(\delta_1^{1-\alpha} + r(\delta_1 + \varepsilon_1)\right).
    \end{align*}

    \textbf{Step 3}. Output $4 \delta_1^{\alpha - 1} \tilde p \approx F_\alpha(\rho, \sigma)$ as the estimation. The additive error is
    \begin{align*}
    \abs{4 \delta_1^{\alpha - 1} \tilde p - F_{\alpha}(\rho, \sigma)}
    & \leq 4 \delta_1^{\alpha - 1} \abs{\tilde p - \tr(\eta')} + \abs{ \tr(4\delta_1^{\alpha - 1}\eta') - \tr(\eta^{\alpha}) } \\
    & \leq \Theta\left( r(\delta_1^\alpha + \varepsilon_1 \delta_1^{\alpha - 1}) + \delta_1^{\alpha - 1} \varepsilon_2 \right).
    \end{align*}
    In order to make $\abs{4 \delta_1^{\alpha - 1} \tilde p - F_{\alpha}(\rho, \sigma)} \leq \varepsilon$, it is sufficient to let $\delta_1 = \Theta((\varepsilon/r)^{1/\alpha})$, $\varepsilon_1 = \Theta((\varepsilon/r)^{1/\alpha})$, and $\varepsilon_2 = \Theta(\varepsilon^{1/\alpha}/r^{1/\alpha - 1})$. Under these conditions, the number of queries to $U_\sigma$ is
    \[
        O\left(\beta d_1 \left(\frac{\sqrt{B}}{\varepsilon_2} + \frac{1}{\sqrt{\varepsilon_2}}\right) \right) = \tilde O\left( \frac{r^{\frac{3-\alpha}{2\alpha}}}{\varepsilon^{\frac{3+\alpha}{2\alpha}}} \right),
    \]
    and the number of queries to $U_\rho$ is
    \[
        O\left( d_1 \left(\frac{\sqrt{B}}{\varepsilon_2} + \frac{1}{\sqrt{\varepsilon_2}}\right) \right) = \tilde O\left(  \frac{r^{\frac{3-\alpha}{2\alpha}}}{\varepsilon^{\frac{3+\alpha}{2\alpha}}} \right),
    \]
    and the number of elementary quantum gates is
    \begin{align*}
        & O\left(\beta d_1 \left(\frac{\sqrt{B}}{\varepsilon_2} + \frac{1}{\sqrt{\varepsilon_2}}\right) \right) \cdot \poly(n, n_\sigma, n_\rho) = \tilde O\left( \frac{r^{\frac{3-\alpha}{2\alpha}}}{\varepsilon^{\frac{3+\alpha}{2\alpha}}} \poly(n) \right).
    \end{align*}

    \subsubsection*{Case 2: $\boldsymbol{\beta}$ is not an integer}

    Let $\{\beta\} = \beta - \floor{\beta}$ denote the decimal part of $\beta$.

    \textbf{Step 1}. By Lemma \ref{lemma:block-encoding of density operators}, there is a unitary operator $U_1$, which is a $(1, n+n_\sigma, 0)$-block-encoding of $\sigma$, using $O(1)$ queries to $U_\sigma$ and $O(n_\sigma)$ elementary quantum gates. By Lemma \ref{lemma:positive-power-unitary}, introducing two parameters $\delta_1$ and $\varepsilon_1$, there is a unitary operator $U_{\{\beta\}}$, which is a $(1, O(n+n_\sigma), 0)$-block-encoding of $A_1$, using $O(Q_1)$ queries to $U_1$ and $O(Q_1 (n+n_\sigma))$ elementary quantum gates, where $Q_1 = O\left(\frac{1}{\delta_1} \log\frac{1}{\varepsilon_1}\right)$ and $A_1$ is a $(2, 0, \Theta(\varepsilon_1 + \delta_1^{\{\beta\}}))$-block-encoding of $\sigma^{\{\beta\}}$.
    By Lemma \ref{lemma:product-block}, there is a unitary operator $U_{\floor{\beta}}$, which is a $(1, O(\beta(n+n_\sigma)), 0)$-block-encoding of $\sigma^{\floor{\beta}}$, using $\floor{\beta}$ queries to $U_1$. Again by Lemma \ref{lemma:product-block}, there is a unitary operator $U_\beta$, which is a $(2, O(\beta(n+n_\sigma)), 0)$-block-encoding of $A_1 \sigma^{\floor{\beta}}$, using $1$ query to $U_{\floor{\beta}}$ and $1$ query to $U_{\{\beta\}}$.
    By Lemma \ref{lemma:density-basic}, there is a unitary operator $\tilde U$, which prepares a subnormalized density operator $A_1 \sigma^{\floor{\beta}} \rho \sigma^{\floor{\beta}} A_1^\dag$, using $1$ query to $\tilde U$ and $1$ query to $U_\rho$. Note that
    \begin{align*}
        \Abs{A_1 \sigma^{\floor{\beta}} \rho \sigma^{\floor{\beta}} A_1^\dag - \frac{1}{4} \sigma^\beta \rho \sigma^\beta} 
        & \leq \Abs{A_1 \sigma^{\floor{\beta}} - \frac 1 2 \sigma^\beta} \Abs{\rho} \Abs{\sigma^{\floor{\beta}} A_1^\dag} + \Abs{\sigma^{\floor{\beta}} A_1^\dag - \frac 1 2 \sigma^\beta} \Abs{\rho} \Abs{\frac 1 2 \sigma^\beta} \\
        & \leq \Theta\left(\varepsilon_1 + \delta_1^{\{\beta\}}\right).
    \end{align*}

    \textbf{Step 2}. By Lemma \ref{lemma:positive-power-density}, introducing two parameters $\delta_2$ and $\varepsilon_2$, there is a unitary operator $U_2$, which prepares a subnormalized density operator $A_2$, using $O(Q_2)$ queries to $\tilde U$ and $O(Q_2 (n+n_\sigma+n_\rho))$ elementary quantum gates, where $Q_2 = O\left(\frac{1}{\delta_2}\log\frac{1}{\varepsilon_2}\right)$ and $A_2$ is a $(4\delta_2^{\alpha-1}, 0, \Theta(\delta_2^{\alpha-1}(\delta_2+\varepsilon_2)))$-block-encoding of $\left(A_1 \sigma^{\floor{\beta}} \rho \sigma^{\floor{\beta}} A_1^\dag\right)^\alpha$.

    In order to analysis the error, we need the following lemma.
    \begin{lemma} \label{lemma:perturbation}
        Suppose that $A$ and $B$ are two positive semidefinite operators of rank $\leq r$, and $0 < \alpha < 1$. Then
        \[
            \abs{\tr(A^{\alpha}) - \tr(B^{\alpha})} \leq 3 r \Abs{A-B}^{\alpha}.
        \]
    \end{lemma}
    \begin{proof}
        Let $J = A - B$. Let the eigenvalues of $A$, $B$ and $J$ be
        \begin{align*}
            \mu_1 \geq \mu_2 \geq \dots \geq \mu_N, \\
            \nu_1 \geq \nu_2 \geq \dots \geq \nu_N, \\
            \xi_1 \geq \xi_2 \geq \dots \geq \xi_N,
        \end{align*}
        respectively. Then we have $\mu_{r+1} = \dots = \mu_{N} = \nu_{r+1} = \dots = \nu_{N} = 0$. By Weyl's inequality \cite{Wey12}, we have
        \[
            \nu_j - \Abs{J} \leq \nu_j + \xi_N \leq \mu_j \leq \nu_j + \xi_1 \leq \nu_j + \Abs{J}
        \]
        for every $1 \leq j \leq N$. Furthermore, it holds that $\abs{\mu_j^\alpha - \nu_j^\alpha} \leq 5 \Abs{J}^\alpha$. This is seen by the following two cases.
        \begin{enumerate}
          \item $\nu_j \geq \Abs{J}$. In this case, $\nu_j^\alpha - \Abs{J}^\alpha \leq (\nu_j - \Abs{J})^\alpha \leq \mu_j^\alpha \leq (\nu_j + \Abs{J})^\alpha \leq \nu_j^\alpha + \Abs{J}^\alpha$. Then we obtain that $\abs{\mu_j^\alpha - \nu_j^\alpha} \leq \Abs{J}^\alpha$.
          \item $\nu_j < \Abs{J}$. In this case, $\abs{\mu_j^\alpha - \nu_j^\alpha} \leq \abs{\mu_j}^\alpha + \abs{\nu_j}^\alpha \leq \abs{\nu_j + \Abs{J}}^\alpha + \abs{\nu_j}^\alpha < (2^\alpha + 1) \Abs{J}^\alpha < 3 \Abs{J}^\alpha$.
        \end{enumerate}
        These yield that
        \begin{align*}
            \abs{\tr(A^{\alpha}) - \tr(B^{\alpha})} 
            & = \abs{\sum_{j=1}^N \mu_j^\alpha - \sum_{j=1}^N \nu_j^\alpha} \\
            & \leq \sum_{j=1}^r \abs{\mu_j^\alpha - \nu_j^\alpha} \\
            & \leq 3 r \Abs{J}^\alpha.
        \end{align*}
    \end{proof}

    By Lemma \ref{lemma:perturbation}, we have
    \begin{align*}
        \abs{\tr\left(\left(A_1 \sigma^{\floor{\beta}} \rho \sigma^{\floor{\beta}} A_1^\dag\right)^\alpha\right) - \tr\left(\left(\frac{1}{4} \sigma^\beta \rho \sigma^\beta\right)^{\alpha}\right)} 
        & \leq 3 r \Abs{A_1 \sigma^{\floor{\beta}} \rho \sigma^{\floor{\beta}} A_1^\dag - \frac{1}{4} \sigma^\beta \rho \sigma^\beta}^{\alpha} \\
        & \leq \Theta\left( r \left(\varepsilon_1 + \delta_1^{\{\beta\}}\right)^\alpha \right).
    \end{align*}

    \textbf{Step 3}. Introducing a parameter $\varepsilon_3$, by Lemma \ref{lemma:trace-estimation}, we can compute $\tilde p$ such that $\abs{\tilde p - \tr(A_2)} \leq \varepsilon_3$, using $O\left(\frac{\sqrt{B}}{\varepsilon_3} + \frac{1}{\sqrt{\varepsilon_3}}\right)$ queries to $U_2$, where $$B = \Theta\left( \delta_2^{1-\alpha} + r\left(\varepsilon_1+\delta_1^{\{\beta\}}\right)^{\alpha} + r(\delta_2+\varepsilon_2) \right)$$ is an upper bound for $\tr(A_2)$. Note that
    \begin{align*}
        \tr(A_2)
        & \leq \frac{1}{4} \delta_2^{1-\alpha} \tr\left(\left(A_1 \sigma^{\floor{\beta}} \rho \sigma^{\floor{\beta}} A_1^\dag\right)^\alpha\right) + \Theta(r(\delta_2+\varepsilon_2)) \\
        & \leq \frac{1}{4} \delta_2^{1-\alpha} \left( \tr\left(\left(\frac{1}{4} \sigma^\beta \rho \sigma^\beta\right)^{\alpha}\right) + \Theta\left(r\left(\varepsilon_1+\delta_1^{\{\beta\}}\right)^{\alpha}\right) \right) + \Theta(r(\delta_2+\varepsilon_2)) \\
        & \leq \Theta\left( \delta_2^{1-\alpha} + r\left(\varepsilon_1+\delta_1^{\{\beta\}}\right)^{\alpha} + r(\delta_2+\varepsilon_2) \right).
    \end{align*}

    \textbf{Step 4}. Output $4^{\alpha+1}\delta_2^{\alpha-1}\tilde p \approx F_\alpha(\rho, \sigma)$ as the estimation. The additive error is
    \begin{align*}
        \abs{4^{\alpha+1}\delta_2^{\alpha-1}\tilde p - F_\alpha(\rho, \sigma)} 
        & \leq 4^{\alpha+1}\delta_2^{\alpha-1} \abs{\tilde p - \tr(A_2)} + 4^\alpha \abs{\tr(4\delta_2^{\alpha-1} A_2) - \tr\left(\left(A_1 \sigma^{\floor{\beta}} \rho \sigma^{\floor{\beta}} A_1^\dag\right)^\alpha\right)} \\
        & \qquad + 4^\alpha \abs{ \tr\left(\left(A_1 \sigma^{\floor{\beta}} \rho \sigma^{\floor{\beta}} A_1^\dag\right)^\alpha\right) - \tr\left( \left(\frac{1}{4} \sigma^\beta \rho \sigma^\beta\right)^\alpha \right)} \\
        & \leq \Theta\left( r\left(\varepsilon_1 + \delta_1^{\{\beta\}}\right)^\alpha + r\delta_2^{\alpha-1}(\varepsilon_2 + \delta_2) + \delta_2^{\alpha-1} \varepsilon_3 \right).
    \end{align*}

    In order to make $\abs{4^{\alpha+1}\delta_2^{\alpha-1}\tilde p - F_\alpha(\rho, \sigma)} \leq \varepsilon$, it is sufficient to let $\delta_1 = \Theta((\varepsilon/r)^{1/\alpha\{\beta\}})$, $\varepsilon_1 = \Theta((\varepsilon/r)^{1/\alpha\{\beta\}})$, $\delta_2 = \Theta\left((\varepsilon/r)^{1/\alpha}\right)$, $\varepsilon_2 = \Theta\left((\varepsilon/r)^{1/\alpha}\right)$ and $\varepsilon_3 = \Theta\left( \varepsilon^{1/\alpha}/r^{1/\alpha-1} \right)$. Under these conditions, the number of queries to $U_\sigma$ is
    \[
        O\left( Q_1 Q_2 \left(\frac{\sqrt{B}}{\varepsilon_3} + \frac{1}{\sqrt{\varepsilon_3}}\right) \right) = \tilde O\left( \frac{r^{\frac{3-\alpha}{2\alpha} + \frac{1}{\alpha\{\beta\}}}}{\varepsilon^{\frac{3+\alpha}{2\alpha} + \frac{1}{\alpha\{\beta\}}}} \right),
    \]
    and the number of queries to $U_\rho$ is
    \[
        O\left( Q_2 \left(\frac{\sqrt{B}}{\varepsilon_3} + \frac{1}{\sqrt{\varepsilon_3}}\right) \right) = \tilde O\left(  \frac{r^{\frac{3-\alpha}{2\alpha}}}{\varepsilon^{\frac{3+\alpha}{2\alpha}}} \right),
    \]
    and the number of elementary quantum gates is
    \begin{align*}
        & O\left( Q_1 Q_2 \left(\frac{\sqrt{B}}{\varepsilon_3} + \frac{1}{\sqrt{\varepsilon_3}}\right) \right) \cdot \poly(n, n_\sigma, n_\rho) = \tilde O\left( \frac{r^{\frac{3-\alpha}{2\alpha} + \frac{1}{\alpha\{\beta\}}}}{\varepsilon^{\frac{3+\alpha}{2\alpha} + \frac{1}{\alpha\{\beta\}}}}  \poly(n) \right).
    \end{align*}

    \subsection{Lower Bounds and Hardness}
    
    Our quantum algorithms for both fidelity estimation and trace distance estimation requires time complexity polynomial in the rank $r$ of quantum states. Here, we show that unless $\BQP = \QSZK$, there is no quantum algorithm for both fidelity estimation and trace distance estimation with time complexity polylogarithmic in $r$.

    \begin{theorem} \label{thm:lower-bounds}
        If there is a quantum algorithm that computes fidelity or trace distance of quantum states of rank $\leq r$ within additive error $\varepsilon$ with time complexity $\poly\left(\log r, 1 / \varepsilon\right)$, then $\BQP = \QSZK$.
    \end{theorem}

    \begin{proof}
        Here, we recall a decision problem called $(\alpha, \beta)$-Quantum State Distinguishability ($(\alpha, \beta)$-QSD). Given $U_\rho$ and $U_\sigma$ that prepares the purifications of density operators $\rho$ and $\sigma$ and a promise that either $T(\rho, \sigma) \leq \alpha$ or $T(\rho, \sigma) \geq \beta$, the problem is to determine which is the case. It was shown in \cite{Wat02} that $(\alpha, \beta)$-QSD is $\QSZK$-complete if $0 \leq \alpha < \beta^2 \leq 1$.
        
        If there is a quantum algorithm for computing trace distance with time complexity $\poly\left(\log r, 1 / \varepsilon\right)$, then we can distinguish the two cases with time complexity $\poly(n)$ by letting $r = 2^n$ be the dimension of the two quantum states and $\varepsilon = (\beta - \alpha)/2 > 0$.
        
        If there is a quantum algorithm for computing fidelity with time complexity $\poly\left(\log r, 1 / \varepsilon\right)$, then we can distinguish the two cases with time complexity $\poly(n)$ by letting $r = 2^n$ be the dimension of the two quantum states and $\varepsilon = \left( (1-\alpha) - \sqrt{1 - \beta^2} \right)/2 > 0$. This is because $T(\rho, \sigma) \leq \alpha$ implies $F(\rho, \sigma) \geq 1 - \alpha$, and $T(\rho, \sigma) \geq \beta$ implies $F(\rho, \sigma) \leq \sqrt{1 - \beta^2}$. Then $(\alpha, \beta)$-QSD is reduced to distinguish which is the case with promise that either $F(\rho, \sigma) \leq \sqrt{1 - \beta^2}$ and $F(\rho, \sigma) \geq 1 - \alpha$. 
    \end{proof}
    
    Our quantum algorithms for estimating the fidelity and trace distance achieve a significant speedup under the low-rank assumption.
    One might wonder whether our algorithms can be ``dequantized'' through quantum-inspired low-rank techniques such as \cite{Tan19,CGL+20}. We suspect that it might be unachievable because the following theorem shows that computing fidelity and trace distance are $\mathsf{DQC1}$-hard. 
    
    \begin{theorem} \label{thm:dqc1-hard}
        Computing the fidelity and trace distance are $\mathsf{DQC1}$-hard, even for pure quantum states. 
    \end{theorem}
    
    \begin{proof}
        It was already proved in \cite{CPCC20} that estimating the fidelity is $\mathsf{DQC1}$-hard, even for pure quantum states. Here, we reduce the problem of estimating the fidelity to that of estimating the trace distance, and therefore show the $\mathsf{DQC1}$-hardness of estimating the trace distance.
        
        For any two pure quantum states $\psi = \ket{\psi}\bra{\psi}$ and $\phi = \ket{\phi} \bra{\phi}$, their trace distance is essentially
        \[
            T(\psi, \phi) = \sqrt{1 - (F(\psi, \phi))^2}.
        \]
        Therefore, any algorithm that computes the trace distance $T(\psi, \phi)$ will immediately yield the fidelity $F(\psi, \phi) = \sqrt{1 - (T(\psi, \phi))^2}$. As a result, estimating the trace distance is $\mathsf{DQC1}$-hard even for pure quantum states. 
    \end{proof}
    
    It was shown in \cite{FKM18} that $\mathsf{DQC1}$ is not (classically) weakly simulatable unless the polynomial hierarchy collapses to the second level, i.e., $\mathsf{PH} = \mathsf{AM}$. This, together with Theorem \ref{thm:dqc1-hard}, means that there is unlikely an efficient classical algorithm that estimates the fidelity or trace distance. It should be noted that this does not rule out the existence of a dequantized version of our quantum algorithms because dequantized algorithms often assume a different input model from Theorem \ref{thm:dqc1-hard}.
    More specifically, dequantized algorithms assume ``sampling and query access'' \cite{Tan19,CGL+20} to the input matrix (in our case, the density operator of the quantum state) stored in a pre-computed data structure.

\section*{Acknowledgment}

The authors would like to thank the anonymous reviewers for their valuable suggestions and for pointing out some mistakes in an earlier version of this paper. 
They also thank Wang Fang, Kean Chen and Minbo Gao for helpful discussions. 

Qisheng Wang was supported in part by the MEXT Quantum Leap Flagship Program (MEXT Q-LEAP) under Grant JPMXS0120319794.
Ji Guan was supported in part by the Youth Innovation Promotion Association, Chinese Academy of Sciences under Grant 2023116 and in part by the Key Research Program of the Chinese Academy of Sciences under Grant ZDRW-XX-2022-1. Zhicheng Zhang was supported by the
Sydney Quantum Academy, NSW, Australia.

\addcontentsline{toc}{section}{References}

\bibliographystyle{unsrt}
\bibliography{references}

\end{document}